\documentclass{article}
\usepackage{kr}

\usepackage{times}
\usepackage{soul}
\usepackage{url}
\usepackage[hidelinks]{hyperref}
\usepackage[utf8]{inputenc}
\usepackage[small]{caption}
\usepackage{graphicx}
\usepackage{amsmath}
\usepackage{amsthm}
\usepackage{booktabs}
\usepackage{algorithm}
\usepackage{algorithmic}
\usepackage{comment}
\usepackage{amssymb}

\def\NL{\text{\rm NL}}
\def\PTIME{\text{\rm PTIME}}

\def\NP{\text{\rm NP}}

\def\ACZ{\rm AC_0}

\def\assign{ \text{:--} }

\newcommand\sem[1]{{[\![ #1 ]\!]}}
\newcommand{\db}[1]{\mathsf{db}(#1)}
\newcommand{\depth}[1]{\mathsf{depth}(#1)}
\newcommand{\scount}[1]{\mathsf{scount}(#1)}

\newcommand{\rank}[2]{\mathsf{rank}_{#2}(#1)}

\newcommand{\mi}[1]{\mathit{#1}}
\newcommand{\ins}[1]{\mathbf{#1}}
\newcommand{\adom}[1]{\mathsf{dom}(#1)}
\newcommand{\ra}{\rightarrow}

\newcommand{\dep}{\Sigma}
\newcommand{\sch}[1]{\mathsf{sch}(#1)}
\newcommand{\esch}[1]{\mathsf{edb}(#1)}
\newcommand{\isch}[1]{\mathsf{idb}(#1)}

\newcommand{\body}[1]{\mathsf{body}(#1)}
\newcommand{\head}[1]{\mathsf{head}(#1)}

\newcommand{\class}[1]{\mathsf{#1}}

\newcommand{\tup}[1]{\langle #1 \rangle}

\newcommand{\base}[1]{\mathsf{base}(#1)}

\newcommand{\var}[1]{\mathsf{var}(#1)}
\newcommand{\vr}[1]{\langle #1 \rangle}

\newcommand{\lsign}[1]{\mathsf{b}(#1)}
\newcommand{\lvar}[1]{\mathsf{v}(#1)}

\newcommand{\support}[1]{\mathsf{support}(#1)}

\newcommand{\arity}[1]{\mathsf{ar}(#1)}

\newcommand{\lfacts}[1]{\mathsf{LFacts}(#1)}
\newcommand{\why}[3]{\mathsf{why}(#1,#2,#3)}
\newcommand{\nrwhy}[3]{\mathsf{why}_{\mathsf{NR}}(#1,#2,#3)}
\newcommand{\mdwhy}[3]{\mathsf{why}_{\mathsf{MD}}(#1,#2,#3)}
\newcommand{\mtd}[3]{\mathsf{min\text{-}tree\text{-}depth}(#1,#2,#3)}
\newcommand{\mgd}[3]{\mathsf{min\text{-}dag\text{-}depth}(#1,#2,#3)}
\newcommand{\unwhy}[3]{\mathsf{why}_{\mathsf{UN}}(#1,#2,#3)}

\newcommand{\OMIT}[1]{}

\newcommand{\can}[1]{\mathsf{can}(#1)}

\newcommand{\LDAT}{\class{LDat}}
\newcommand{\NRDAT}{\class{NRDat}}

\newcommand{\DAT}{\class{Dat}}

\def\eqtree{\approx}
\newcommand{\downof}[2]{#1_{\downarrow#2}}
\newcommand{\gri}[2]{\mathsf{gri}(#1,#2)}
\newcommand{\downc}[3]{\mathsf{down}(#1,#2,#3)}
\newcommand{\quot}[1]{#1_{/\eqtree}}

\newcommand{\cq}[1]{\mathsf{cq}(#1)}

\newcommand{\cqeq}[1]{\mathsf{cq}^{\approx}(#1)}

\def\qed{\hfill{\qedboxempty}      
  \ifdim\lastskip<\medskipamount \removelastskip\penalty55\medskip\fi}

\def\qedboxempty{\vbox{\hrule\hbox{\vrule\kern3pt
                 \vbox{\kern3pt\kern3pt}\kern3pt\vrule}\hrule}}

\def\qedfull{\hfill{\qedboxfull}   
  \ifdim\lastskip<\medskipamount \removelastskip\penalty55\medskip\fi}

\def\qedboxfull{\vrule height 4pt width 4pt depth 0pt}

\newcommand{\markfull}{\qedboxfull}

\newtheorem{manualtheoreminner}{Theorem}
\newenvironment{manualtheorem}[1]{%
  \renewcommand\themanualtheoreminner{#1}%
  \manualtheoreminner
}{\endmanualtheoreminner}

\newtheorem{manualpropositioninner}{Proposition}
\newenvironment{manualproposition}[1]{%
  \renewcommand\themanualpropositioninner{#1}%
  \manualpropositioninner
}{\endmanualpropositioninner}

\newtheorem{manuallemmainner}{Lemma}
\newenvironment{manuallemma}[1]{%
  \renewcommand\themanuallemmainner{#1}%
  \manuallemmainner
}{\endmanuallemmainner}

\newtheorem{example}{Example}
\newtheorem{theorem}{Theorem}

\newtheorem{proposition}[theorem]{Proposition}
\newtheorem{lemma}[theorem]{Lemma}
\newtheorem{definitionAux}[theorem]{Definition}
\newenvironment{definition}{\begin{definitionAux}
	}{\end{definitionAux}}

\title{The Complexity of Why-Provenance for Datalog Queries}

\author{
	Marco Calautti$^1$\and
	Ester Livshits$^2$\and
	Andreas Pieris$^{2,3}$\and
	Markus Schneider$^2$\\
	\affiliations
	$^1$Department of Computer Science, University of Milan\\
	$^2$School of Informatics, University of Edinburgh\\
	$^3$Department of Computer Science, University of Cyprus\\[1mm]
	\emails 
	marco.calautti{@}unimi.it, ester.livshits{@}ed.ac.uk,
	apieris{@}inf.ed.ac.uk, m.schneider{@}ed.ac.uk
}

\begin{document}

\maketitle

\begin{abstract}
Explaining why a database query result is obtained is an essential task towards the goal of Explainable AI, especially nowadays where expressive database query languages such as Datalog play a critical role in the development of ontology-based applications. A standard way of explaining a query result is the so-called why-provenance, which essentially provides information about the witnesses to a query result in the form of subsets of the input database that are sufficient to derive that result. To our surprise, despite the fact that the notion of why-provenance for Datalog queries has been around for decades and intensively studied, its computational complexity remains unexplored. The goal of this work is to fill this apparent gap in the why-provenance literature. Towards this end, we pinpoint the data complexity of why-provenance for Datalog queries and key subclasses thereof. The takeaway of our work is that why-provenance for recursive queries, even if the recursion is limited to be linear, is an intractable problem, whereas for non-recursive queries is highly tractable. Having said that, we experimentally confirm, by exploiting SAT solvers, that making why-provenance for (recursive) Datalog queries work in practice is not an unrealistic goal.
\end{abstract}

\section{Introduction}\label{sec:introduction}

Datalog has emerged in the 1980s as a logic-based query language from Logic Programming and has been extensively studied since then~\cite{AbHV95}. The name Datalog reflects the intention of devising a counterpart of Prolog for data processing. It essentially extends the language of unions of conjunctive queries, which corresponds to the select-project-join-union fragment of relational algebra, with the important feature of recursion, much needed to express some natural queries.
Among numerous applications, Datalog has been heavily used in the context of ontological query answering. In particular, for several important ontology languages based on description logics and existential rules, ontological query answering can be reduced to the problem of evaluating a Datalog query (see, e.g.,~\cite{EOSTX12,BBGKM22}), which in turn enables the exploitation of efficient Datalog engines such as DLV~\cite{LPFEGPS06} and Clingo~\cite{GKKOSW16}.

As for any other query language, explaining why a result to a Datalog query is obtained is crucial towards explainable and transparent data-intensive applications. A standard way for providing such explanations to query answers is the so-called {\em why-provenance}~\cite{BuKT01}.
Its essence is to collect all the subsets of the input database that are sufficient to derive a certain answer. More precisely, in the case of Datalog queries, the why-provenance of an answer tuple $\bar t$ is obtained by considering all the possible proof trees $T$ of the fact ${\rm Ans}(\bar t)$, with ${\rm Ans}$ being the answer predicate of the Datalog query in question, and then collecting all the database facts that label the leaves of $T$. Recall that a proof tree of a fact $\alpha$ w.r.t.~a database $D$ and a set $\dep$ of Datalog rules forms a tree-like representation of a way for deriving $\alpha$ by starting from $D$ and executing the rules occurring in $\dep$~\cite{AbHV95}.

There are recent works that studied the concept of why-provenance for Datalog queries. In particular, there are theoretical studies on computing the why-provenance~\cite{DaAA13,DMRT14}, attempts to under-approximate the why-provenance towards an efficient computation~\cite{ZhSS20}, studies on the restricted setting of non-recursive Datalog queries~\cite{LeLG19}, attempts to compute the why-provenace by transforming the grounded Datalog rules to a system of equations~\cite{EsLS14}, and attempts to compute the why-provenance on demand via transformations to existential rules~\cite{ElKM22}.

Despite the above research activity on the concept of why-provenance for Datalog queries, to our surprise, there is still a fundamental question that remains unexplored: 

\medskip

\noindent {\em \textbf{Main Research Question:} What is the exact computational complexity of why-provenance for Datalog queries?}

\medskip

\noindent The goal of this work is to provide an answer to the above question. To this end, for a Datalog query $Q$, we study the complexity of the following algorithmic problem, dubbed $\mathsf{Why\text{-}Provenance}[Q]$: given a database $D$, an answer $\bar t$ to $Q$ over $D$, and a subset $D'$ of $D$, is it the case that $D'$ belongs to the why-provenance of $\bar t$ w.r.t.~$D$ and $Q$? Pinpointing the complexity of the above decision problem will let us understand the inherent complexity of why-provenance for Datalog queries w.r.t.~the size of the database, which is precisely what matters when using why-provenance in practice.

\medskip
\noindent \textbf{Our Contribution.} The takeaway of our complexity analysis is that explaining Datalog queries via why-povenance is, in general, an intractable problem. In particular, for a Datalog query $Q$, we show that $\mathsf{Why\text{-}Provenance}[Q]$ is in \NP, and there are queries for which it is \NP-hard.
We further analyze the complexity of the problem when $Q$ is linear (i.e., the recursion is restricted to be linear) or non-recursive, with the aim of clarifying whether the feature of recursion affects the inherent complexity of why-provenance.
We show that restricting the recursion to be linear does not affect the complexity, namely the problem is in \NP~and for some queries it is even \NP-hard. However, completely removing the recursion significantly reduces the complexity; in particular, we prove that the problem is in $\ACZ$.

It is clear that the notion of why-provenance for Datalog queries, and hence the problem $\mathsf{Why\text{-}Provenance}[Q]$, heavily rely on the notion of proof tree. However, as already discussed in the literature (see, e.g., the recent work~\cite{BBPT22}), there are proof trees that are counterintuitive since they represent unnatural derivations (e.g., a fact is used to derive itself, or a fact is derived in several different ways), and this also affects the why-provenance.
With the aim of overcoming this conceptual limitation of proof trees, we propose the class of unambiguous proof trees. All occurrences of a fact in such a proof tree must be proved via the same derivation. We then study the problem $\mathsf{Why\text{-}Provenance}[Q]$ focusing on unambiguous proof trees, and show that its complexity remains the same.
This should be perceived as a positive outcome as we can overcome the limitation of arbitrary proof trees without increasing the complexity.

We finally verify that unambiguous proof trees, apart from their conceptual advantage, also help to exploit off-the-shelf SAT solvers towards an efficient computation of the why-provenance for Datalog queries. In particular, we discuss a proof-of-concept implementation that exploits the state-of-the-art SAT solver Glucose (see, e.g.,~\cite{Audemard18}), and present encouraging results based on queries and databases that are coming from the Datalog literature.

\medskip

\noindent {\em An extended version with further details, as well as the experimental scenarios and the source code, can be found at https://gitlab.com/mcalautti/datalog-why-provenance.}

\section{Preliminaries}\label{sec:preliminaries}

We consider the disjoint countably infinite sets $\ins{C}$ and $\ins{V}$ of {\em constants} and {\em variables}, respectively. We may refer to constants and variables as {\em terms}. For brevity, given an integer $n > 0$, we may write $[n]$ for the set of integers $\{1,\ldots,n\}$.

\medskip

\noindent 
\textbf{Relational Databases.} A {\em schema} $\ins{S}$ is a finite set of relation names (or predicates) with associated arity. We write $R/n$ to say that $R$ has arity $n \ge 0$; we may also write $\arity{R}$ for $n$.
A {\em (relational) atom} $\alpha$ over $\ins{S}$ is an expression of the form $R(\bar t)$, where $R/n \in \ins{S}$ and $\bar t$ is an $n$-tuple of terms. By abuse of notation, we may treat tuples as the \emph{set} of their elements. A {\em fact} is an atom that mentions only constants.
%
%
A {\em database} over $\ins{S}$ is a finite set of facts over $\ins{S}$. The {\em active domain} of a database $D$, denoted $\adom{D}$, is the set of constants in $D$.

\medskip
\noindent
\textbf{Syntax and Semantics of Datalog Programs.} A {\em (Datalog) rule} $\sigma$ over a schema $\ins{S}$ is an expression of the form
\[
R_0(\bar x_0)\ \assign\ R_1(\bar x_1),\ldots,R_n(\bar x_n)
\]
for $n \geq 1$, where $R_i(\bar x_i)$ is a (constant-free) relational atom over $\ins{S}$ for $i \in \{0,\ldots,n\}$, and each variable in $\bar x_0$ occurs in $\bar x_k$ for some $k \in [n]$.
We refer to $R_0(\bar x_0)$ as the {\em head} of $\sigma$, denoted $\head{\sigma}$, and to the expression that appears on the right
of the\assign symbol as the {\em body} of $\sigma$, denoted $\body{\sigma}$, which we may treat as the set of its atoms.

A {\em Datalog program} over a schema $\ins{S}$ is defined as a finite set $\dep$ of Datalog rules over $\ins{S}$. 
%
A predicate $R$ occurring in $\dep$ is called {\em extensional} if there is no rule in $\dep$ having $R$ in its head, and {\em intentional} if there exists at least one rule in $\dep$ with $R$ in its head.
The {\em extensional (database) schema} of $\dep$, denoted $\esch{\dep}$, is the set of all extensional predicates in $\dep$, while the {\em intentional
schema} of $\dep$, denoted $\isch{\dep}$, is the set of all intensional predicates in $\dep$. Note that, by definition, $\esch{\dep} \cap \isch{\dep} = \emptyset$. The {\em schema} of $\dep$, denoted $\sch{\dep}$, is the set $\esch{\dep} \cup \isch{\dep}$, which is in general a subset of $\ins{S}$ since some predicates of $\ins{S}$ may not appear in $\dep$.

There are interesting fragments of Datalog programs that somehow limit the recursion and have been extensively studied in the literature.
A Datalog program $\dep$ is called {\em linear} if, for each rule $\sigma \in \dep$, there exists at most one atom in $\body{\sigma}$ over $\isch{\dep}$, namely $\body{\sigma}$ mentions at most one intensional predicate. Roughly, linear Datalog programs can have only linear recursion.
Another key fragment is the one that completely forbids recursion. A Datalog program $\dep$ is called {\em non-recursive} if its predicate graph, which encodes how the predicates of $\sch{\dep}$ depend on each other, is acyclic. Recall that the nodes of the predicate graph of  $\dep$ are the predicates of $\sch{\dep}$, and there is an edge from $R$ to $P$ if there is a rule of the form $P(\bar x)\ \assign\ \ldots, R(\bar y),\ldots$ in $\dep$.

An elegant property of Datalog programs is that they have three equivalent semantics: model-theoretic, fixpoint, and proof-theoretic~\cite{AbHV95}.
We proceed to recall the proof-theoretic semantics of Datalog programs since it is closer to the notion of why-provenance. To this end, we need the key notion of proof tree of a fact, which will anyway play a crucial role in our work.
For a database $D$ and a Datalog program $\dep$, let $\base{D,\dep} = \{R(\bar t) \mid R \in \sch{\dep} \text{ and } \bar t \in \adom{D}^{\arity{R}}\}$, the set of all facts that can be formed using predicates of $\sch{\dep}$ and terms of $\adom{D}$.

\begin{definition}[\textbf{Proof Tree}]\label{def:proof-tree}
	Consider a Datalog program $\dep$, a database $D$ over $\esch{\dep}$, and a fact $\alpha$ over $\sch{\dep}$. 
	A {\em proof tree of $\alpha$ w.r.t.~$D$ and $\dep$} is a finite labeled rooted tree $T=(V,E,\lambda)$, with $\lambda : V \ra \base{D,\dep}$, such that:
	\begin{enumerate}
		\item If $v \in V$ is the root, then $\lambda(v) = \alpha$.
		
		\item If $v \in V$ is a leaf, then $\lambda(v) \in D$.
		
		\item If $v \in V$ is a node with $n\ge1$ children $u_1,\ldots,u_n$, then there is a rule $R_0(\bar x_0)\ \assign\ R_1(\bar x_1),\ldots,R_n(\bar x_n) \in \dep$ and a function $h : \bigcup_{i \in [n]} \bar x_i \ra \ins{C}$ such that $\lambda(v) = R_0(h(\bar x_0))$, and $\lambda(u_i) = R_i(h(\bar x_i))$ for each $i \in [n]$. \hfill\markfull
	\end{enumerate} 
\end{definition}

Essentially, a proof tree of a fact $\alpha$ w.r.t.~$D$  and $\dep$ indicates that we can prove $\alpha$ using $D$ and $\dep$, that is, we can derive $\alpha$ starting from $D$ end executing the rules of $\dep$. An example, which will also serve as a running example throughout the paper, that illustrates the notion of proof tree follows.

\begin{example}\label{exa:proof-tree}
	Consider the Datalog program $\dep$ consisting of
	\begin{eqnarray*}
		A(x) &\emph{\assign}& S(x)\\
		A(x) &\emph{\assign}& A(y), A(z),T(y,z,x)
	\end{eqnarray*}
	that encodes the {\em path accessibility problem}~\cite{Cook74}. The predicate $S$ represents source nodes, $A$ represents nodes that are accessible from the source nodes, and $T$ represents accessibility conditions, that is, $T(y,z,x)$ means that if both $y$ and $z$ are accessible from the source nodes, then so is $x$.
	We further consider the database
	\[
	D\ =\ \{S(a), T(a,a,b),T(a,a,c), T(a,a,d), T(b,c,a)\}.
	\]
	A simple proof tree of the fact $A(d)$ w.r.t.~$D$ and $\dep$ follows:
	
	\centerline{\includegraphics[width=.27\textwidth]{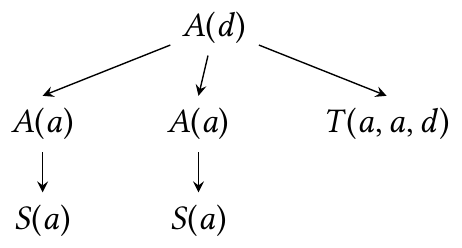}}
	
	\noindent The following is another, slightly more complex, proof tree of the fact $A(d)$ w.r.t.~$D$ and $\dep$:
	
	\centerline{
	\includegraphics[width=.48\textwidth]{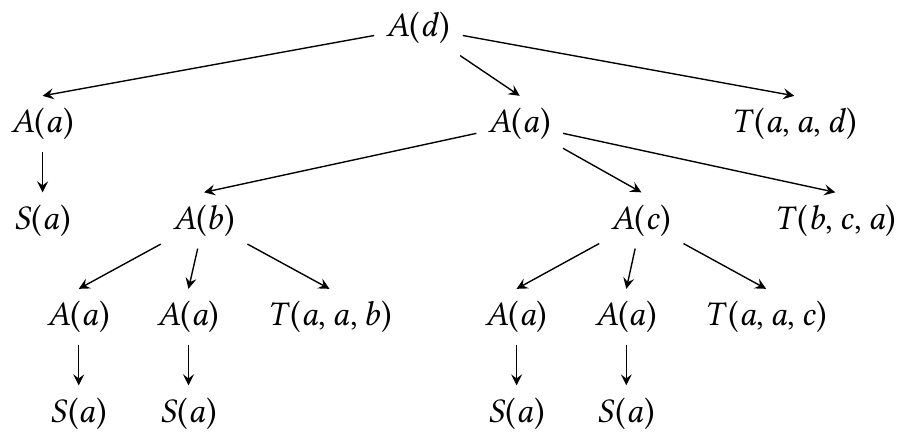}}

	\noindent Note that the above are only two out of the many proof trees of $A(d)$ w.r.t.~$D$ and $\dep$. In fact, there exist infinitely many as one can build larger and larger such proof trees: whenever we encounter a node labeled by $A(a)$, we can choose to apply the recursive rule instead of the rule $A(x)\ \emph{\assign}\ S(x)$. \hfill\markfull
\end{example}

Now, given a Datalog program $\dep$ and a database $D$ over $\sch{\dep}$, the {\em semantics of $\dep$ on $D$}, denoted $\dep(D)$, is the set
\[
\dep(D)\ =\ \{\alpha \mid \text{ there is a proof tree of } \alpha \text{ w.r.t. } D \text{ and } \dep\},
\]
that is, the set of facts that can be proven using $D$ and $\dep$.

\medskip
\noindent
\textbf{Datalog Queries.} Having the syntax and the semantics of Datalog programs in place, it is now straightforward to recall the syntax and the semantics of Datalog queries.
A {\em Datalog query} is a pair $Q = (\dep,R)$, where $\dep$ is a Datalog program and $R$ a predicate of $\isch{\dep}$.
We further call $Q$ {\em linear} (resp., {\em non-recursive}) if the program $\dep$ is linear (resp., non-recursive).
Now, for a database $D$ over $\esch{\dep}$, the {\em answer} to $Q$ over $D$ is defined as the set of tuples
\[
Q(D)\ =\ \{\bar t \in \adom{D}^{\arity{R}} \mid R(\bar t) \in \dep(D)\},
\]
i.e., the tuples $\bar t$ such that the fact $R(\bar t)$ can be proven using $D$ and $\dep$. The class that collects all the Datalog queries is denoted $\DAT$. We also write $\LDAT$ and $\NRDAT$ for the classes of linear and non-recursive Datalog queries, respectively.

\section{Why-Provenance for Datalog Queries}\label{sec:why-provenance}

As already discussed in the Introduction, why-provenance is a standard way of explaining why a query result is obtained. It essentially collects all the subsets of the database (without unnecessary atoms) that allow us to prove (or derive) a query result. We proceed to formalize this simple idea, and then introduce the main problem of interest.

Given a proof tree $T = (V,E,\lambda)$ (of some fact w.r.t.~some database and Datalog program), the {\em support} of $T$ is the set
\[
\support{T}\ =\ \left\{\lambda(v) \mid v \in V \text{ is a leaf of } T\right\},
\]
which is essentially the set of facts that label the leaves of the proof tree $T$. Note that $\support{T}$ is a subset of the underlying database since, by definition, the leaves of a proof tree are labeled with database atoms.
The formal definition of why-provenance for Datalog queries follows.

\begin{definition}[\textbf{Why-Provenance for Datalog}]\label{def:why-provenance}
	Consider a Datalog query  $Q = (\dep,R)$, a database $D$ over $\esch{\dep}$, and a tuple $\bar t \in \adom{D}^{\arity{R}}$. The {\em why-provenance of $\bar t$ w.r.t.~$D$ and $Q$} is defined as the family of sets of facts 
	\[
	\{\support{T} \mid T \text{ is a proof tree of } R(\bar t) \text{ w.r.t. } D \text{ and } \dep\}
	\]
	which we denote by $\why{\bar t}{D}{Q}$. \hfill\markfull
\end{definition}

Intuitively speaking, a set of facts $D' \subseteq D$ that belongs to $\why{\bar t}{D}{Q}$ should be understood as a ``real'' reason why the tuple $\bar t$ is an answer to the query $Q$ over the database $D$, i.e., $D'$ explains why $\bar t \in Q(D)$. By ``real'' we mean that all the facts of $D'$ are really used in order to derive the tuple $\bar t$ as an answer. Here is a simple example of why-provenance.

\begin{example}\label{exa:why-provenance}
	Let $Q = (\dep,A)$, where $\dep$ is the program that encodes the path accessibility problem as in Example~\ref{exa:proof-tree}, and let $D$ be the database from Example~\ref{exa:proof-tree}. 
	It can be verified that the why-provenance of the unary tuple $(d)$ w.r.t.~$D$ and $Q$ consists of $\{S(a), T(a,a,d)\}$ and the database $D$ itself. The former set is actually the support of the first proof tree given in Example~\ref{exa:proof-tree}, while $D$ is the support of the second proof tree. 
	Recall that $A(d)$ has infinitely many proof trees w.r.t.~$D$ and $\dep$, whereas $\why{(d)}{D}{Q}$ contains only two sets. Thus, in general, there is no 1-1 correspondence between proof trees of a fact $R(\bar t)$ and members of the why-provenance of $\bar t$. \hfill\markfull
\end{example}

We would like to pinpoint the inherent complexity of the problem of computing the why-provenance of a tuple w.r.t.~a database and a Datalog query.
To this end, we need to study the complexity of  recognizing whether a certain subset of the database belongs to the why-provenance, that is, whether a candidate explanation is indeed an explanation.
This leads to the following algorithmic problem parameterized by a class $\class{C}$ of Datalog queries; $\class{C}$ can be, e.g., $\DAT$, $\LDAT$, or $\NRDAT$:

\smallskip

\begin{center}
	\fbox{\begin{tabular}{ll}
			{\small PROBLEM} : & $\mathsf{Why\text {-}Provenance[C]}$
			\\[2mm]
			{\small INPUT} : & A Datalog query $Q = (\dep,R)$ from $\class{C}$,\\
			& a database $D$ over $\esch{\dep}$,\\
			& a tuple ${\bar t} \in \adom{D}^{\arity{R}}$, and $D' \subseteq D$.\\[2mm]
			{\small QUESTION} : &  Does $D' \in \why{\bar t}{D}{Q}$?
	\end{tabular}}
\end{center}

\smallskip

Our goal is to study the above problem and pinpoint its complexity. We are actually interested in the {\em data complexity} of $\mathsf{Why\text {-}Provenance[C]}$, where the query $Q$ is fixed, and only the database $D$, the tuple $\bar t$, and $D'$ are part of the input, i.e., for each $Q = (\dep,R)$ from $\class{C}$, we consider the problem:

\smallskip

\begin{center}
	\fbox{\begin{tabular}{ll}
			{\small PROBLEM} : & $\mathsf{Why\text {-}Provenance}[Q]$
			\\[2mm]
			{\small INPUT} : & A database $D$ over $\esch{\dep}$,\\
			& a tuple ${\bar t} \in \adom{D}^{\arity{R}}$, and $D' \subseteq D$.\\[2mm]
			{\small QUESTION} : &  Does $D' \in \why{\bar t}{D}{Q}$?
	\end{tabular}}
\end{center}

\smallskip

%
\noindent By the typical convention, the problem $\mathsf{Why\text {-}Provenance[C]}$ is in a certain complexity class $C$ in data complexity if, for every query $Q$ from $\class{C}$, $\mathsf{Why\text {-}Provenance}[Q]$ is in $C$.
On the other hand, $\mathsf{Why\text {-}Provenance[C]}$ is hard for a certain complexity class $C$ in data complexity if there exists a query $Q$ from $\class{C}$ such that $\mathsf{Why\text {-}Provenance}[Q]$ is hard for $C$.


\section{Data Complexity of Why-Provenance}\label{sec:complexity}

The goal of this section is to pinpoint the data complexity of $\mathsf{Why\text {-}Provenance[C]}$, for each $\class{C} \in \{\DAT,\LDAT,\NRDAT\}$. As we shall see, the main outcome of our analysis is that  for recursive queries, even if the recursion is linear, the problem is in general intractable, whereas for non-recursive queries it is highly tractable.
We first focus on recursive queries.

\subsection{Recursive Queries}\label{sec:recursive-complexity}

We show the following complexity result:

\def\therecursivecomplexity{
$\mathsf{Why\text {-}Provenance[C]}$ is \NP-complete in data complexity, for each class $\class{C} \in \{\DAT,\LDAT\}$.
}

\begin{theorem}\label{the:recursive-complexity}
\therecursivecomplexity
\end{theorem}

Note that there is a striking difference between the problem of why-provenance and the problem of query evaluation, which is known to be in \PTIME~in data complexity; in fact, for linear Datalog queries it is in \NL~\cite{DEGV01}.
%
To prove Theorem~\ref{the:recursive-complexity}, it suffices to show that:
\begin{itemize}
	\item $\mathsf{Why\text {-}Provenance[\DAT]}$ is in \NP~in data complexity.
	\item $\mathsf{Why\text {-}Provenance[\LDAT]}$ is \NP-hard in data complexity.
\end{itemize}
The lower bound is established via a reduction from $\mathsf{3SAT}$. We actually devise a linear Datalog query $Q$, and provide a reduction from $\mathsf{3SAT}$ to $\mathsf{Why\text {-}Provenance}[Q]$.
%
Let us now discuss the key ingredients underlying the upper bound.
%
%
%
%
The central property is that whenever there is a proof tree $T$ that witnesses the fact that the given subset of the input database belongs to the why-provenance, then there is always a way to compactly represent $T$ as a polynomially-sized directed acyclic graph. This in turn leads to an easy guess-and-check algorithm that runs in polynomial time. We proceed to give further details for the above crucial property.

\medskip
\noindent
\textbf{Proof DAG.} 
We first introduce the notion of proof directed acyclic graph (DAG) of a fact, which is essentially a generalization of the notion of proof tree. Recall that a DAG $G$ is {\em rooted} if it has exactly one node, the {\em root}, with no incoming edges. A node of $G$ is a {\em leaf} if it has no outgoing edges.

\begin{definition}[\textbf{Proof DAG}]\label{def:proof-dag}
	Consider a Datalog program $\dep$, a database $D$ over $\esch{\dep}$, and a fact $\alpha$ over $\sch{\dep}$. 
	A {\em proof DAG of $\alpha$ w.r.t.~$D$ and $\dep$} is a finite labeled rooted DAG $G=(V,E,\lambda)$, with $\lambda : V \ra \base{D,\dep}$, such that:
	\begin{enumerate}
		\item If $v \in V$ is the root, then $\lambda(v) = \alpha$.
		
		\item If $v \in V$ is a leaf, then $\lambda(v) \in D$.
		
		\item If $v \in V$ has $n \geq 1$ outgoing edges $(v,u_1),\ldots,(v,u_n)$, then there is a rule $R_0(\bar x_0)\ \assign\ R_1(\bar x_1),\ldots,R_n(\bar x_n) \in \dep$ and a function $h : \bigcup_{i \in [n]} \bar x_i \ra \ins{C}$ such that $\lambda(v) = R_0(h(\bar x_0))$, and $\lambda(u_i) = R_i(h(\bar x_i))$ for $i \in [n]$. \hfill\markfull
	\end{enumerate} 
\end{definition}

The key difference between a proof tree and a proof DAG is that a proof DAG might reuse nodes to compactly represent a proof tree. This is shown by the following example.


\begin{example}\label{exa:proof-dag}
	Let $Q = (\dep,A)$, where $\dep$ is the program given in Example~\ref{exa:proof-tree}, and let $D$ be the database from Example~\ref{exa:proof-tree}. 
	A simple proof DAG of the fact $A(d)$ w.r.t.~$D$ and $\dep$ is
	
	\centerline{\includegraphics[width=.25\textwidth]{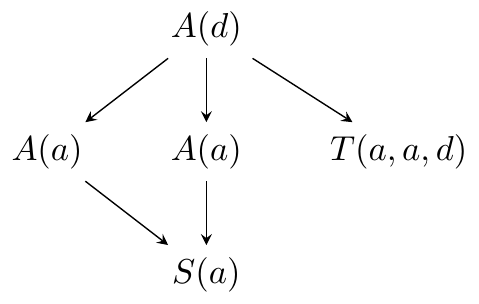}}
	
	\noindent which compactly represents the first proof tree given in Example~\ref{exa:proof-tree}.
	The following is another, slightly more complex, proof DAG of the fact $A(d)$ w.r.t.~$D$ and $\dep$:
	
	\centerline{
		\includegraphics[width=.25\textwidth]{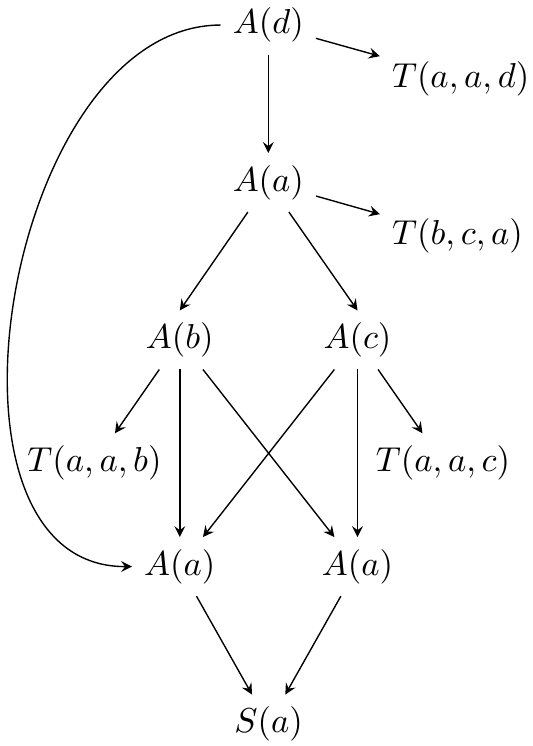}}
	
	\noindent It clearly represents the second proof tree from Example~\ref{exa:proof-tree}. \hfill\markfull
\end{example}

\medskip
\noindent
\textbf{Compact Representation of Proof Trees.} 
Given a proof DAG $G$ (of some fact w.r.t.~some database and Datalog program), we define its {\em support}, denoted $\support{G}$, as the set of facts that label the leaves of $G$. The key result follows:

\def\procharacterizationalltrees{
For a Datalog program $\dep$, there is a polynomial $f$ such that, for every database $D$ over $\esch{\dep}$, fact $\alpha$ over $\sch{\dep}$, and $D' \subseteq D$, the following are equivalent:
	\begin{enumerate}
		\item There exists a proof tree $T$ of $\alpha$ w.r.t.~$D$ and $\dep$ such that $\support{T} = D'$.
		\item There exists a proof DAG $G = (V,E,\lambda)$ of $\alpha$ w.r.t.~$D$ and $\dep$ such that $\support{G} = D'$ and $|V| \leq f(|D|)$.
	\end{enumerate}
}

\begin{proposition}\label{pro:characterization-all-trees}
	\procharacterizationalltrees
\end{proposition}

It is easy to show that $(2)$ implies $(1)$ by ``unravelling'' the proof DAG $G$ of $\alpha$ w.r.t.~$D$ and $\dep$ into a proof tree $T$ of $\alpha$ w.r.t.~$D$ and $\dep$ with $\support{T} = \support{G}$. 
%
Now, the direction $(1)$ implies $(2)$ is rather non-trivial and requires a careful construction that converts a proof tree $T$ of $\alpha$ w.r.t.~$D$ and $\dep$ into a compact proof DAG $G$ of $\alpha$ w.r.t.~$D$ and $\dep$ such that $\support{T} = \support{G}$. This construction proceeds in three main steps captured by Lemmas~\ref{lem:depth-reduction},~\ref{lem:scount-reduction}, and~\ref{lem:from-trees-to-dags}.


\medskip 

$\bullet$ The \textbf{\textit{first step}} is to show that a proof tree $T$ of $\alpha$ w.r.t.~$D$ and $\dep$ with $\support{T} = D'$ can be converted into a proof tree $T'$ of $\alpha$ w.r.t.~$D$ and $\dep$ with $\support{T'} = D'$ that has ``small'' depth. Let us recall that the {\em depth} of a rooted tree $T$, denoted $\depth{T}$, is the length of the longest path from its root to a leaf node. The corresponding lemma follows:

\def\lemmadepthreduction{
	For each Datalog program $\dep$, there is a polynomial $f$ such that, for every database $D$ over $\esch{\dep}$, fact $\alpha$ over $\sch{\dep}$, and $D' \subseteq D$, if there exists a proof tree $T$ of $\alpha$ w.r.t.~$D$ and $\dep$ with $\support{T} = D'$, then there exists also such a proof tree $T'$ with $\depth{T'} \leq f(|D|)$.
}

\begin{lemma}\label{lem:depth-reduction}
\lemmadepthreduction
\end{lemma}

\smallskip 

$\bullet$ The \textbf{\textit{second step}} consists of proving that a proof tree $T$ of $\alpha$ w.r.t.~$D$ and $\dep$ with $\support{T} = D'$ of ``small'' depth can be converted into a proof tree $T'$ of $\alpha$ w.r.t.~$D$ and $\dep$ with $\support{T'} = D'$ of ``small'' subtree count.
Roughly speaking, the subtree count of a proof tree $T$ is the maximum number of different (w.r.t.~node-labels) subtrees of $T$ rooted at nodes with the same label. Let us formalize this notion.

Two rooted trees $T=(V,E,\lambda)$ and $T' = (V',E',\lambda')$ are {\em isomorphic}, denoted $T \eqtree T'$, if there is a bijection $h : V \ra V'$ such that, for each node $v \in V$, $\lambda(v) = \lambda'(h(v))$, and for each two nodes $u,v \in V$, $(u,v) \in E$ iff $(h(u),h(v)) \in E'$. It is clear that $\eqtree$ is an equivalence relation over the set of all rooted trees.
We further write $T[\alpha]$, for a fact $\alpha$, to denote the set of all subtrees of $T$ whose root is labeled with $\alpha$, i.e., $T[\alpha] = \{T[v] \mid v \in V \text{ and } \lambda(v) = \alpha\}$ with $T[v]$ being the subtree of $T$ rooted at $v$.
Let $\quot{T[\alpha]}$ be the quotient set of $T[\alpha]$ w.r.t.~$\eqtree$, i.e., the set of all equivalence classes of $T[\alpha]$ w.r.t.~$\eqtree$. In other words, each member of $\quot{T[\alpha]}$ is a maximal set of trees of $T[\alpha]$ that are labeled in exactly the same way.
Then, the {\em subtree count} of $T$, denoted $\scount{T}$, is $\max_{\alpha \in \{\lambda(v) \mid v \in V\}} \{|\quot{T[\alpha]}|\}$.
The key lemma follows:

\def\lemmascountreduction{
    For each Datalog program $\dep$ and a polynomial $f$, there is a polynomial $g$ such that, for every database $D$ over $\esch{\dep}$, fact $\alpha$ over $\sch{\dep}$, and $D' \subseteq D$, if there exists a proof tree $T$ of $\alpha$ w.r.t.~$D$ and $\dep$ such that $\support{T} = D'$ and $\depth{T} \leq f(|D|)$, then there exists also such a proof tree $T'$ with $\scount{T'} \leq g(|D|)$.
}

\begin{lemma}\label{lem:scount-reduction}
	\lemmascountreduction
\end{lemma}

\smallskip

$\bullet$ The \textbf{\textit{third step}} shows that a proof tree $T$ of $\alpha$ w.r.t.~$D$ and $\dep$ with $\support{T} = D'$ of ``small'' subtree count can be converted into a compact proof DAG $G$ of $\alpha$ w.r.t.~$D$ and $\dep$ with $\support{G} = D'$. Here is the corresponding lemma:

\def\lemmafromtreestodags{
	For each Datalog program $\dep$ and a polynomial $f$, there is a polynomial $g$ such that, for every database $D$ over $\esch{\dep}$, fact $\alpha$, and $D' \subseteq D$, if there is a proof tree $T$ of $\alpha$ w.r.t.~$D$ and $\dep$ with $\support{T} = D'$ and $\scount{T} \leq f(|D|)$, then there exists a proof DAG $G = (V,E,\lambda)$ of $\alpha$ w.r.t.~$D$ and $\dep$ with $\support{G} = D'$ and $|V| \leq g(|D|)$.
}

\begin{lemma}\label{lem:from-trees-to-dags}
\lemmafromtreestodags
\end{lemma}

\medskip

It is now clear that the direction (1) implies (2) of Proposition~\ref{pro:characterization-all-trees} is an immediate consequence of Lemmas~\ref{lem:depth-reduction},~\ref{lem:scount-reduction} and~\ref{lem:from-trees-to-dags}.

\subsection{Non-Recursive Queries}

We now focus on non-recursive Datalog queries, and show the following about the data complexity of why-provenance:

\def\thenonrecursivecomplexity{
	$\mathsf{Why\text {-}Provenance[\NRDAT]}$ is in $\ACZ$ in data complexity.
}

\begin{theorem}\label{the:non-recursive-complexity}
\thenonrecursivecomplexity
\end{theorem}

The above result is shown via {\em first-order rewritability}, i.e., given a non-recursive Datalog query $Q = (\dep,R)$, we construct a first-order query $Q_{\mi{FO}}$ such that, for every input instance of $\mathsf{Why\text {-}Provenance}[Q]$, namely a database $D$ over $\esch{\dep}$, a tuple $\bar t \in \adom{D}^{\arity{R}}$, and a subset $D'$ of $D$, the fact that $D'$ belongs to $\why{\bar t}{D}{Q}$ is equivalent to the fact that $\bar t$ is an answer to the query $Q_{\mi{FO}}$ over $D'$. Since first-order query evaluation is in $\ACZ$ in data complexity~\cite{Vardi95}, Theorem~\ref{the:non-recursive-complexity} follows.
Before delving into the details, let us first recall the basics about first-order queries.

\medskip
\noindent
\textbf{First-Order Queries.}
A {\em first-order (FO)} query $Q$ is an expression of the form $\varphi(\bar x)$, where $\varphi$ is an FO formula, $\bar x$ is a tuple of (not necessarily distinct) variables, and the set of variables occurring in $\bar x$ is precisely the set of free variables of $\varphi$.
The {\em answer} to $Q$ over a database $D$ is the set of tuples
$
Q(D) = \{\bar t \in \adom{D}^{|\bar x|} \mid D \models \varphi[\bar x/\bar t]\}, 
$
where $|\bar x|$ denotes the length of $\bar x$, $\varphi[\bar x/\bar t]$ is the sentence obtained after replacing the variables of $\bar x$ with the corresponding constants of $\bar t$, and $\models$ denotes the standard FO entailment. Let $\var{\varphi}$ be the set of variables occurring in $\varphi$.
A \emph{conjunctive query (CQ)} is an FO query $\varphi(\bar x)$, where $\varphi$ is of the form $\exists \bar y \, (R_1(\bar x_1) \wedge \cdots \wedge R_n(\bar x_n))$ with $\bar x \cap \bar y = \emptyset$ and $\bar x_i \subseteq \bar x \cup \bar y$.

\medskip
\noindent
\textbf{Some Preparation.} Towards the construction of the desired first-order query, we need some auxiliary notions. The {\em canonical form} of a fact $\alpha$, denoted $\can{\alpha}$, is the atom obtained by replacing each constant $c$ in $\alpha$ with a variable $\vr{c}$, i.e., the name of the variable is uniquely determined by the constant $c$.
Given a Datalog query $Q = (\dep,R)$, we say that a labeled rooted tree $T = (V,E,\lambda)$ is a {\em $Q$-tree} if it is the proof tree of some fact $R(\bar t)$ w.r.t.~some database $D$ over $\esch{\dep}$ and $\dep$. The notion of the induced CQ by a $Q$-tree follows:

\begin{definition}[\textbf{Induced CQ}]\label{def:induced-cq}
	Consider a Datalog query $Q = (\dep,R)$ and a $Q$-tree $T  = (V,E,\lambda)$, where $v \in V$ is the root node and $\lambda(v) = R(c_1,\ldots,c_n)$. The {\em CQ induced by $T$}, denoted $\cq{T}$, is the CQ $\varphi_{T}(\vr{c_1},\ldots,\vr{c_n})$ with 
	\[
	\varphi_T\ =\ \exists \bar x \left( \bigwedge_{\alpha \in \support{T}} \can{\alpha}\right),
	\]
	where $\bar x$ consists of all $\vr{c}$ for $c \in \adom{\support{T}} \setminus \{c_1,\ldots,c_n\}$. We let $\cq{Q} = \{\cq{T} \mid T \text{ is a $Q$-tree}\}$. \hfill\markfull
\end{definition}

In simple words, $\cq{T}$ is the CQ obtained by taking the conjunction of the facts that label the leaves of $T$ in canonical form, and then existentially quantify all the variables apart from those occurring in the canonical form of the fact that labels the root node of $T$.
Now, given two CQs $\varphi(\bar x)$ and $\psi(\bar y)$, we write $\varphi(\bar x) \eqtree \psi(\bar y)$ if they are isomorphic.
Clearly, $\approx$ is an equivalence relation over the set of CQs.
For a Datalog query $Q$, $\cq{Q}_{/\approx}$ is the quotient set of $\cq{Q}$ w.r.t.~$\approx$, i.e., the set of all equivalence classes of $\cq{Q}$ w.r.t.~$\approx$. Let $\cqeq{Q}$ be the set of CQs that keeps one arbitrary representative from each member of $\cq{Q}_{/\approx}$ .
Then:

\begin{lemma}\label{lem:finitely-many-cqs}
	For every non-recursive Datalog query $Q$, it holds that $\cqeq{Q}$ is finite.
\end{lemma}

\medskip
\noindent
\textbf{First-Order Rewriting.} Having $\cqeq{Q}$ in place for a non-recursive Datalog query $Q = (\dep,R)$, we can now proceed with the construction of the desired FO query $Q_{\mi{FO}}$.

We start by constructing, for a CQ $\varphi(\bar y) \in \cqeq{Q}$, an FO query $Q_{\varphi(\bar y)} = \psi_{\varphi(\bar y)}(x_1,\ldots,x_{\arity{R}})$, where $x_1,\ldots,x_{\arity{R}}$ are distinct variables that do not occur in any of the CQs of $\cqeq{Q}$, with the following property: for every database $D$ and tuple $\bar t \in \adom{D}^{\arity{R}}$, $\bar t \in Q_{\varphi(\bar y)}(D)$ iff $\bar t$ is an answer to $\varphi(\bar y)$ over $D$, and, in addition, {\em all} the atoms of $D$ are used in order to entail the sentence $\varphi[\bar y/\bar t]$, i.e., there are no other facts in $D$ besides the ones that have been used as witnesses for the atoms occurring in $\varphi[\bar y/\bar t]$.
Assume that $\varphi$ is of the form $\exists \bar z \, (R_1(\bar w_1) \wedge \cdots \wedge R_n(\bar w_n))$. The formula $\psi_{\varphi(\bar y)}$, with free variables $x_1,\ldots,x_{\arity{R}}$, is of the form
\[
\exists \bar y \exists \bar z \left(\varphi_1\ \wedge\ \varphi_2\ \wedge\ \varphi_3\right),
\]
where each conjunct is defined as follows. We write $\bar x$ for the tuple $(x_1,\ldots,x_{\arity{R}})$ and $\bar u_P$, where $P$ is a predicate, for the tuple of  variables $(u_1,\ldots,u_{\arity{P}})$. Furthermore, for two tuples of variables $\bar u = (u_1,\ldots,u_k)$ and $\bar v = (v_1,\ldots,v_k)$, $(\bar u = \bar v)$ is a shortcut for $\bigwedge_{i = 1}^{k} (u_i = v_i)$.
The formula $\varphi_1$ is
\[
\bigwedge\limits_{i \in [n]}\, R_i(\bar w_i)\ \wedge\ (\bar x = \bar y)\ \wedge\ \bigwedge\limits_{\substack{u,v \in \var{\varphi}, \\ u \neq v}} \neg (u = v)
\]
which states that each atom in $\varphi$ should be satisfied by assigning different values to different variables of $\varphi$.
The formula $\varphi_2$ is defined as
\[
\bigwedge\limits_{P \in \{R_1,\ldots,R_n\}} \neg \left(\exists \bar u_{P} \left(P(\bar u_P)\ \wedge\ \bigwedge\limits_{\substack{i \in [n], \\ R_i = P}}\, \neg(\bar w_i = \bar u_P)\right)\right)
\]
which essentially states that, for each predicate $P$ occurring in $\varphi$, the only atoms in the underlying database with predicate $P$ are those used as witnesses for the atoms of $\varphi$.
Finally, the formula $\varphi_3$ is defined as
\[
\bigwedge\limits_{P \in \esch{\dep} \setminus \{R_1,\ldots,R_n\}} \neg \left(\exists \bar u_P \, P(\bar u_P)\right)
\]
which expresses that there are no atoms in the underlying database with a predicate that does not appear in $\varphi$.
%

With the FO query $Q_{\varphi(\bar y)}$ for each CQ $\varphi(\bar y) \in \cqeq{Q}$ in place, it should be clear that the desired FO query $Q_{\mi{FO}}$ is defined as $\Phi(x_1\ldots,x_{\arity{R}})$, where $\Phi = \bigvee_{\varphi(\bar y) \in \cqeq{Q}} \psi_{\varphi(\bar y)}$ and the next technical result follows:

\def\lemmaforewriting{
		Given a non-recursive Datalog query $Q=(\dep,R)$, a database $D$ over $\esch{\dep}$, $\bar t \in \adom{D}^{\arity{R}}$, and $D' \subseteq D$, it holds that $D' \in \why{\bar t}{D}{Q}$ iff $\bar t \in Q_{\mi{FO}}(D')$.
}
\begin{lemma}\label{lem:fo-tree-equiv}
	\lemmaforewriting
\end{lemma}

\subsection{Refined Proof Trees}\label{sec:refined-trees}

The standard notion of why-provenance
relies on arbitrary proof trees without any restriction. 
%
However, as already discussed in the literature (see, e.g., the recent work~\cite{BBPT22}), there are proof trees that are counterintuitive. Such a proof tree, for instance, is the second one in Example~\ref{exa:proof-tree} as the fact $A(a)$ is derived from itself.
Now, a member $D'$ of $\why{\bar t}{D}{Q}$, witnessed via such an unnatural proof tree, might be classified as a counterintuitive explanation of $\bar t$ as it does not correspond to an intuitive derivation process, which can be extracted from the proof tree, that leads from $D'$ to the fact $R(\bar t)$.
%
This leads to the need of considering refined classes of proof trees that overcome the conceptual limitations of arbitrary proof trees.
%
Two well-justified notions considered in the literature are {\em non-recursive proof trees} and {\em minimal-depth proof trees}~\cite{BBPT22}. Roughly, a non-recursive proof tree is a proof tree that does not contain two nodes labeled with the same fact and such that one is the descendant of the other, whereas a minimal-depth proof tree is a proof tree that has the minimum depth among all the proof trees of a certain tuple.
We analyzed the data complexity of why-provenance focusing only on proof trees from those refined classes, and proved that it remains unchanged. Due to space constraints, we omit the details that can be found in the extended version of the paper.

\section{Unambiguous Proof Trees}\label{sec:unambiguous-trees}

Although non-recursive and minimal-depth proof trees form central classes that deserve our attention, there are still proof trees from those classes that can be classified as counterintuitive. More precisely, we can devise proof trees that are both non-recursive and minimal-depth, but they are ambiguous concerning the way some facts are derived.

\begin{example}\label{exa:unambiguous-trees}
	Let $Q = (\dep,A)$, where $\dep$ is the Datalog program that encodes the path accessibility problem as in Example~\ref{exa:proof-tree}. Consider also the database
	\[
	D\ =\ \{S(a),S(b),T(a,a,c),T(b,b,c),T(c,c,d)\}.
	\]
	The following is a proof tree of the fact $A(d)$ w.r.t.~$D$ and $\dep$ that is both non-recursive and minimal-depth, but suffers from the ambiguity issue mentioned above:
	
	\centerline{
	\includegraphics[width=.48\textwidth]{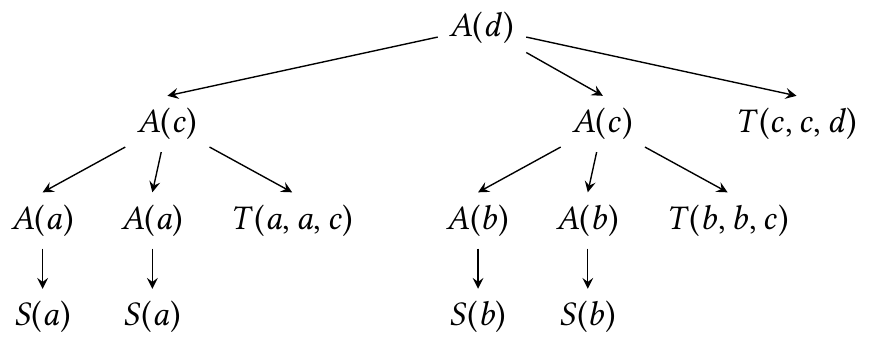}}

	\noindent Indeed, there are two nodes labeled with the fact $A(c)$, but their subtrees differ, and thus, it is ambiguous how $A(c)$ is derived. Hence, the database $D$, which belongs to the why-provenance of $(d)$ w.r.t.~$D$ and $Q$ relative to non-recursive and minimal-depth proof trees
	due to the above proof tree, might be classified as a counterintuitive explanation since it does not correspond to an intuitive derivation process where each fact is derived once due to an unambiguous  reason. \hfill\markfull
\end{example}

The above discussion leads to the novel class of unambiguous proof trees, where all occurrences of a fact in such a tree must be proved via the same derivation.


\begin{definition}[\textbf{Unambiguous Proof Tree}]\label{def:unambiguous-proof-tree}
	Consider a Datalog program $\dep$, a database $D$ over $\esch{\dep}$, and a fact $\alpha$ over $\sch{\dep}$. An {\em unambiguous proof tree of $\alpha$ w.r.t.~$D$ and $\dep$} is a proof tree $T = (V,E,\lambda)$ of $\alpha$ w.r.t.~$D$ and $\dep$ such that, for all $v,u \in V$, $\lambda(v) = \lambda(u)$ implies $T[v] \eqtree T[u]$. \hfill\markfull
\end{definition}

Considering again Example~\ref{exa:unambiguous-trees}, we can construct an unambiguous proof tree of $A(d)$ w.r.t.~$D$ and $\dep$ by simply replacing the subtree of the second child of $A(d)$ with the subtree of its first child (or vice versa).
Now, why-provenance relative to unambiguous proof trees is defined as expected: for a Datalog query $Q = (\dep,R)$, a database $D$ over $\esch{\dep}$, and a tuple $\bar t \in \adom{D}^{\arity{R}}$, the {\em why-provenance of $\bar t$ w.r.t.~$D$ and $Q$ relative to unambiguous proof trees} is the family
\begin{multline*}
\{\support{T} \mid T \text{ is an unambiguous proof tree of }\\
R(\bar t) \text{ w.r.t. } D \text{ and } \dep\}
\end{multline*}
denoted $\unwhy{\bar t}{D}{Q}$.
Considering again Example~\ref{exa:unambiguous-trees}, $\unwhy{(d)}{D}{Q}$ consists of $\{S(a),T(a,a,c),T(c,c,d)\}$ and $\{S(b),T(b,b,c),T(c,c,d)\}$, which is what one expects as conceptually intuitive explanations for the tuple $(d)$, unlike the whole database $D$.
The algorithmic problems
\[
\mathsf{Why\text {-}Provenance_{UN}[C]} \quad \text{and} \quad  \mathsf{Why\text {-}Provenance_{UN}}[Q]
\] 
are defined in the expected way. We can show that the data complexity of why-provenance remains unchanged.

\def\theunambiguouscomplexity{
	The following hold:
	\begin{enumerate}
		\item $\mathsf{Why\text {-}Provenance_{UN}[C]}$ is \NP-complete in data complexity, for each class $\class{C} \in \{\DAT,\LDAT\}$.
		\item $\mathsf{Why\text {-}Provenance_{UN}[\NRDAT]}$ is in $\ACZ$ in data compl.
	\end{enumerate}
}
\begin{theorem}\label{the:complexity-unambiguous-proof-trees}
	\theunambiguouscomplexity
\end{theorem}


For item (1), we show that $\mathsf{Why\text {-}Provenance_{UN}[\DAT]}$ is in \NP~and $\mathsf{Why\text {-}Provenance_{NR}[\LDAT]}$ is \NP-hard.
The latter is established via a reduction from the problem of deciding whether a directed graph has a Hamiltonian cycle.
The \NP~upper bound relies on a characterization of the existence of an unambiguous proof tree of a fact $\alpha$ w.r.t.~a database $D$ and a Datalog program $\dep$ with $\support{T} = D' \subseteq D$ via the existence of a so-called {\em unambiguous proof DAG} $G$ of $\alpha$ w.r.t.~$D$ and $\dep$ with $\support{G} = D'$ of polynomial size. 
Interestingly, unlike arbitrary
proof trees, we can directly go from an unambiguous proof tree $T$ to a polynomially-sized unambiguous proof DAG with the same support as $T$, without applying any intermediate steps for reducing the depth or the subtree count of $T$. This is because an unambiguous proof tree has, by definition, ``small'' depth and subtree count (in fact, the subtree count is one).
The $\ACZ$ upper bound in item (2) is shown via FO rewritability. The target FO query is obtained as in the proof of Theorem~\ref{the:non-recursive-complexity}, but considering only unambiguous proof trees in the definition of $\cq{Q}$.

\subsection{Computing Why-Provenance via SAT Solvers}\label{sec:reduction-to-sat}

We proceed to discuss how off-the-shelf SAT solvers can be used to efficiently compute the why-provenance of a tuple relative to unambiguous proof trees. 
We then discuss a proof-of-concept implementation and report encouraging results of a preliminary experimental evaluation. 
Let us stress that focusing on unambiguous proof trees was crucial towards these encouraging results as it is unclear how a SAT-based implementation can be made practical for proof trees that are not unambiguous. This is mainly because unambiguous proof trees, unlike other classes of proof trees, have always subtree count one, which is crucial for keeping the size of the Boolean formula manageable.


%

Consider a Datalog query $Q = (\dep,R)$, a database $D$ over $\esch{\dep}$, and a tuple $\bar t \in \adom{D}^{\arity{R}}$. We construct in polynomial time in $D$ a Boolean formula $\phi_{(\bar t,D,Q)}$ such that the why-provenance of $\bar t$ w.r.t.~$D$ and $Q$ relative to unambiguous proof trees can be computed from the truth assignments that make $\phi_{(\bar t,D,Q)}$ true.
This relies on the characterization mentioned above of the existence of an unambiguous proof tree of $R(\bar t)$ w.r.t.~$D$ and $\dep$ with $\support{T} = D' \subseteq D$ via the existence of an unambiguous proof DAG $G$ of $R(\bar t)$ w.r.t.~$D$ and $\dep$ with $\support{G} = D'$.
The formula $\phi_{(\bar t,D,Q)}$ is of the form
$\phi_{\mi{graph}} \wedge \phi_{\mi{acyclic}} \wedge \phi_{\mi{root}} \wedge \phi_{\mi{proof}}$, where $\phi_{\mi{graph}}$ verifies that a truth assignment corresponds to a syntactically correct labeled directed graph $G$, $\phi_{\mi{acyclic}}$ verifies that $G$ is acyclic, $\phi_{\mi{root}}$ verifies that $R(\bar t)$ is the unique root of $G$, and $\phi_{\mi{proof}}$ verifies that $G$ is an unambiguous proof DAG.
%

The key ingredient in the construction of $\phi_{(\bar t,D,Q)}$ is the so-called  {\em downward closure of $R(\bar t)$ w.r.t.~$D$ and $\dep$}, taken from~\cite{ElKM22}, which, intuitively speaking, is a hypergraph that encodes all possible proof DAGs of $R(\bar t)$ w.r.t.~$D$ and $\dep$. 
We first construct this hypergraph $H$, which can be done in polynomial time in the size of $D$, and then guided by $H$ we build the formula $\phi_{(\bar t,D,Q)}$, which essentially searches for an unambiguous proof DAG inside the hypergraph $H$.
%
Now, a truth assignment $\tau$ to the variables of $\phi_{(\bar t,D,Q)}$ naturally gives rise to a database denoted $\db{\tau}$. Let $\sem{\phi_{(\bar t,D,Q)}}$ be the family
\[
\left\{\db{\tau} \mid \tau \text{ is a satisfying assignment of } \phi_{(\bar t,D,Q)}\right\}.
\]
We can then show the next technical result:

\def\prowhyprovenancesat{
	Consider a Datalog query $Q = (\dep,R)$, a database $D$ over $\esch{\dep}$, and a tuple $\bar t \in \adom{D}^{\arity{R}}$. It holds that $\unwhy{\bar t}{D}{Q} = \sem{\phi_{(\bar t,D,Q)}}$.
}

\begin{proposition}\label{pro:why-provenance-sat}
\prowhyprovenancesat
\end{proposition}

The above proposition provides a way for computing the why-provenance of a tuple relative to unambiguous proof trees via off-the-shelf SAT solvers.
But how does this machinery behave when applied in a practical context? In particular, we are interested in the incremental computation of the why-provenance by enumerating its members instead of computing the whole set at once. The rest of the section is devoted to providing a preliminary answer to this question.

\subsection{Some Implementation Details}

Before presenting our experimental results, let us first briefly discuss some interesting aspects of the implementation. In what follows, fix a Datalog query $Q = (\dep,R)$, a database $D$ over $\esch{\dep}$, and a tuple $\bar t \in \adom{D}^{\arity{R}}$.

\medskip

\noindent \textbf{Constructing the Downward Closure.} Recall that the construction of $\phi_{(\bar t,D,Q)}$ relies on the downward closure of $R(\bar t)$ w.r.t.~$D$ and $\dep$. It turns out that the hyperedges of the downward closure can be computed by executing a slightly modified Datalog query $Q_{\downarrow}$ over a slightly modified database $D_{\downarrow}$. In other words, the answers to $Q_{\downarrow}$ over $D_{\downarrow}$ coincide with the hyperedges of the downward closure. Hence, to construct the downward closure we exploit a state-of-the-art Datalog engine, that is, version 2.1.1 of DLV~\cite{AACC+18}.
Note that our approach based on evaluating a Datalog query differs form the one in~\cite{ElKM22}, which uses an extension of Datalog with set terms.


\medskip
\noindent \textbf{Constructing the Formula.} Recall that $\phi_{(\bar t,D,Q)}$ consists of four conjuncts, where each one is responsible for a certain task. As it might be expected, the heavy task is to verify that the graph in question is acyclic (performed by the formula $\phi_{\mi{acyclic}}$).
Checking the acyclicity of a directed graph via a Boolean formula is a well-studied problem in the SAT literature.
%
For our purposes, we employ the technique of {\em vertex elimination}~\cite{RankoohR22}.
The advantage of this approach is that the number of Boolean variables needed for the encoding of $\phi_{\mi{acyclic}}$ is of the order $O(n \cdot \delta)$, where $n$ is the number of nodes of the graph, and $\delta$ is the so-called \emph{elimination width} of the graph, which, intuitively speaking, is related to how connected the graph is.

\medskip
\noindent \textbf{Incrementally Constructing the Why-Provenance.} Recall that we are interested in the incremental computation of the why-provenance, which is more useful in practice than computing the whole set at once. To this end, we need a way to enumerate all the members of the why-provenance without repetitions. This is achieved by adapting a standard technique from the SAT literature for enumerating the satisfying assignments of a Boolean formula, called {\em blocking clause}.
We initially collect in a set $S$ all the facts of $D$ occurring in the downward closure of $R(\bar t)$ w.r.t.~$D$ and $\dep$. Then, after asking the SAT solver for an arbitrary satisfying assignment $\tau$ of $\phi_{(\bar t,D,Q)}$, we output the database $\db{\tau}$, and then construct the ``blocking'' clause
$
\vee_{\alpha \in S} \ell_\alpha,
$
where $\ell_\alpha = \neg x_\alpha$ if $\alpha \in \db{\tau}$, and $\ell_\alpha = x_\alpha$ otherwise. We then add this clause to the formula, which expresses that no other satisfying assignment $\tau'$ should give rise to the same member of the why-provenance.
This will exclude the previously computed explanations from the computation. We keep adding such blocking clauses each time we get a new member of the why-provenance until the formula is unsatisfiable.

\subsection{Experimental Evaluation}

We now proceed to experimentally evaluate the SAT-based approach discussed above. To this end, we consider a variety of scenarios from the literature consisting of a Datalog query $Q = (\dep,R)$ and a family of databases $\mathcal{D}$ over $\esch{\dep}$.

{\footnotesize 
	\begingroup
	\setlength{\tabcolsep}{5pt} 
	\renewcommand{\arraystretch}{1.3} 
	\begin{table*}[t]
		\centering
		\begin{tabular}{|c||c|c|c|}
			\hline
			\textbf{Scenario} & \textbf{Databases} & \textbf{Query Type}  & \textbf{Number of Rules}\\ \hline
			\hline
			$\mathsf{TransClosure}$ & $D_\mathsf{bitcoin}$ (235K), $D_\mathsf{facebook}$ (88.2K) & linear, recursive & 2 \\ \hline
			$\mathsf{Doctors\text{-}}i$, $i \in [7]$ & $D_1$ (100K) & linear, non-recursive & 6\\ \hline
			$\mathsf{Galen}$ & $D_1$ (26.5K), $D_2$ (30.5K), $D_3$ (67K), $D_4$ (82K) & non-linear, recursive & 14\\ \hline
			$\mathsf{Andersen}$ & $D_1$ (68K), $D_2$ (340K), $D_3$ (680K), $D_4$ (3.4M), $D_5$ (6.8M) & non-linear, recursive & 4 \\ \hline
			$\mathsf{CSDA}$ & $D_\mathsf{httpd}$ (10M), $D_\mathsf{postgresql}$ (34.8M), $D_\mathsf{linux}$ (44M) & linear, recursive & 2\\ \hline
		\end{tabular}
		\caption{Experimental scenarios.}
		\label{tab:scenarios}
	\end{table*}
	\endgroup
}

\medskip
\noindent \textbf{Experimental Scenarios.} All the considered scenarios are summarized in Table~\ref{tab:scenarios}.
Here is brief description:

\begin{description}
	\item[$\mathsf{TransClosure}$.] This scenario computes the transitive closure of a graph and asks for connected nodes. The database $D_\mathsf{bitcoin}$ stores a portion of the Bicoin network~\cite{Weber19}, whereas
	$D_\mathsf{facebook}$ stores different ``social circles'' from Facebook~\cite{McAuley12}.
	%

	\item[$\mathsf{Doctors}$.] The scenarios $\mathsf{Doctors\text{-}}i$, for $i \in [7]$, were used in~\cite{ElKM22} and represent queries obtained from a well-known data-exchange benchmark involving existential rules (the existential variables have been replaced with fresh constants). All such scenarios share the same database with 100K facts.

	\item[$\mathsf{Galen}$.] This scenario used in~\cite{ElKM22} implements the ELK calculus~\cite{KazakovKS14} and asks for all pairs of concepts that are related with the $\mathsf{subClassOf}$ relation. The various databases contain different portions of the Galen ontology~\cite{Galen}.
	
	\item[$\mathsf{Andersen}$.] This scenario used in~\cite{FanMK22} implements the classical Andersen ``points-to'' algorithm for determining the flow of data in procedural programs and asks for all the pairs of a pointer $p$ and a variable $v$ such that $p$ points to $v$. The databases are encodings of program statements of different length.

	\item[$\mathsf{CSDA}$.] This scenario (Context-Sensitive Dataflow Analysis) used in~\cite{FanMK22} is similar to $\mathsf{Andersen}$ but asks for null references in a program. The databases $D_\mathsf{httpd}$, $D_{\mathsf{postgresql}}$, and $D_{\mathsf{linux}}$ store the statements of the httpd web server, the PostgreSQL DBMS, and the Linux kernel, respectively. 
\end{description}

\medskip
\noindent \textbf{Experimental Setup.} For each scenario $s$ consisting of the query $Q = (\dep,R)$ and the family of databases $\mathcal{D}$, and for each $D \in \mathcal{D}$, we have computed $Q(D)$ using DLV, and then selected five tuples $\bar t^1_{s,D},\ldots,\bar t^5_{s,D}$ from $Q(D)$ uniformly at random. 
Then, for each $i \in [5]$, we constructed the downward closure of $R(\bar t^i_{s,D})$ w.r.t.~$D$ and $\dep$ by first computing the adapted query $Q_{\downarrow}$ and database $D_{\downarrow}$ via a Python 3 implementation and then using DLV for the actual computation of the downward closure, then
we constructed the Boolean formula $\phi_{(\bar t^i_{s,D},D,Q)}$ via a C++ implementation, and finally 
we ran the state-of-the-art SAT solver Glucose (see, e.g.,~\cite{Audemard18}), version 4.2.1, with input the above formula to enumerate the members of $\unwhy{\bar t^i_{s,D}}{D}{Q}$.
%
%
All the experiments have been conducted on a laptop with an Intel(R) Core(TM) i7-10750H CPU @ 2.60GHz, and 32GB of RAM, running Fedora Linux 37. The Python code is executed with Python 3.11.2, and the C++ code has been compiled with g++ 12.2.1, using the -O3 optimization flag.

\begin{figure}[t]
	\centering
	\includegraphics[width=.464\textwidth]{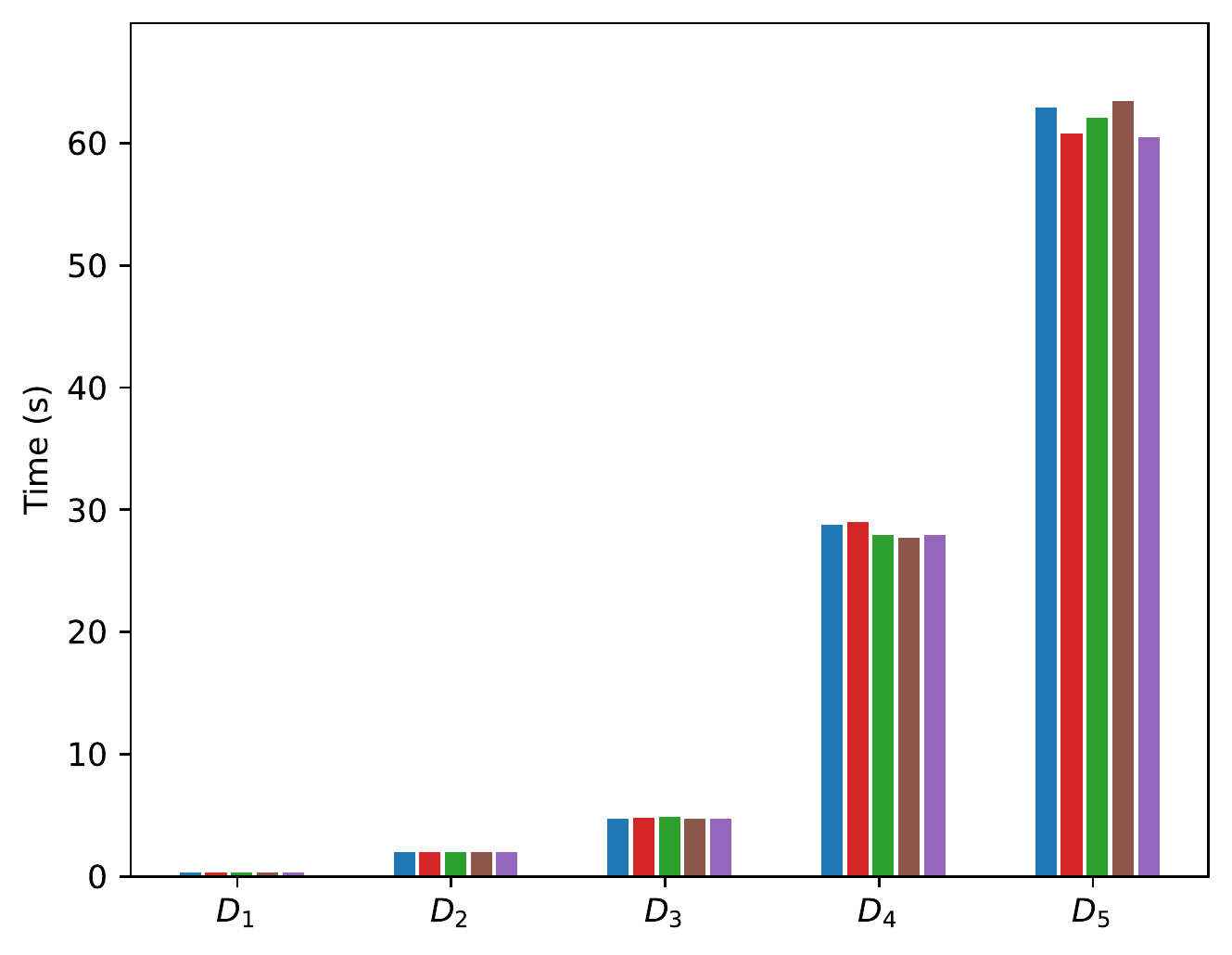}
	\caption{Building the downward closure and the Boolean formula.}
	\label{fig:andersen-task1}
\end{figure}

\begin{figure}[t]
	\centering
	\includegraphics[width=.47\textwidth]{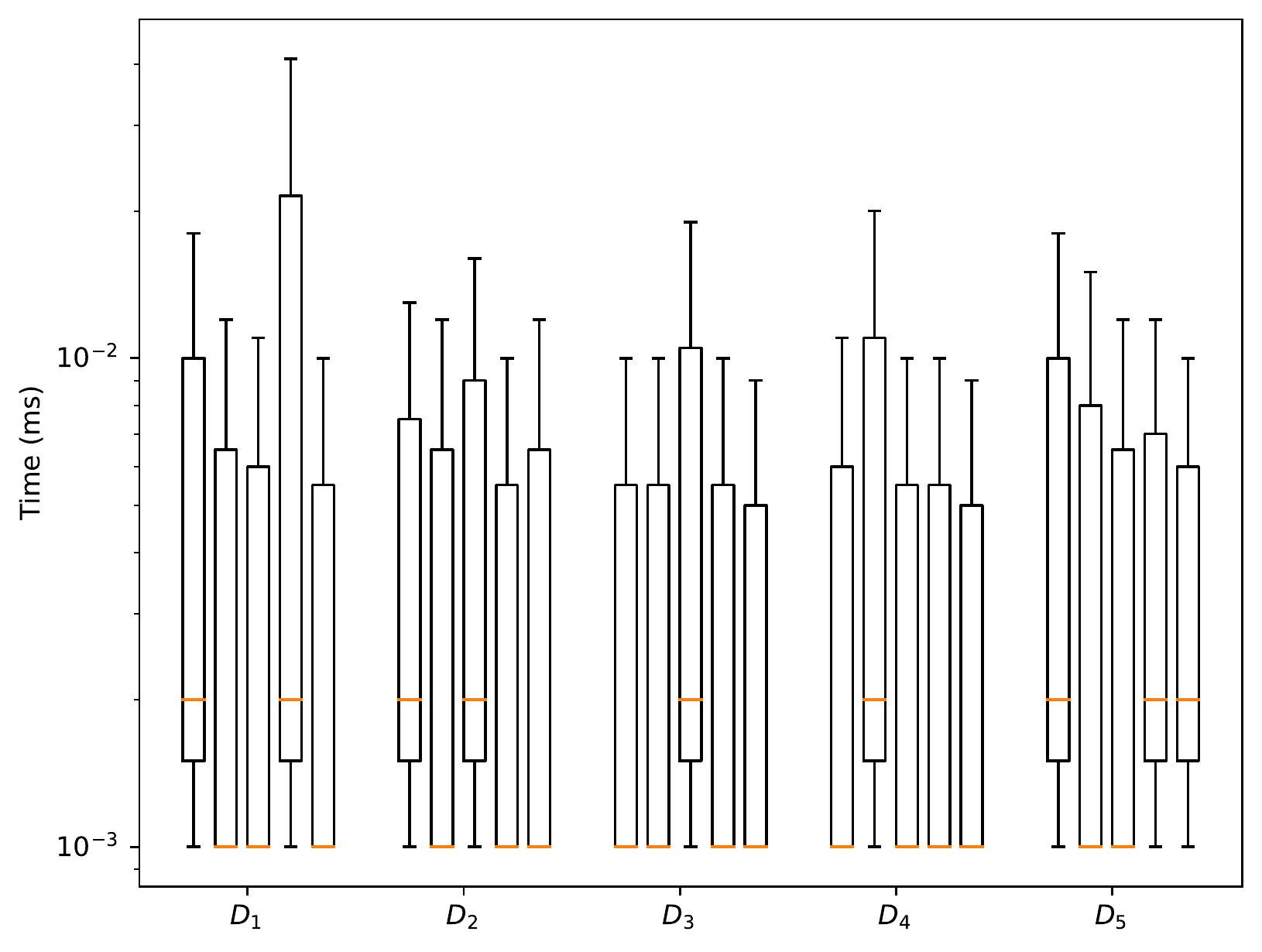}
	\caption{Incremental computation of the why-provenance.}
	\label{fig:andersen-task2}
\end{figure}

\medskip
\noindent \textbf{Experimental Results.} Due to space constraints, we are going to present only the results based on the $\mathsf{Andersen}$ scenario. Nevertheless, the final outcome is aligned with what we have observed based on all the other scenarios.
%

Concerning the construction of the downward closure and the Boolean formula, we report in Figure~\ref{fig:andersen-task1} the total running time for each database of the $\mathsf{Andersen}$ scenario (recall that there are five databases of varying size, and thus we have five plots). Furthermore, each plot consists of five bars that correspond to the five randomly chosen tuples. Each such bar shows the time for building the downward closure plus the time for constructing the Boolean formula. 
%
%
We have observed that almost all the time is spent for computing the downward closure, whereas the time for building the formula is negligible. Hence, our efforts should concentrate on improving the computation of the downward closure.
Moreover, for the reasonably sized databases (68K, 340K, and 680K facts) the total time is in the order of seconds, which is quite encouraging. Now, for the very large databases that we consider (3.4M and 6.8M facts), the total time is between half a minute and a minute, which is also encouraging taking into account the complexity of the query, the large size of the databases, and the limited power of our machine.

For the incremental computation of the why-provenance, we give in Figure~\ref{fig:andersen-task2}, for each database of the $\mathsf{Andersen}$ scenario, the times required to build an explanation, that is, the time between the current member of the why-provenance and the next one (this time is also known as the delay).
Each of the five plots collects the delays of constructing the members of the why-provenance (up to a limit of 10K members or 5 minutes timeout) for each of the five randomly chosen tuples. We use box plots, where the bottom and the top borders of the box represent the first and third quartile, i.e., the delay under which 25\% and 75\% of all delays occur, respectively, and the orange line represents the median delay. Moreover, the bottom and the top whisker represent the minimum and maximum delay, respectively. All times are expressed in milliseconds and we use logarithmic scale.
As we can see, most of the delays are below 1 millisecond, with the median in the order of microseconds. Therefore, once we have the Boolean formula in place, incrementally computing the members of the why-provenance is extremely fast.

\section{Conclusions}\label{sec:conclusions}
%


The takeaway of our work is that for recursive queries the why-provenance problem is, in general, intractable, whereas for non-recursive queries it is highly tractable in data complexity. 
With the aim of overcoming the conceptual limitations of arbitrary proof trees, we considered the new class of unambiguous proof trees and showed that it does not affect the data complexity of the why-provenance problem.
Interestingly, we have experimentally confirmed that unambiguous proof trees help to exploit off-the-shelf SAT solvers towards an efficient computation of the why-provenance.
Note that we have performed a preliminary comparison with~\cite{ElKM22} by focusing on a setting that both approaches can deal with. In particular, we used the scenarios $\mathsf{Doctors}\text{-}i$, for $i \in [7]$, and measured the end-to-end runtime of our approach (not the delays). For the simple scenarios, the two approaches are comparable in the order of a second. For the demanding scenarios ($\mathsf{Doctors}\text{-}i$ for $i \in \{1,5,7\}$), our approach is generally faster.


It would be extremely useful to provide a complete classification of the data complexity of the why-provenance problem in the form of a dichotomy result. It would also provide further insights to pinpoint the combined complexity of the problem, where the Datalog query is part of the input. Finally, it is crucial to perform a more thorough experimental evaluation of our SAT-based machinery in order to understand better whether it can be applied in practice.

\bibliographystyle{kr}

\newpage
\appendix

\section{Data Complexity of Why-Provenance}\label{appsec:all-trees}

In this section, we provide the missing details for Section~\ref{sec:complexity}.

\subsection{Recursive Queries}

We proceed to give the full proof of Theorem~\ref{the:recursive-complexity}, which we recall here for convenience:

\begin{manualtheorem}{\ref{the:recursive-complexity}}
	\therecursivecomplexity
\end{manualtheorem}

To prove the above result, it suffices to show that:
\begin{itemize}
	\item $\mathsf{Why\text {-}Provenance[\DAT]}$ is in \NP~in data complexity.
	\item $\mathsf{Why\text {-}Provenance[\LDAT]}$ is \NP-hard in data complexity.
\end{itemize}

\medskip
\noindent \underline{\textbf{Upper Bound}}
\smallskip

\noindent Our main task is to prove Proposition~\ref{pro:characterization-all-trees}, which we recall below, that will allow us to devise a guess-and-check procedure that runs in polynomial time in the size of the database.

\begin{manualproposition}{\ref{pro:characterization-all-trees}}
	\procharacterizationalltrees
\end{manualproposition}

The direction $(2)$ implies $(1)$ is shown by ``unravelling'' the proof DAG $G$ of $\alpha$ w.r.t.~$D$ and $\dep$ into a proof tree $T$ of $\alpha$ w.r.t.~$D$ and $\dep$ with $\support{T} = \support{G}$. 
More precisely, we go over the nodes of $G$ starting from its root and ending at its leaves using breadth-first search. Whenever we encounter a node $v$ that has $k$ incoming edges, we create $k$ copies of its subDAG. The subDAG of a node $v$ contains $v$ itself and every node reachable from $v$, and an edge $(u_1,u_2)$ if there is an edge $(w_1,w_2)$ in $G$, where $u_1$ is a copy of $w_1$ and $u_2$ is a copy of $w_2$. Note that these copies preserve the labels of the nodes. We then replace each incoming edge of $v$ with an edge to the root of a distinct copy of its subDAG. Note that since $G$ is acyclic, the above operation on $v$ has no impact on the nodes that have been processed before $v$.
It is rather straightforward that the result is a tree $T$ with a root $v$ that has the same label as the root of $G$, and where the leaves have the same labels as the leaves of $G$ (hence, for each leaf $v$ of the tree we have that $\lambda(v)\in D$, as the same holds for the labels of the leaves of the proof DAG). Moreover, it is easy to verify that Property~(3) of Definition~\ref{def:proof-tree} holds since $G$ satisfies the equivalent property~$(3)$ of Definition~\ref{def:proof-dag} and our copies preserve the labels of the nodes. Therefore, the resulting tree is a proof tree of $\alpha$ w.r.t.~$D$ and $\Sigma$. 

Concerning the direction $(1)$ implies $(2)$, as discussed in the main body of the paper, the proof proceeds in three main steps captured by Lemmas~\ref{lem:depth-reduction},~\ref{lem:scount-reduction}, and~\ref{lem:from-trees-to-dags}, which we prove next.

\begin{manuallemma}{\ref{lem:depth-reduction}}
	\lemmadepthreduction
\end{manuallemma}
\begin{proof}
	We prove the claim for $f(|D|)=|\base{D,\dep}|\times |D|$ by induction on $n = \depth{T}$.
	
	\medskip
	\noindent \textbf{Base Case.} For any $n\le |\base{D,\dep}|\times |D|$, the claim holds trivially.
	
	\medskip
	\noindent \textbf{Inductive Step.} We assume that the claim holds for $n \in \{|\base{D,\dep}|\times |D|,\dots,p\}$, and prove that it holds for $n=p+1$. Let $T$ be a proof tree of $\alpha$ w.r.t.~$D$ and $\Sigma$ with $\depth{T} = p+1$. Since $p+1>|\base{D,\dep}|\times |D|$, there exists a path $v_1\rightarrow v_2\dots\rightarrow v_{p+2}$ of length $p+1$ in $T$ and a label $\beta$, such that $\beta$ is the label of $k>|D|$ nodes $v_{i_1},\dots,v_{i_k}$ along the path. (Note that $|\base{D,\dep}|$ is an upper bound on the number of distinct labels in $T$.) We assume, without loss of generality, that $i_1<i_2<\dots<i_k$.
	We will show that for some $v_{i_j}$ and $v_{i_r}$ with $j< r$, it holds that
	\[
	\support{T[v_{i_j}]}\ =\ \support{T[v_{i_r}]}. 
	\]
	Recall that for a node $v$, $T[v]$ is the subtree of $T$ rooted at $v$.
	
	An easy observation is that for $T_1,T_2$ such that $T_2$ is a subtree of $T_1$, it holds that $\support{T_2}\subseteq\support{T_1}$.
	Hence, for all $j<r$, we have that: 
	\[
	\support{T[v_{i_r}]}\subseteq\support{T[v_{i_j}]}.
	\]
	Now assume, towards a contradiction, that for every $j<r$,
	\[\support{T[v_{i_r}]}\subsetneq\support{T[v_{i_j}]}.\] 
	We then conclude that
	\[
	\support{T[v_{i_k}]}\subsetneq\dots\subsetneq \support{T[v_{i_1}]}.
	\]
	This means that $\support{T[v_{i_1}]}$ contains $k>|D|$ distinct facts, which in turn means that $\support{T}$ contains at least $k>|D|$ distinct facts. This is a contradiction to the fact that $\support{T}=D'$ for some $D'\subseteq D$.

	Therefore, for some $j<r$ it holds that 
	\[
	\support{T[v_{i_r}]}=\support{T[v_{i_j}]}.
	\]
	We can now shorten the path $v_1\rightarrow v_2\dots\rightarrow v_{p+2}$ in $T$ and obtain another proof tree $T_1$ with the same support $D'$, by replacing the subtree $T[v_{i_j}]$ with the subtree $T[v_{i_r}]$. An important observation here is that $T_1$ is still a proof tree of $\alpha$ w.r.t.~$D$ and $\Sigma$. Since we do not modify the root node $v$, it still holds that $\lambda(v)=\alpha$. Moreover, the set of leaves of $T_1$ is contained in the set of leaves of $T$; hence, for every leaf $v$ of $T_1$ it holds that $\lambda(v)\in D$. Finally, since $T[v_{i_r}]$ is a subtree of $T$, it satisfies property $(3)$ of Definition~\ref{def:proof-tree} (that is, if $v$ is a node with $n\ge1$ children $u_1,\ldots,u_n$, then there is a rule $R_0(\bar x_0)\ \assign\ R_1(\bar x_1),\ldots,R_n(\bar x_n) \in \dep$ and a function $h : \bigcup_{i \in [n]} \bar x_i \ra \ins{C}$ such that $\lambda(v) = R_0(h(\bar x_0))$, and $\lambda(u_i) = R_i(h(\bar x_i))$ for each $i \in [n]$). Therefore, this property also holds for every node $v$ of $T_1$ (for the parent of the node $v_{i_j}$ that we replace with the node $v_{i_r}$ the property holds because $\lambda(v_{i_j})=\lambda(v_{i_r})$).
	
	Clearly, when applying the above procedure, we eliminate at least one path of length $p+1$ and we do not introduce any new path of length $p+1$.
	If we repeat this process for every path of length $p+1$, we will eventually obtain a proof tree $T_2$ of $\alpha$ w.r.t.~$D$ and $\dep$ with $\depth{T_2} \le p$ and $\support{T_2} = D'$. The claim follows by the inductive hypothesis.
\end{proof}

Before proving Lemma~\ref{lem:scount-reduction}, we show the following result, where we do not consider the support of the proof tree. 

\begin{lemma}\label{lem:exists-unumbiguous-tree}
	For each Datalog program $\dep$, database $D$ over $\esch{\dep}$, and fact $\alpha$ over $\sch{\dep}$, if there exists a proof tree $T$ of $\alpha$ w.r.t.~$D$ and $\dep$, then there exists also such a proof tree $T'$ with $|\quot{T'[\beta]}| =1$ for every fact $\beta$ that occurs in $T'$.
\end{lemma}
\begin{proof}
	Let $T=(V,E,\lambda)$ be a proof tree of $\alpha$ w.r.t.~$D$ and $\dep$. We construct another proof tree $T'=(V',E',\lambda')$ of $\alpha$ w.r.t.~$D$ and $\dep$ with the desired property in the following way. Let $S$ be the set that contains, for every fact $\beta$ that occurs in $T$ (i.e., it is the label of some node in $V$), one subtree $T[v]$ of $T$ with $\lambda(v)=\beta$ of smallest depth among all such subtrees; if more than one such tree exists, we choose one arbitrarily. For every $1\le i\le \depth{T}$, let $S^i$ be the set that contains all the trees of $S$ of depth exactly $i$. We will now inductively construct a set of trees that will contain a single representative tree for every fact $\beta$ that occurs in $T$. Then, we will use the representative tree of $\alpha$ as the tree $T'$.
	
	We define:
	\begin{itemize}
		\item $Z_1 = S^1$;
		\item $Z_{i+1} = \mathsf{Op}^{i+1}(Z_i) \cup Z_i$, for $i\ge 1$.
	\end{itemize}
	where $\mathsf{Op}^{i+1}(Z_i)$ contains, for every tree $T''$ in $S^{i+1}$, the tree that is obtained from it using the following procedure. Let $u_1,\dots,u_n$ be the direct children of the root of $T''$. For every child $u_j$ with $\lambda(u_j)=\beta$, we replace the subtree $T''[u_j]$ with a tree of $Z_i$ whose root is labeled with $\beta$. Intuitively, the existence of such a tree is guaranteed because $S$ contains a smallest depth subtree for each fact, and since $T''[u_j]$ is of depth at most $i$, the set $S$ has a tree of depth at most $i$ rooted with a node labeled with $\beta$.
	Formally, we prove the following properties of the sets $Z_i$:
	\begin{enumerate}
		\item For every label $\beta$, if $S^i$ contains a tree with root $v$ such that $\lambda(v)=\beta$, then $Z_i$ contains a tree with root $u$ such that $\lambda(u)=\beta$.
		\item For every label $\beta$, if $Z_i$ contains a tree with root $u$ such that $\lambda(u)=\beta$, then there is a tree $T''$ with root $w$ such that $\lambda(w)=\beta$ and $T''\in S^{k}$ for some $k\le i$.
		\item For every label $\beta$, if $Z_i$ contains a tree with root $u$ such that $\lambda(u)=\beta$, then there is precisely one such tree.
		\item For every tree $T''$ of $Z_i$, if $u$ is a leaf node of $T''$, then $\lambda(u)\in D$.
		\item For every tree $T''$ of $Z_i$, if $u$ is a node of $T''$ with $n\ge1$ children $u_1,\ldots,u_n$, then there exists a rule $R_0(\bar x_0)\ \assign\ R_1(\bar x_1),\ldots,R_n(\bar x_n)$ in $\dep$ and a function $h : \bigcup_{i \in n} \bar x_i \ra \ins{C}$ such that $\lambda(u) = R_0(h(\bar x_0))$ and $\lambda(u_j) = R_j(h(\bar x_j))$, for $j \in [n]$.
		\item For every tree $T''$ of $Z_i$, $|\quot{T''[\beta]}| =1$, for each fact $\beta$ that occurs in $T''$.
	\end{enumerate}
	We prove all six properties by induction on $i$. 
	
	\medskip
	\noindent \textbf{Base Case.} For $i=1$, the first two properties trivially hold as $Z_1=S^1$ by definition. Since $S$ contains a single tree for each fact $\beta$ (i.e., a tree where the root is labeled with $\beta$), so does $S^1$, and the third property also holds. 
	The fourth and fifth properties hold because every tree of $S^1$ (and so every tree of $Z_1$) is a subtree of a proof tree, and these are properties of proof trees. The last property is satisfied since a tree of $Z_1$ contains one root node $v$ and its children $u_1,\dots,u_n$, and it cannot be the case that $\lambda(v)=\lambda(u_j)$ for some $j \in [n]$ (as the leaves correspond to extensional predicates, while the root corresponds to an intentional predicate).
	
	\medskip
	\noindent \textbf{Inductive Step.} We assume that the claim holds for $i=1,\dots,p$ and prove that it holds for $i=p+1$. 
	The first property holds by construction, since the set $\mathsf{Op}^{p+1}(Z_{p})$ contains, for every tree $T''$ of $S^{p+1}$, another tree with the same root (we only modify the subtrees of its children). Moreover, if the children of the root of $T''$ are $u_1,\dots,u_n$, then for every $r\in [n]$, the subtree $T''[u_r]$ (which is also a subtree of the original $T$) is of depth at most $p$. Hence, the smallest depth subtree for the label $\lambda(u_r)$ in $T$ occurs in $S_k$ for some $1\le k\le p$. By the inductive assumption, the set $Z_k$ contains a tree with root $v$ such that $\lambda(v)=\lambda(u_r)$, and since $Z_k\subseteq Z_{p}$, this tree also appears in $Z_{p}$; hence, our construction is well-defined.

	The second property is satisfied since every tree of $Z_{p+1}$ with root $u$ such that $\lambda(u)=\beta$ either occurs in $Z_p$ or is obtained from a tree of $S_{p+1}$. In the first case, the inductive assumption implies that there is a tree $T''$ with a root $w$ and $\lambda(w)=\beta$ in $S^k$ for some $k\le p$. In the second case, the definition of $\mathsf{Op}^{p+1}(Z_{p})$ implies that there is a tree $T''$ with a root $w$ and $\lambda(w)=\beta$ in $S^{p+1}$.
	
	The third property holds because $S^{p+1}$ contains a single tree per fact, and so the same holds for $\mathsf{Op}^{p+1}(Z_{p})$. Moreover, $Z_p$ contains a single tree per fact due to the inductive assumption. We will show that it cannot be the case that there is a label $\beta$ and two trees $T_1,T_2$ such that: \textit{(1)} $\lambda(v_1)=\beta$ for the root $v_1$ of $T_1$, \textit{(2)} $\lambda(v_2)=\beta$ for the root $v_2$ of $T_2$, \textit{(3)} $T_1\in Z_p$, and \textit{(4)} $T_2\in \mathsf{Op}^{p+1}(Z_{p})$. Assume, towards a contradiction, that such two trees exist. Then, $S^{p+1}$ contains a tree $T_3$ with root $v_3$ such that $\lambda(v_3)=\beta$ (this is the tree from which $T_2$ is obtained). Moreover, the inductive assumption and property $(2)$ imply that there is a tree $T_4$ with root $v_4$ in $S^k$ for some $k\le p$ such that  $\lambda(v_4)=\beta$. We conclude that $S$ contains two trees whose root is labeled with $\beta$ -- one with depth $p+1$ and one with depth $k\le p$. This is a contradiction to the fact that $S$ only contains one smallest depth subtree of $T$ whose root is labeled with $\beta$.
	
	As for the fourth property, as aforementioned, every tree $T''$ of $Z_{p+1}$ either occurs in $Z_p$ or is obtained from a tree of $S_{p+1}$. In the first case, the claim immediately follows from the inductive assumption. In the second case, let $T''$ be a tree of $\mathsf{Op}^{p+1}(Z_{p})$. Assume that the root of $T''$ is $u$ and its children are $u_1,\dots,u_n$. Then, every leaf node of $T''$ is also a leaf node of $T''[u_j]$ for some $j\in [n]$, and since $T''[u_j]$ is a tree of $Z_p$ by construction, we have that $\lambda(u_j)\in D$ by the inductive assumption.
	
	The fifth property holds for every tree of $Z_p$ by the inductive assumption. We will show that it also holds for every tree $T''$ of $\mathsf{Op}^{p+1}(Z_{p})$.
	Each tree of $S^{p+1}$ is a subtree of $T$; hence, it satisfies the desired property (which is a property of proof trees). In particular, the property is satisfied by the root node, and since we do not modify the label of the root node or the labels of its children, the root of the obtained tree $T''$ also satisfies this property. For every child of the root, its subtree is replaced with a tree from $Z_{p}$ that satisfies the desired property by the inductive assumption, and so every node of $T''$ satisfies this property.
	
	Finally, for a tree of $Z_{p+1}$ that also occurs in $Z_p$, the last property holds from the inductive assumption. For a tree $T''$ of $Z_{p+1}$ that comes from $\mathsf{Op}^{p+1}(Z_{p})$, the last property holds for every label $\delta$ that occurs in $T''$ and is not the label of the root, due to the inductive assumption (since we replace the subtrees under the children of the root with trees from $Z_p$). Note that the label $\beta$ of the root of $T''$ cannot occur in a tree of $Z_p$ (and, in particular, as the label of one of its children). This holds since $Z_p\subseteq Z_{p+1}$ and due to property $(3)$ that we have already proved. Thus, it also holds that $|\quot{T''[\beta]}| =1$.
	This concludes our proof for the six properties.
	
	Now, by definition, $S$ contains a tree $T_1$ whose root is labeled with $\alpha$. Assume that the depth of this tree is $k$, then $T_1\in S^k$. Property $(1)$ then implies that $Z_k$ contains a tree $T_2$ whose root is labeled with $\alpha$. It is only left to show that this tree is a proof tree of $\alpha$ w.r.t.~$D$ and $\dep$ that satisfies the desired properties, and then we will define $T'=T_2$, and that will conclude our proof. The first property of proof trees (Definition~\ref{def:proof-tree}) is clearly satisfied as $\lambda(v)=\alpha$ for the root node $v$ of $T_2$. Properties $(2)$ and $(3)$ of proof trees are satisfied due to properties $(4)$ and $(5)$, respectively, of the sets $Z_i$. Hence, $T_2$ is indeed a proof tree of $\alpha$ w.r.t.~$D$ and $\dep$. Property $(6)$ of the sets $Z_i$ implies that $|\quot{T_2[\beta]}| =1$, for each fact $\beta$ that occurs in $T_2$. Therefore, we can indeed define $T'=T_2$ and obtain the desired proof tree.
\end{proof}

We now proceed to prove Lemma~\ref{lem:scount-reduction}, which we recall here:

\begin{manuallemma}{\ref{lem:scount-reduction}}
	\lemmascountreduction
\end{manuallemma}
\begin{proof}
	Given a proof tree $T$ of $\alpha$ w.r.t.~$D$ and $\dep$ with $\depth{T}\le f(|D|)$ and $\support{T}=D'$, we construct another proof tree $T'$ of $\alpha$ w.r.t.~$D$ and $\dep$ with $\scount{T'} \leq g(|D|)$ and $\support{T'}=D'$ in two steps. First, for every fact of $D'$, we select one path in $T$ from the root to a leaf labeled with this fact. Then, we ``freeze'' those paths (i.e., we do not modify them in $T'$) in order to preserve the support. However, it is not sufficient to keep only these paths, but we also need to keep the siblings of each node along these paths to obtain a valid proof tree, and, in particular, to satisfy the last property of Definition~\ref{def:proof-tree}. The second step is then to reduce, for every sibling node $v$ and for every fact $\beta$ in $T[v]$, the number of equivalence classes in $\quot{T[v][\beta]}$.
	
	Let $v$ be a node of $T$. An easy observation is that if $\lambda(v)=\beta$ for some fact $\beta$, then $T[v]$ is a proof tree of $\beta$ w.r.t.~$D$ and $\dep$. Then, Lemma~\ref{lem:exists-unumbiguous-tree} implies that there exists another proof tree $T_v$ of $\beta$ w.r.t.~$D$ and $\dep$ whose root node is labeled with $\beta$, such that $|\quot{T_v[\delta]}| =1$, for every fact $\delta$ that occurs in $T_v$. We can then replace the subtree $T[v]$ of $T$ with the tree $T_v$. We do that for each sibling node. Clearly, the tree $T'$ that we obtain via this procedure remains a proof tree of $\alpha$ w.r.t.~$D$ and $\dep$. This holds since we do not modify the label of the root; hence, Property~$(1)$ of Definition~\ref{def:proof-tree} holds. Moreover, Property~$(2)$ of Definition~\ref{def:proof-tree} holds because the set of leaves of $T'$ is a subset of the set of leaves of $T$ (due to the construction in the proof of Lemma~\ref{lem:exists-unumbiguous-tree}). Finally, Property~$(3)$ of Definition~\ref{def:proof-tree} is satisfied for every node in a subtree $T_v$ since, as aforementioned, every $T_v$ is a proof tree of some fact w.r.t.~$D$ and $\dep'$. Moreover, this property holds for every node along the frozen paths because for each such path $v_1\ra\dots\ra v_n$, if $u_1,\dots,u_m$ are the children of $v_j$, then one of these children is $v_{j+1}$ and we do not modify its label or the labels of its siblings (we only replaces the subtrees underneath them).
	It is only left to show that $T'$ satisfies the desired property.
	
	To this end, we observe that: \textit{(1)} every fact $\beta$ in $T'$ occurs polynomialy many times on the frozen paths, and \textit{(2)} there are polynomialy many such sibling nodes. Property $(1)$ holds since $\depth{T} \le f(|D|)$ for some polynomial $f$; hence, there are polynomialy many nodes on a path (at most $f(|D|)+1$). Moreover, since we freeze one path per fact of $D'$, there are $|D'|$ paths. Hence, each fact occurs at most $[f(|D|)+1]\times |D'|$ times on the frozen paths. Property $(2)$ holds because each node on a frozen path has at most $b-1$ siblings, where $b$ is the maximal number of atoms occurring in the body of some rule in $\Sigma$. Therefore, the number of sibling nodes is bounded by $[f(|D|)+1]\times |D'|\times (b-1)$. Due to our construction, for every fact $\beta$ that occurs in a subtree $T''$ of some sibling node, it hold that $|\quot{T''[\beta]}|=1$. Therefore, 
	\[|\quot{T'[\beta]}|\ \le\ [f(|D|)+1]\times |D'| \times b \]
	(one equivalence class for each node on the frozen paths, and one equivalence class for each sibling node). The claim then follows with:
	\[g(|D|)\ =\ [f(|D|)+1]\times |D'| \times b.\]
	This concludes our proof.
\end{proof}

Finally, we prove Lemma~\ref{lem:from-trees-to-dags}, which we recall here:

\begin{manuallemma}{\ref{lem:from-trees-to-dags}}
	\lemmafromtreestodags
\end{manuallemma}
\begin{proof}
	Our goal here is to construct a DAG $G=(V,E,\lambda)$ that contains, for every fact $\beta$ that occurs in $T$, and every equivalence class of $\quot{T[\beta]}$, a single DAG representing this class. To this end, we first add, for every fact $\beta$ that occurs in $T$, and every equivalence class $C_{\beta}$ of $\quot{T[\beta]}$, $k$ nodes $v^{C_{\beta}}_1,\dots,v^{C_{\beta}}_k$ to $V$, where $k$ is the maximal number of occurrences of the trees of $C_{\beta}$ under a single node of $T$. Clearly, $k\le b$, where $b$ is the maximal number of atoms occurring in the body of some rule in $\Sigma$. We then define $\lambda(v^{C_{\beta}}_j)=\beta$ for every $j\in [k]$. Note that the tree $T$ itself belongs to some equivalence class $C\in\quot{T[\alpha]}$ (since $\lambda(v)=\alpha$ for the root node $v$ of $T$ by the definition of proof trees), and this tree does not appear under any node of $T$; hence, for this equivalence class we have that $k=0$. In this case, we add a single node $v^{C}$ to $V$.
	
	Next, we add the edges $(v^{C_{\beta}}_j,v^{C'_{\beta'}}_{1}),\dots,(v^{C_{\beta}}_j,v^{C'_{\beta'}}_{p})$ to $E$ if for every tree $T'$ of $C$ with root $v$, we have that $T'[u]\in C'$ for precisely $p$ children $u$ or $v$. (Observe that this either holds for all trees of $C$ or none of them, since $C$ is an equivalence class.) The number of nodes in $V$ is bounded by $|\base{D,\dep}|\times f(|D|)\times b$ as the number of facts that occur in $T$ is bounded by $|\base{D,\dep}|$, there are at most $f(|D|)$ equivalence classes for each fact, and we have at most $b$ nodes for each combination of a fact and its equivalence class; hence, by defining 
	\[
	g(|D|)\ =\ |\base{D,\dep}|\times f(|D|)\times b,
	\]
	$|V|\le g(|D|)$. It remains to show that $G$ is a proof DAG.
	
	Let $v$ be the root of the proof tree $T$. As aforementioned, the subtree of $v$ (which is $T$ itself) belongs to some equivalence class $C$ of $\quot{T[\alpha]}$, and we have a node $v^C$ in $V$ with $\lambda(v^C)=\alpha$. Clearly, there is no other node $u\neq v$ in $T$ such that $T[u]\in C$; hence, by the definition of $E$, the node $v^C$ has no incoming edges. Contrarily, for every other equivalence class $C_\beta$ of some fact $\beta$, if $V$ contains precisely $k$ nodes $v^{C_\beta}_1,\dots,v^{C_\beta}_k$ corresponding to this class, then there exists, by definition, a node $u$ in $T$ with $\lambda(u)=\delta$ for some $\delta$ that has $k$ children $u_1,\dots,u_k$ whose subtrees all belong to $C_\beta$. In this case, $E$ contains the edges  $(v^{C_{\delta}}_1,v^{C_{\beta}}_{1}),\dots,(v^{C_{\delta}}_1,v^{C_{\beta}}_{k})$, where $C_{\delta}$ is the equivalence class of the subtree $T[u]$. We conclude that $G$ has a single node with no incoming edges, and the label of this node is $\alpha$; hence, property $(1)$ of Definition~\ref{def:proof-dag} holds.
	
	We can similarly show that every leaf $v\in V$ corresponds to a leaf $v'$ of $T$ labelled with the same fact. That is, if $v$ is a leaf of $G$, then it is of the form $v^{C_\beta}_j$, where $\beta$ is the label of some leaf of $T$ and $C_\beta$ is an equivalence class that contains trees with a single node labeled with $\beta$. Since $T$ is a proof tree, we have that $\lambda(v')\in D$ for every leaf $v'$ of $T$ and so $\lambda(v)\in D$ for every leaf $v$ of $G$, and property $(2)$ of Definition~\ref{def:proof-dag} holds. 
	
	Finally, we show that property $(3)$ of Definition~\ref{def:proof-dag} is satisfied by $G$. Let $v^{C_\beta}_j$ be a node in $V$ with $n\ge 1$ outgoing edges $(v^{C_\beta}_j,v^{C_{\delta_1}}_{j_1}),\dots,(v^{C_\beta}_j,v^{C_{\delta_n}}_{j_n})$. By the definition of $V$, there exists a node $u$ in $T$ with $\lambda(u)=\beta$ such that $T[u]\in C_\beta$. 
	By the definition of $E$, if among the equivalence classes $C_{\delta_1},\dots,C_{\delta_n}$ there are precisely $p$ occurrences of some class $C$ (corresponding to some fact $\gamma$), then for every tree of $C_\beta$ (in particular, for $T[u]$), the root node of the tree (in particular, the node $u$) has precisely $p$ children $u_1,\dots,u_p$ such that $T[u_\ell]\in C$ for every $\ell\in[p]$. This also means that $\lambda(u_\ell)=\gamma$ for every $\ell\in[p]$ by the definition of $V$. We conclude that there is a node $u$ of $T$ with children $u_1,\dots,u_n$ such that $\lambda(u)=\beta$ and $\lambda(u_i)=\delta_i$ for every $i\in[n]$. 
	Since $T$ is a proof tree, this means that there exists a rule $R_0(\bar x_0)\ \assign\ R_1(\bar x_1),\ldots,R_n(\bar x_n) \in \dep$ and a function $h : \bigcup_{i \in [n]} \bar x_i \ra \ins{C}$ such that $\lambda(v) = R_0(h(\bar x_0))$, and $\lambda(u_i) = R_i(h(\bar x_i))$ for $i \in [n]$.
	Therefore, we also have that $\lambda(v^{C_\beta}_j)=R_0(h(\bar x_0))$ and $\lambda(v^{C_{\delta_i}}_{j_i})=R_i(h(\bar x_i))$ for every $i\in [n]$, and property $(3)$ indeed holds. We conclude that $G$ is a proof DAG of $\alpha$ w.r.t.~$D$ and $\dep$ with $|V|\le g(|D|)$.
\end{proof}

\medskip
\noindent \textbf{Finalize the Proof.} With the above technical lemmas in place, it is now easy to show the NP upper bound. Fix a Datalog query $Q = (\dep,R)$. Given a database $D$ over $\esch{\dep}$, a tuple $\bar t \in \adom{D}^{\arity{R}}$, and a subset $D'$ of $D$, to decide whether $D' \in \why{\bar t}{D}{Q}$ we simply need to check for the existence of a proof tree $T$ of $R(\bar t)$ w.r.t.~$D$ and $\dep$ such that $\support{T} = D'$. By Proposition~\ref{pro:characterization-all-trees}, this is tantamount to the existence of a compact proof DAG $G$ of $R(\bar t)$ w.r.t.~$D$ and $\dep$ with $\support{G} = D'$.
It is clear that the existence of such a proof DAG can be checked by simply guessing a polynomially-sized (w.r.t.~$|D|$) labeled directed graph $G = (V,E,
\lambda)$, and then checking whether $G$ is acyclic, rooted, and a proof DAG of $R(\bar t)$ w.r.t.~$D$ and $\dep$ with $\support{G} = D'$. Since both steps can be carried out in polynomial time, $\mathsf{Why\text {-}Provenance}[Q]$ is in \NP, and thus, $\mathsf{Why\text {-}Provenance[\DAT]}$ is in \NP~in data complexity.


\medskip

\noindent \underline{\textbf{Lower Bound}}
\smallskip

\noindent We proceed to establish that $\mathsf{Why\text {-}Provenance[\LDAT]}$ is \NP-hard in data complexity. To this ends, we need to show that there exists a linear Datalog query $Q$ such that the problem $\mathsf{Why\text {-}Provenance}[Q]$ is \NP-hard. The proof is via a reduction from $\mathsf{3SAT}$, which takes as input a Boolean formula $\varphi = C_1 \wedge \ldots \wedge C_m$ in 3CNF, where each clause has exactly 3 literals (a Boolean variable $v$ or its negation $\neg v$), and asks whether $\varphi$ is satisfiable.

\medskip
\noindent \textbf{The Linear Datalog Query.}
We start by defining the linear Datalog query $Q = (\dep,R)$. If the name of a variable is not important, then we use $\_$ for a fresh variable occurring only once in $\dep$. By abuse of notation, we use semicolons instead of commas in a tuple expression in order to separate terms with a different semantic meaning. The program $\dep$ follows:
\begin{eqnarray*}
\sigma_1 &:& R(x)\,\, \assign\,\, {\rm Var}(x;z,\_),{\rm Assign}(x,z), \\
\sigma_2 &:& R(x)\,\, \assign\,\, {\rm Var}(x;\_,z), {\rm Assign}(x,z),\\
\sigma_3 &:& {\rm Assign}(x,y)\,\, \assign\,\, C(x,y;\_,\_;\_,\_),{\rm Assign}(x,y), \\
\sigma_4 &:& {\rm Assign}(x,y)\,\, \assign\,\, C(\_,\_;x,y;\_,\_),{\rm Assign}(x,y), \\
\sigma_5 &:& {\rm Assign}(x,y)\,\, \assign\,\, C(\_,\_;\_,\_;x,y),{\rm Assign}(x,y), \\
\sigma_6 &:& {\rm Assign}(x,z)\,\, \assign\,\, {\rm Next}(x,y;z,\_),R(y), \\
\sigma_7 &:& {\rm Assign}(x,z)\,\, \assign\,\, {\rm Next}(x,y;\_,z),R(y), \\
\sigma_8 &:& R(x)\,\, \assign\,\, {\rm Last}(x).
\end{eqnarray*}	
It is easy to verify that $\dep$ is indeed a linear Datalog program.
The high-level idea underlying the program $\dep$ is, for each variable $v$ occurring in a given Boolean formula $\varphi$, to non-deterministically assign a value ($0$ or $1$) to $v$, and then check whether the global assignment makes $\varphi$ true.
The rules $\sigma_1$ and $\sigma_2$ are responsible for assigning $0$ or $1$ to a variable $v$; the last two positions of the relation ${\rm Var}$ always store the values $0$ and $1$, respectively.
The rules $\sigma_3$, $\sigma_4$, and $\sigma_5$ are responsible for checking whether an assignment for a certain variable $v$ makes a literal that mentions $v$ in some clause $C$ (and thus, $C$ itself) true.
The rules $\sigma_6$ and $\sigma_7$ are responsible, once we are done with a certain variable $v$, to consider the variable $u$ that comes after $v$; the relation ${\rm Next}$ provides an ordering of the variables in the given 3CNF Boolean formula.
Finally, once all the variables of the formula have been considered, $\sigma_8$ brings us to the last variable, which is a dummy one, that indicates the end of the above process.

\medskip
\noindent \textbf{From $\mathsf{3SAT}$ to  $\mathsf{Why\text {-}Provenance}[Q]$.} We now establish that $\mathsf{Why\text {-}Provenance}[Q]$ is \NP-hard by reducing from $\mathsf{3SAT}$.
Consider a 3CNF Boolean formula $\varphi = C_1 \wedge \cdots \wedge C_m$ with $n$ Boolean variables $v_1,\ldots,v_n$. For a literal $\ell$, we write $\lvar{\ell}$ for the variable occurring in $\ell$, and $\lsign{\ell}$ for the number $1$ (resp., $0$) if $\ell$ is a variable (resp., the negation of a variable).
We define $D_\varphi$ as the database over $\esch{\dep}$
\begin{eqnarray*}
	&& \{{\rm Var}(v_i;0,1) \mid i \in [n]\}\\
	&\cup& \{{\rm Next}(v_i,v_{i+1};0,1) \mid i \in [n-1]\}\\
	&\cup& \{{\rm Next}(v_n,\bullet;0,1), {\rm Last}(\bullet)\}\\
	&\cup& \{C(\lvar{\ell_1},\lsign{\ell_1};\lvar{\ell_2},\lsign{\ell_2};\lvar{\ell_3},\lsign{\ell_3}) \mid \\
	&& \hspace{28mm} (\ell_1 \vee \ell_2 \vee \ell_3) \text{ is a clause of } \varphi \},
\end{eqnarray*}
which essentially stores the clauses of $\varphi$ and provides an ordering of the variables occurring in $\varphi$, with $\bullet$ being a dummy one.
We can show the next lemma, which essentially states that the above construction leads to a correct polynomial-time reduction from $\mathsf{3SAT}$ to $\mathsf{Why\text {-}Provenance}[Q]$:

\begin{lemma}\label{lem:reduction-from-3sat}
		$D_\varphi$ can be constructed in polynomial time in $\varphi$. Furthermore,
		$\varphi$ is satisfiable iff $D_\varphi \in \why{(v_1)}{D_\varphi}{Q}$.
\end{lemma}
\begin{proof}
Clearly, $D_{\varphi}$ is over $\esch{\dep}$, and $D_{\varphi},(v_1)$ can be constructed in polynomial time w.r.t.~$\varphi$. We now show that $\varphi$ is satisfiable if and only if there exists a proof tree $T$ of $R(v_1)$ w.r.t.~$D_{\varphi}$ and $\Sigma$ such that $\support{T} = D_{\varphi}$.

We start with the $(\Rightarrow)$ direction. Assume $\varphi$ is satisfiable via the truth assignment $\mu$. For every variable $v_i$, we denote by $S_{v_i}$ the set of facts of the form $C(v_i,\mu(v_i);\_,\_;\_,\_)$, $C(\_,\_;v_i,\mu(v_i);\_,\_)$, and $C(\_,\_;\_,\_;v_i,\mu(v_i))$. We define the labeled rooted tree $T=(V,E,\lambda)$ where the root $v \in V$ is labeled with $\lambda(v) = R(v_1)$, and inductively:
	\begin{enumerate}
		\item if $v \in V$ is labeled with $\lambda(v) = R(v_i)$, for $i \in [n]$, then $v$ has two children $u_1,u_2$, where 
		$\lambda(u_1)$ is the (only) fact in $D_{\varphi}$ of the form $Var(v_i;0,1)$ and  $\lambda(u_2) = {\rm Assign}(v_i,\mu(v_i))$.
	\item if $v \in V$ is labeled with $\lambda(v) = {\rm Assign}(v_i,\mu(v_i))$ for $i \in [n]$, and $S_{v_i}$ is not empty, then $v$ 
has two children $u_1,u_2$, where $\lambda(u_1)=f$ for some fact $f\in S_{v_i}$ and $\lambda(u_2) = {\rm Assign}(v_i,\mu(v_i))$. We then remove $f$ from $S_{v_i}$. 
\item if $v \in V$ is labeled with $\lambda(v) = {\rm Assign}(v_i,\mu(v_i))$ for $i \in [n-1]$, and $S_{v_i}$ is empty, then $v$ 
has two children $u_1,u_2$, where $\lambda(u_1)={\rm Next}(v_i,v_{i+1};0,1)$ and $\lambda(u_2) = R(v_{i+1})$.
\item if $v \in V$ is labeled with $\lambda(v) = {\rm Assign}(v_n,\mu(v_n))$ and $S_{v_n}$ is empty, then $v$
has two children $u_1,u_2$, where $\lambda(u_1)={\rm Next}(v_n,\bullet;0,1)$ and and $\lambda(u_2) = R(\bullet)$.
\item if $v \in V$ is labeled with $\lambda(v) = R(\bullet)$, then $v$ has one child $u_1$, where $\lambda(u_1)={\rm Last}(\bullet)$.
	\end{enumerate}
	One can verify that $T$ is indeed a proof tree of $R(v_1)$ w.r.t.~$D_\varphi$ and $\Sigma$. In particular, the root is labeled with $R(v_1)$ by construction, and it is easy to see that the labels of the leaves all occur in $D_\varphi$. Furthermore, since $\dep$ is linear, at each level of $T$ there exists at most one non-leaf node. The edges from this node are defined in items~$(1)-(5)$ above. The edges defined in item~(1) are obtained by considering rules $\sigma_1,\sigma_2$, the edges defined in item~(2) are obtained by considering rules $\sigma_3,\sigma_4,\sigma_5$, and the edges defined in items~(3) and~(4) are obtained by considering rules $\sigma_6,\sigma_7$. Finally, the edge defined in item~(5) is obtained by considering rule $\sigma_8$. Hence, Property~$(3)$ of Definition~\ref{def:proof-tree} holds.
 
 Regarding the support of $T$, since $\mu$ is a satisfying assignment, every clause of $\varphi$ contains at least one literal $\ell_j$ such that $\mu(\lvar{\ell_j})=\lsign{\ell_j}$. Item~(1) ensures that we follow this satisfying assignment (i.e., add a node labeled with ${\rm Assign}(v_i,\mu(v_i))$ under a node $R(v_i)$ corresponding to the variable $v_i$). Item~(2) ensures that whenever we consider a variable $v_i$, we touch every atom over the predicate $C$ corresponding to a clause that contains the variable $v_i$ with the correct sign (positive if $\mu(v_i)=1$ or negated if $\mu(v_i)=0$); that is, every clause satisfied by the assignment to this variable. Items~(3) and~(4) ensure that we go over all the variables, as the atoms over the predicate ${\rm Next}$ only allow us to move from a variable $v_i$ to the next variable $v_{i+1}$. This, in turn, ensures that we touch every atom over the predicate ${\rm Var}$ by item~(1), every atom over the predicate $C$ by item~(2) (since, as aforementioned, each such clause contains a literal with a sign that is consistent with the assignment to the corresponding variable), and every atom over the predicate ${\rm Next}$ by items~(3) and~(4). Finally, item~(5) ensures that we touch the atom ${\rm Last}(\bullet)$. Hence, we conclude that $\support{T} = D_{\varphi}$.

Next, we prove the $(\Leftarrow)$ direction. Assume that $T=(V,E,\lambda)$ is a proof tree for $R(v_1)$ w.r.t.~$D_\varphi$ and $\Sigma$ such that $\support{T}=D_{\varphi}$. 
	Note that by linearity of $\dep$, at each level of $T$, besides the last one, there exists precisely one non-leaf node. These nodes thus form a path $u_0,u_1,\ldots$ in $T$. Note that $u_0$, i.e., the root, is necessarily labeled with the fact $R(v_1)$ by the definition of $Q$. Then, the children of the root are labeled with ${\rm Var}(v_1;0,1)$ and ${\rm Assign}(v_1,b_1)$ (this is the node $u_1$) for some $b_1\in\{0,1\}$, based on rule $\sigma_1$ or $\sigma_2$, as these are the only options. At this point, either the node $u_2$ is labeled with $R(v_2)$ based on rule $\sigma_6$ or $\sigma_7$, or $u_2,\dots,u_j$ for some $j\ge 2$ are all labeled with ${\rm Assign}(v_1,b_1)$, and then $u_{j+1}$ is labeled with $R(v_2)$. Then, the same reasoning applies to $v_2$ and all the variables that come next.

    Since the atoms over the predicate ${\rm Next}$ only consider consecutive variables, it is only possible to move from $v_i$ to $v_{i+1}$ along the path, and since $\support{T}=D_\varphi$, the set $\support{T}$ contains all the atoms over ${\rm Next}$, and we conclude that we go over all variables. Moreover, for every variable $v_i$, once we select the label ${\rm Assign}(v_i,b_i)$ for the child of the node labeled with $R(v_i)$, there is no rule that allows us to obtain the label ${\rm Assign}(v_i,1-b_i)$, and so the labels over the predicate ${\rm Assign}$ correspond to a truth assignment $\mu$ to the variables $v_1,\dots,v_n$. Finally, since all the atoms over the predicate $C$ appear in the support, the rules $\sigma_3,\sigma_4,\sigma_5$ imply that for every clause $(\ell_1\vee\ell_2\vee\ell_3)$, either ${\rm Assign}(\lvar{\ell_1},\lsign{\ell_1})$, or ${\rm Assign}(\lvar{\ell_2},\lsign{\ell_2})$, or ${\rm ssign}(\lvar{\ell_3},\lsign{\ell_3})$ appears as a label along the path. Hence, $\mu$ is a satisfying truth assignment.
\end{proof}

By Lemma~\ref{lem:reduction-from-3sat}, $\mathsf{Why\text {-}Provenance}[Q]$ is \NP-hard, and thus, $\mathsf{Why\text {-}Provenance[\LDAT]}$ is \NP-hard in data complexity.

\subsection{Non-Recursive Queries}

We now focus on non-recursive Datalog queries, and give the full proof of Theorem~\ref{the:non-recursive-complexity}, which we recall here:

\begin{manualtheorem}{\ref{the:non-recursive-complexity}}
	\thenonrecursivecomplexity
\end{manualtheorem}

Given a non-recursive Datalog query $Q$, we have already explained in the main body how the FO query $Q_{\mi{FO}}$ is constructed. Our main task here is to establish the correctness of this construction, i.e., Lemma~\ref{lem:fo-tree-equiv}, which we recall here:

\begin{manuallemma}{\ref{lem:fo-tree-equiv}}
	\lemmaforewriting
\end{manuallemma}

\begin{proof}
	We first discuss the $(\Rightarrow)$ direction. There is a proof tree $T$ of $R(\bar t)$ w.r.t.~$D$ and $\dep$ with $\support{T} = D'$. Thus, $T$ is a $Q$-tree, the set $\cq{Q}$ contains the CQ $\cq{T}$, that is, the CQ induced by $T$, and the set $\cqeq{Q}$ contains a CQ $\varphi(\bar y)$ that is the same as $\cq{T}$ up to variable renaming. It is then an easy exercise to show that that $D'$ satisfies the sentence $\psi_{\varphi(\bar y)}[\bar x / \bar t]$. This in turn implies that $\bar t \in Q_{\mi{FO}}(D')$.
	
	We now discuss the $(\Leftarrow)$ direction. There is a CQ $\varphi(\bar y) \in \cqeq{Q}$ such that $D'$ satisfies the sentence $\psi_{\varphi(\bar y)}[\bar x / \bar t]$. It is then an easy exercise to show that there exists a proof tree of $R(\bar t)$ w.r.t.~$D$ and $\dep$ with $\support{T} = D'$. This in turn implies that $D \in \why{R(\bar t)}{D}{Q}$.
\end{proof}
\section{Non-Recursive Proof Trees}\label{appsec:refined-trees}

As discussed in the main body of the paper (see Section~\ref{sec:refined-trees}), the standard notion of why-provenance, which was defined in Section~\ref{sec:why-provenance} and thoroughly analyzed in Section~\ref{sec:complexity}, relies on arbitrary proof trees without any restriction. Indeed, a subset of the input database $D$ belongs to the why-provenance of a tuple $\bar t$ w.r.t.~$D$ and a Datalog query $Q = (\dep,R)$ as long as it is the support of {\em any} proof tree of $R(\bar t)$ w.r.t.~$D$ and $\dep$.
However, as already discussed in the literature (see, e.g., the recent work~\cite{BBPT22}), there are proof trees that are counterintuitive. Such a proof tree is the second one in Example~\ref{exa:proof-tree} as the fact $A(a)$ is derived from itself.
Now, a member $D'$ of $\why{\bar t}{D}{Q}$, witnessed via such an unnatural proof tree, might be classified as a counterintuitive explanation of $\bar t$ as it does not correspond to an intuitive derivation process, which can be extracted from the proof tree, that derives from $D'$ the fact $R(\bar t)$.
This leads to the need of considering refined classes of proof trees that overcome the conceptual limitations of arbitrary proof trees, which in turn lead to conceptually intuitive explanations.
In this section, we focus on the class of non-recursive proof trees. Roughly, a non-recursive proof tree is a proof
tree that does not contain two nodes labeled with the same fact and such that one is the descendant of the other, which reflects the above discussion that using a fact to derive itself is a counterintuitive phenomenon. The formal definition follows:

\begin{definition}[\textbf{Non-Recursive Proof Tree}]\label{def:non-recursove-proof-tree}
	Consider a Datalog program $\dep$, a database $D$ over $\esch{\dep}$, and a fact $\alpha$ over $\sch{\dep}$. A {\em non-recursive proof tree of $\alpha$ w.r.t.~$D$ and $\dep$} is a proof tree $T = (V,E,\lambda)$ of $\alpha$ w.r.t.~$D$ and $\dep$ such that, for every two nodes $v,u \in V$, if there is a path from $v$ to $u$ in $T$, then $\lambda(v) \neq \lambda(u)$. \hfill\markfull
\end{definition}

We now define why-provenance relative to non-recursive proof trees. Given a Datalog query $Q = (\dep,R)$, a database $D$ over $\esch{\dep}$, and a tuple $\bar t \in \adom{D}^{\arity{R}}$, the {\em why-provenance of $\bar t$ w.r.t.~$D$ and $Q$ relative to non-recursive proof trees} is defined as the family of sets of facts
\begin{multline*}
\{\support{T} \mid T \text{ is a non-recursive proof tree of }\\
R(\bar t) \text{ w.r.t. } D \text{ and } \dep\}
\end{multline*}
denoted $\nrwhy{\bar t}{D}{Q}$.
Then, the algorithmic problems
\[
\mathsf{Why\text {-}Provenance_{NR}[C]} \quad \text{and} \quad  \mathsf{Why\text {-}Provenance_{NR}}[Q]
\] 
are defined in the exact same way as those in Section~\ref{sec:why-provenance} with the key difference that $\nrwhy{\bar t}{D}{Q}$ is used instead of $\why{\bar t}{D}{Q}$, i.e., the question is whether the given subset of the database belongs to $\nrwhy{\bar t}{D}{Q}$. 
We proceed to study the data complexity of $\mathsf{Why\text {-}Provenance_{NR}[C]}$ for each class $\class{C} \in \{\DAT,\LDAT,\NRDAT\}$. As shown in the case of arbitrary proof trees, for recursive queries, even if the recursion is restricted to be linear, the problem is in general intractable, whereas for non-recursive queries it is highly tractable. We first focus on recursive queries.

\subsection{Recursive Queries}

We show the following complexity result:

\begin{theorem}\label{the:complexity-non-recursive-proof-trees-np}
	$\mathsf{Why\text {-}Provenance_{NR}[C]}$ is \NP-complete in data complexity, for each class $\class{C} \in \{\DAT,\LDAT\}$.
\end{theorem}

To prove Theorem~\ref{the:complexity-non-recursive-proof-trees-np}, it suffices to show that:
\begin{itemize}
	\item $\mathsf{Why\text {-}Provenance_{NR}[\DAT]}$ is in \NP~in data complexity.
	\item $\mathsf{Why\text {-}Provenance_{NR}[\LDAT]}$ is \NP-hard in data complexity.
\end{itemize}
Let us first focus on the upper bound.

\medskip

\noindent \underline{\textbf{Upper Bound}}
\smallskip

\noindent The proof is similar to the proof of the analogous result for $\mathsf{Why\text {-}Provenance[\DAT]}$ established in Section~\ref{sec:why-provenance}.
Given a Datalog program $\dep$, a database $D$ over $\esch{\dep}$, and a fact $\alpha$ over $\sch{\dep}$, we first define the notion of {\em non-recursive proof DAG of $\alpha$ w.r.t.~$D$ and $\dep$}.
We then proceed to establish a result analogous to Proposition~\ref{pro:characterization-all-trees}: the existence of a non-recursive proof tree of $\alpha$ w.r.t.~$D$ and $\dep$ with $\support{T} = D' \subseteq D$ is equivalent to the existence of a polynomially-sized non-recursive proof DAG $G$ of $\alpha$ w.r.t.~$D$ and $\dep$ with $\support{G} = D'$. This in turn leads to a guess-and-check algorithm that runs in polynomial time.
Let us formalize the above high-level description.

\begin{definition}[\textbf{Non-Recursive Proof DAG}]\label{def:non-recursive-proof-dag}
	Consider a Datalog program $\dep$, a database $D$ over $\esch{\dep}$, and a fact $\alpha$ over $\sch{\dep}$. 
	A {\em non-recursive proof DAG of $\alpha$ w.r.t.~$D$ and $\dep$} is a proof DAG $G=(V,E,\lambda)$ of $\alpha$ w.r.t.~$D$ and $\dep$ such that, for every two nodes $v,u \in V$, if there is a path from $v$ to $u$ in $G$, then $\lambda(v) \neq \lambda(u)$. \hfill\markfull
\end{definition}

The analogous result to Proposition~\ref{pro:characterization-all-trees} follows:

\begin{proposition}\label{pro:characterization-nr-trees}
		For a Datalog program $\dep$, there is a polynomial $f$ such that, for every database $D$ over $\esch{\dep}$, fact $\alpha$ over $\sch{\dep}$, and $D' \subseteq D$, the following are equivalent:
	\begin{enumerate}
		\item There is a non-recursive proof tree $T$ of $\alpha$ w.r.t.~$D$ and $\dep$ such that $\support{T} = D'$.
		\item There is a non-recursive proof DAG $G = (V,E,\lambda)$ of $\alpha$ w.r.t.~$D$ and $\dep$ with $\support{G} = D'$ and $|V| \leq f(|D|)$.
	\end{enumerate}
\end{proposition}

The direction (2) implies (1) is shown by ``unravelling'' the non-recursive proof DAG $G$ into a non-recursive proof $T$ with $\support{G} = \support{T}$. We use the same ``unravelling'' construction as in the proof of direction (2) implies (1) of Proposition~\ref{pro:characterization-all-trees}, which {\em preserves non-recursiveness}.

We now proceed with (1) implies (2). The underlying construction proceeds in two main steps captured by Lemmas~\ref{lem:scount-reduction-nr} and~\ref{lem:from-trees-to-dags-nr} given below.


\medskip 

$\bullet$ The \textbf{\textit{first step}} is to show that a non-recursive proof tree $T$ of $\alpha$ w.r.t.~$D$ and $\dep$ with $\support{T} = D'$ can be converted into a non-recursive proof tree $T'$ of $\alpha$ w.r.t.~$D$ and $\dep$ with $\support{T'} = D'$ that has ``small'' subtree count.

\begin{lemma}\label{lem:scount-reduction-nr}
		For each Datalog program $\dep$, there is a polynomial $f$ such that, for every database $D$ over $\esch{\dep}$, fact $\alpha$ over $\sch{\dep}$, and $D' \subseteq D$, if there is a non-recursive proof tree $T$ of $\alpha$ w.r.t.~$D$ and $\dep$ with $\support{T} = D'$, then there is also such a proof tree $T'$ with $\scount{T'} \leq f(|D|)$.
\end{lemma}

\begin{proof}
	We first observe that the non-recursive proof tree $T$, by definition, has ``small'' depth. In particular, since no two nodes on a path of $T$ have the same label, the length of a path is bounded by the number of labels, that is, $|\base{D,\dep}|$, which is clearly polynomial in the size of the database $D$.
	The other crucial observation is that the construction underlying Lemma~\ref{lem:scount-reduction}, which converts a proof tree of ``small'' depth into a proof tree of ``small`` subtree count with the same support {\em preserves non-recursiveness}.
	Consequently, we can apply the construction underlying Lemma~\ref{lem:scount-reduction} to the non-recursive proof tree $T$ and get a non-recursive proof tree $T'$ with $\support{T} = \support{T'}$ such that $\scount{T'} \leq f(|D|)$, where $f$ is the polynomial provided by Lemma~\ref{lem:scount-reduction}.
\end{proof}

\smallskip

$\bullet$ The \textbf{\textit{second step}} shows that a non-recursive proof tree $T$ of $\alpha$ w.r.t.~$D$ and $\dep$ with $\support{T} = D'$ of ``small'' subtree count can be converted into a compact non-recursive proof DAG $G$ of $\alpha$ w.r.t.~$D$ and $\dep$ with $\support{G} = D'$.

\begin{lemma}\label{lem:from-trees-to-dags-nr}
	For each Datalog program $\dep$ and a polynomial $f$, there is a polynomial $g$ such that, for every database $D$ over $\esch{\dep}$, fact $\alpha$, and $D' \subseteq D$, if there is a non-recursive proof tree $T$ of $\alpha$ w.r.t.~$D$ and $\dep$ with $\support{T} = D'$ and $\scount{T} \leq f(|D|)$, then there is a non-recursive proof DAG $G = (V,E,\lambda)$ of $\alpha$ w.r.t.~$D$ and $\dep$ with $\support{G} = D'$ and $|V| \leq g(|D|)$.
\end{lemma}

\begin{proof}
	We employ the construction underlying Lemma~\ref{lem:from-trees-to-dags}, which converts a proof tree of ``small'' subtree count into a non-recursive proof DAG of polynomial size with the same support, since it {\em preserves non-recursiveness}. The latter holds since, for each path of the proof tree, there is a path in the proof DAG with the same labels, and vice versa. 
\end{proof}

It is now clear that the direction (1) implies (2) of Proposition~\ref{pro:characterization-nr-trees} is an immediate consequence of Lemmas~\ref{lem:scount-reduction-nr} and~\ref{lem:from-trees-to-dags-nr}.

\medskip
\noindent
\textbf{Finalize the Proof.} We can now finalize the proof of the claim that $\mathsf{Why\text {-}Provenance_{NR}[\DAT]}$ is in \NP~in data complexity.
Fix a Datalog query $Q = (\dep,R)$. Given a database $D$ over $\esch{\dep}$, a tuple $\bar t \in \adom{D}^{\arity{R}}$, and a subset $D'$ of $D$, to decide whether $D' \in \nrwhy{\bar t}{D}{Q}$ we simply need to check for the existence of a non-recursive proof tree $T$ of $R(\bar t)$ w.r.t.~$D$ and $\dep$ such that $\support{T} = D'$. By Proposition~\ref{pro:characterization-nr-trees}, this is tantamount to the existence of a polynomially-sized non-recursive proof DAG $G$ of $R(\bar t)$ w.r.t.~$D$ and $\dep$ with $\support{G} = D'$.
The existence of such a non-recursive proof DAG can be checked via a non-deterministic algorithm that runs in polynomial time in the size of the database as it was done for proving that $\mathsf{Why\text {-}Provenance[\DAT]}$ is in \NP~in data complexity (Theorem~\ref{the:recursive-complexity}). The only difference is that now we need to additionally check that the guessed DAG is also non-recursive.
This can be done by going over the nodes of the
DAG using depth-first search, and remembering all the labels that we have seen on the current path. Whenever we encounter a new node, we need to verify that its label does not belong to the set of labels that we have already seen on this path. If it does, then we reject; otherwise, we add its label to the set of labels and continue. Whenever we go ``back'' in the DAG (move from a node $u$ to a node $v$ when there is an edge $(v,u)$), we remove the label of the child node from the set. This requires $O(|V| \cdot |E|)$ time, where $V$ and $E$ are the sets of nodes and edges of the DAG, respectively.
Consequently, $\mathsf{Why\text {-}Provenance}[Q]$ is in \NP, and thus, $\mathsf{Why\text {-}Provenance[\DAT]}$ is in \NP~in data complexity.

\medskip

\noindent \underline{\textbf{Lower Bound}}
\smallskip

\noindent We proceed to establish that $\mathsf{Why\text {-}Provenance_{NR}[\LDAT]}$ is \NP-hard in data complexity. To this end, we need to show that there exists a linear Datalog query $Q$ such that the problem $\mathsf{Why\text {-}Provenance_{NR}}[Q]$ is \NP-hard. The proof is via a reduction from the problem $\mathsf{Ham\text{-}Cycle}$, which takes as input a directed graph $G = (V,E)$ and asks whether $G$ has a Hamiltonian cycle, i.e., whether there exists a cycle $v_1,\ldots,v_n,v_1$ in $G$ such that, for distinct integers $i,j \in [n]$, $v_i \neq v_j$, and $V = \{v_1,\ldots,v_n\}$.

\medskip
\noindent \textbf{The Linear Datalog Query.}
We start by defining the linear Datalog query $Q = (\dep,{\rm Path})$. If the name of a variable is not important, then we use $\_$ for a fresh variable occurring only once in $\dep$. By abuse of notation, we use semicolons instead of commas in a tuple expression in order to separate terms with a different meaning. The program $\dep$ follows:
\begin{eqnarray*}
\sigma_1 &:& {\rm MarkedE}(x)\,\, \assign \,\, {\rm First}(x)\\
\sigma_2 &:& {\rm MarkedE}(y)\,\, \assign \,\, E(\_,\_;x,y;\_), {\rm MarkedE}(x),\\
\sigma_3 &:& {\rm Path}(y)\,\, \assign \,\, E(x,y;\_,\_;z), {\rm MarkedE}(z), N(x), \\
\sigma_4 &:& {\rm Path}(y)\,\, \assign \,\, E(x,y;\_,\_;\_), {\rm Path}(x), N(x).
\end{eqnarray*}
It is easy to verify that $\dep$ is indeed a linear Datalog program.

\medskip
\noindent \textbf{From $\mathsf{Ham\text{-}Cycle}$ to  $\mathsf{Why\text {-}Provenance_{NR}}[Q]$.} We now show that $\mathsf{Why\text {-}Provenance}[Q]$ is \NP-hard by reducing from $\mathsf{Ham\text{-}Cycle}$.
Consider a directed graph $G = (V,E)$ with $E = \{e_1,\ldots,e_m\}$.
Let $D_G$ be the database over $\esch{\dep}$
\begin{multline*}
\{{\rm First}(1)\} \cup \{N(v) \mid v \in V\}\ \cup \\
\{E(u,v;i,i+1;m+1) \mid i \in [m] \text{ and } e_i = (u,v)\},
\end{multline*}
which essentially stores the graph $G$ and provides an ordering of its edges.
We now show the next lemma, which states that the above construction leads to a correct polynomial-time reduction from $\mathsf{Ham\text{-}Cycle}$  to $\mathsf{Why\text {-}Provenance_{NR}}[Q]$.

\begin{lemma}\label{lem:reduction-from-hasm-cycle}
	The following hold:
	\begin{enumerate}
	\item $D_G$ can be constructed in polynomial time w.r.t.~$G$.
	\item $G$ has a Hamiltonian cycle iff $D_G \in \nrwhy{(v^*)}{D_G}{Q}$ for some arbitrary node $v^* \in V$.
	\end{enumerate}
\end{lemma}

\begin{proof}
	It is straightforward to see that $D_G$ can be constructed in polynomial time in the size of $G$. We proceed to establish item (2). We start with the direction $(\Rightarrow )$.
	%
	%
	Assume that $G$ has a Hamiltonian cycle $v_1,\ldots,v_n,v_{n+1}$, and w.l.o.g.~let $v_1 = v_{n+1} = v^*$. Hence, there exist edges $(v_1,v_2),\ldots,(v_{n-1},v_n),(v_n,v_{n+1})$ in $G$. We define the labeled rooted tree $T=(V',E',\lambda)$: for the root $v \in V'$, let $\lambda(v) = {\rm Path}(v_{n+1}) = {\rm Path}(v^*)$, and inductively:
	\begin{enumerate}
		\item If $v \in V'$ is such that $\lambda(v) = {\rm Path}(v_{i+1})$, for $i \in \{2,\ldots,n\}$, then $v$ has 3 children $u_1,u_2,u_3$, where 
		$\lambda(u_1)$ is the (only) fact in $D_{G}$ of the form $E(v_i,v_{i+1};\cdot,\cdot;\cdot)$, $\lambda(u_2) = {\rm Path}(v_i)$, and $\lambda(u_3) = N(v_i)$.
		
		\item if $v \in V'$ is such that $\lambda(v) = {\rm Path}(v_2)$, $v$ has 3 children $u_1,u_2,u_3$, where $\lambda(u_1)$ is the (only) fact in $D_G$ of the form $E(v_1,v_2;\cdot,\cdot;m+1)$, $\lambda(u_2) = {\rm MarkedE}(m+1)$, and $\lambda(u_3) = N(v_1)$.
		
		\item if $v \in V'$ is such that $\lambda(v) = {\rm MarkedE}(i+1)$, for $i \in \{1,\ldots,m\}$, then $v$ has 2 children $u_1,u_2$, where $\lambda(u_1)$ is the (only) fact in $D_G$ of the form $E(\cdot,\cdot;i,i+1;\cdot)$, and $\lambda(u_2) = {\rm MarkedE}(i)$.
		
		\item if $v \in V'$ is such that $\lambda(v) = {\rm MarkedE}(1)$, then $v$ has only one child $u$, where $\lambda(u) = {\rm First}(1)$.
	\end{enumerate}
	This completes the construction of $T$.
	We proceed to show that $T$ is a non-recursive proof tree of ${\rm Path}(v^*)$ w.r.t.~$D_G$ and $Q$ such that $\support{T} = D_G$, which in turn implies that $D_G \in \nrwhy{(v^*)}{D_G}{Q}$, as needed.
	Since $\dep$ is linear, at each level of $T$ there exists at most one non-leaf node. Moreover, the edges in $T$ going from level $i$ to level $i+1$, with $i \in \{0,\ldots,n-1\}$, and the labels of the corresponding nodes, as defined in item~(1), are valid since they are obtained by considering rule $\sigma_4$ and there exist edges $(v_i,v_{i+1}) \in E$, for $i \in \{2,\ldots,n\}$. The edges in $T$ from level $n-1$ to $n$ are obtained considering $\sigma_3$, and the labels are again valid since $(v_1,v_2) \in E$. Then, for all the remaining levels, except the last one, we consider $\sigma_2$ in item~(3), and for the last level we consider $\sigma_1$, in item~(4). In these last two cases, it is easy to verify that the labels are valid.
	Crucially, by construction of $T$, for every non-leaf node there is no other node in $T$ with the same label, and thus, the proof tree is trivially non-recursive. Regarding the support, by item~(1), $N(v_i) \in \support{T}$, for $i \in \{2,\ldots,n\}$. By item~(2), $N(v_1) \in \support{T}$. Moreover, by item~(3), $E(u,v;i,i+1;m+1) \in \support{T}$, where $(u,v) = e_i$, for $i \in [m]$. Finally, by item~(4), ${\rm First}(1) \in \support{T}$. Therefore, $\support{T}=D_G$, and the claim follows.
	
	
	We now proceed with the direction $(\Leftarrow)$. Assume that $D_G \in \nrwhy{(v^*)}{D_G}{Q}$, which in turn implies that there is a non-recursive proof tree $T=(V,E,\lambda)$ of $(v^*)$ w.r.t.~$D_G$ and $Q$ such that $\support{T}=D_{G}$. Let $n$ be the number of nodes in $G$, and assume w.l.o.g.\ that $n \ge 2$.
	Note that by linearity of $\dep$, at each level of $T$, besides the last one, there exists precisely one non-leaf node. Let $u_i$ be the non-leaf node at level $i$. Clearly, $u_0,u_1,\ldots$ is a path in $T$. Note that, by the definition of $Q$, $u_0$, i.e., the root, is necessarily labeled with a fact using the predicate ${\rm Path}$; this is also the case for $u_i$, at level $i \in \{1,\ldots,n-1\}$. Indeed, assume, towards a contradiction, that $i \in \{1,\ldots,n-1\}$ is the first level in $T$ where $u_i$ is \emph{not} labeled with a fact using the predicate ${\rm Path}$. Since $\sigma_3$ and $\sigma_4$ are the only rules in $\dep$ with a body-atom using the predicate $N$, and since these are the only rules where ${\rm Path}$ appears in the head, $\support{T}$ contains no more than $i < n$ facts using the predicate $N$, and thus, $\support{T} \neq D_G$, which is a contradiction.
	
	Since $u_0,\ldots,u_{n-1}$ is a path in $T$, and such nodes are all labeled with a fact using ${\rm Path}$, we conclude that such labels are obtained by using $\sigma_4$, as it is the only rule having the predicate ${\rm Path}$ both in its body and its head. Hence, because of the atom $E(x,y;\_,\_;\_)$ in $\body{\sigma_4}$, which uses an extensional predicate of $\dep$, and from the fact that $T$ is a proof tree, $\lambda(u_i) = {\rm Path}(v_i)$, where $v_i$ is some node of the graph $G$. Therefore, since $T$ is non-recursive, ${\rm Path}(v_0),{\rm Path}(v_1),\ldots,{\rm Path}(v_{n-1})$ are all distinct, and thus, $v_0,v_1,\ldots,v_{n-1}$ are distinct. Hence, again from the fact that the atom $E(x,y;\_,\_;\_)$ appears in $\body{\sigma_4}$ with an extensional predicate, the fact that $T$ is a proof tree, and by the construction of $D_G$, we conclude that $(v_{n-1},v_{n-2}),\ldots,(v_1,v_0)$ are all the edges of the graph $G$. Therefore, $v_0,\ldots,v_{n-1}$ is the reverse of a Hamiltonian path in $G$. 
	Moreover, let $S_i \subseteq \support{T}$ be the set of facts using the predicate $N$ that label nodes of $T$ up to level $i$. Because of the atom $N(x)$ in the body of $\sigma_4$, we conclude that $S_{n-1} = \{N(v_1),\ldots,N(v_{n-1})\}$, that is, $S_{n-1}$ contains all nodes of the graph $G$, except for $v_0$.
	
	Let us focus now on the node $u_{n-1}$ of $T$. Recall that $\lambda(u_{n-1}) = {\rm Path}(v_{n-1})$. Note that the children of $u_{n-1}$ cannot be labeled using $\sigma_4$ anymore, since it contains the two body atoms ${\rm Path}(x),N(x)$, where $N$ is an extensional predicate. Hence, since $v_0,\ldots,v_{n-1}$ are precisely all the nodes of $G$, one of the children of $u_{n-1}$ would necessarily be labeled with a fact ${\rm Path}(v)$, where $v$ will necessarily coincide with some of the nodes in $v_0,\ldots,v_{n-1}$, and thus, $T$ would not be non-recursive. Hence, the only rule left is $\sigma_3$. As already discussed, up to level $n-1$, the set of facts in $\support{T}$ with predicate $N$ is $S_{n-1}=\{N(v_1),\ldots,N(v_{n-1})\}$. Hence, for $\support{T}$ to also contain ${\rm Path}(v_0) \in D_G$, there must be at least one more node below level $n-1$ in $T$ that is labeled with $N(v_0)$. Since $\sigma_3$ has no body atom with predicate ${\rm Path}$, necessarily one of the children of $u_{n-1}$ is labeled with $N(v_0)$. Hence, thanks to the atom $E(x,y;\_,\_;z)$ in the body of $\sigma_4$, we conclude that there is also an edge $(v_0,v_{n-1})$ in $G$, and thus, $v_0,v_1,\ldots,v_{n-1},v_0$ is the reverse of a Hamiltonian cycle of $G$, and the claim follows.
\end{proof}

By Lemma~\ref{lem:reduction-from-hasm-cycle}, $\mathsf{Why\text {-}Provenance_{NR}}[Q]$ is \NP-hard. Thus, $\mathsf{Why\text {-}Provenance_{NR}[\LDAT]}$ is \NP-hard in data complexity.

\subsection{Non-Recursive Queries}

We now focus on non-recursive Datalog queries, and show the following about the data complexity of why-provenance relative to non-recursive proof trees:

\begin{theorem}\label{the:non-recursive-complexity-nr}
	$\mathsf{Why\text {-}Provenance_{NR}[\NRDAT]}$ is in $\ACZ$ in data complexity.
\end{theorem}

\begin{proof}
This is shown via first-order rewritability as done for Theorem~\ref{the:non-recursive-complexity}. In fact, the construction of the target FO query is exactly the same as in the proof of Theorem~\ref{the:non-recursive-complexity} with the key difference that, for a Datalog query $Q$, the set of CQs $\cq{Q}$ is defined by considering only non-recursive proof trees, i.e., is the set $\{\cq{T} \mid T \text{ is a {\em non-recursive} $Q$-tree}\}$.
\end{proof}


\section{Minimal-Depth Proof Trees}\label{sec:minimal-depth-proof-trees}

We now focus on another refined class of proof trees that has been considered in the literature. Recall that the depth of a rooted tree $T$, denoted $\depth{T}$, is the length of the longest path from its root to a leaf node. Given a Datalog program $\dep$, a database $D$ over $\esch{\dep}$, and a fact $\alpha \in \dep(D)$, let $\mtd{\alpha}{D}{\dep}$ be the integer
\[
\min\{\depth{T} \mid T \text{ is a proof tree of } \alpha \text{ w.r.t. } D \text{ and } \dep\},
\]
i.e., the minimal depth over all proof trees of $\alpha$ w.r.t.~$D$ and $\dep$. The notion of minimal-depth proof tree follows:

\begin{definition}[\textbf{Minimal-Depth Proof Tree}]\label{def:min-depth-proof-tree}
	Consider a Datalog program $\dep$, a database $D$ over $\esch{\dep}$, and a fact $\alpha$ over $\sch{\dep}$. A {\em minimal-depth proof tree of $\alpha$ w.r.t.~$D$ and $\dep$} is a proof tree $T$ of $\alpha$ w.r.t.~$D$ and $\dep$ such that $\depth{T}$ coincides with $\mtd{\alpha}{D}{\dep}$. \hfill\markfull
\end{definition}

Why-provenance relative to minimal-depth proof trees is defined as expected. Given a Datalog query $Q = (\dep,R)$, a database $D$ over $\esch{\dep}$, and a tuple $\bar t \in \adom{D}^{\arity{R}}$, the {\em why-provenance of $\bar t$ w.r.t.~$D$ and $Q$ relative to minimal-depth proof trees} is defined as the family of sets of facts
\begin{multline*}
\{\support{T} \mid T \text{ is a minimal-depth proof tree of }\\
R(\bar t) \text{ w.r.t. } D \text{ and } \dep\}
\end{multline*}
denoted $\mdwhy{\bar t}{D}{Q}$.
Then, the algorithmic problems
\[
\mathsf{Why\text {-}Provenance_{MD}[C]} \quad \text{and} \quad  \mathsf{Why\text {-}Provenance_{MD}}[Q]
\] 
are defined as expected. We proceed to study the data complexity of $\mathsf{Why\text {-}Provenance_{MD}[C]}$ for each class $\class{C} \in \{\DAT,\LDAT,\NRDAT\}$. As shown in the case of arbitrary and non-recursive proof trees, for recursive queries, even if the recursion is restricted to be linear, the problem is in general intractable, whereas for non-recursive queries it is highly tractable. We first focus on recursive queries.

\subsection{Recursive Queries}

We show the following complexity result:

\begin{theorem}\label{the:complexity-minimal-depth-proof-trees-np}
	$\mathsf{Why\text {-}Provenance_{MD}[C]}$ is \NP-complete in data complexity, for each class $\class{C} \in \{\DAT,\LDAT\}$.
\end{theorem}

To prove Theorem~\ref{the:complexity-non-recursive-proof-trees-np}, it suffices to show that:
\begin{itemize}
	\item $\mathsf{Why\text {-}Provenance_{MD}[\DAT]}$ is in \NP~in data complexity.
	\item $\mathsf{Why\text {-}Provenance_{MD}[\LDAT]}$ is \NP-hard in data compl.
\end{itemize}
Let us first focus on the upper bound.

\medskip

\noindent \underline{\textbf{Upper Bound}}
\smallskip

\noindent The proof is similar to the proof of the analogous result for $\mathsf{Why\text {-}Provenance_{NR}[\DAT]}$ (see Theorem~\ref{the:complexity-non-recursive-proof-trees-np}). 
Given a Datalog program $\dep$, a database $D$ over $\esch{\dep}$, and a fact $\alpha$ over $\sch{\dep}$, we first define the notion of {\em minimal-depth proof DAG of $\alpha$ w.r.t.~$D$ and $\dep$}.
We then establish a result analogous to Proposition~\ref{pro:characterization-nr-trees}: the existence of a minimal-depth proof tree of $\alpha$ w.r.t.~$D$ and $\dep$ with $\support{T} = D' \subseteq D$ is equivalent to the existence of a polynomially-sized minimal-depth proof DAG $G$ of $\alpha$ w.r.t.~$D$ and $\dep$ with $\support{G} = D'$. This in turn allows us to devise a guess-and-check algorithm that runs in polynomial time. We proceed to formalize this high-level description.

The notion of depth can be naturally transferred to rooted DAGs. In particular, for a rooted DAG $G$, the {\em depth} of $G$, denoted $\depth{G}$, is defined as the length of the longest path from the root of $G$ to a leaf node of $G$.
Given a Datalog program $\dep$, a database $D$ over $\esch{\dep}$, and a fact $\alpha \in \dep(D)$, let $\mgd{\alpha}{D}{\dep}$ be the integer
\[
\min\{\depth{G} \mid G \text{ is a proof DAG of } \alpha \text{ w.r.t. } D \text{ and } \dep\},
\]
i.e., the minimal depth over all proof DAGs of $\alpha$ w.r.t.~$D$ and $\dep$. Before introducing minimal-depth proof DAGs, let us establish a key property of $\mgd{\alpha}{D}{\dep}$, which will play a crucial role in our complexity analysis.

\begin{proposition}\label{pro:depth-ptime}
	Consider a Datalog program $\dep$, a database $D$ over $\esch{\dep}$, and a fact $\alpha \in \dep(D)$. The integer $\mgd{\alpha}{D}{\dep}$ can be computed in polynomial time in $|D|$.
\end{proposition}

\begin{proof}
	The proof relies on the well-known immediate consequence operator for Datalog. Roughly, the operator constructs in different steps all the facts that can be derived starting from $D$ and applying the rules of $\dep$. In particular, each fact is constructed ``as early as possible'', and we are going to show that the step at which a fact $\alpha$ is first obtained coincides with $\mgd{\alpha}{D}{\dep}$.
	
	A fact $R(\bar t)$ is an {\em immediate consequence of $D$ and $\dep$} if
	\begin{itemize}
		\item $R(\bar t) \in D$, or
		\item there exists a Datalog rule $R(\bar x)\,\assign\, R_1(\bar x_1),\ldots,R_n(\bar x_n)$ in $\dep$ and a function $h : \bigcup_{i \in [n]} \bar x_i \ra \ins{C}$ such that $\{R_1(h(\bar x_1)),\ldots,R_n(h(\bar x_n))\} \subseteq D$ and $h(\bar x) = \bar t$.
	\end{itemize}
	The immediate consequence operator for $\dep$ is defined as the function $T_\dep$ from the set $S$ of databases over $\sch{\dep}$ to $S$
	\begin{multline*}
	T_\dep(D)\ =\ \{R(\bar t) \mid R(\bar t) \text{ is an immediate}\\
	\text{consequence of } D \text{ and } \dep\}.
	\end{multline*}
	We then define
	\[
	T_\dep^0(D)\ =\ D, 
	\]
	and for each $i>0$, 
	\[
	T_\dep^i(D)\ =\ T_\dep( T_\dep^{i-1}(D)).
	\]
	Finally, we define
	\[
	T_\dep^\infty(D)\ =\ \bigcup_{i \ge 0} T_\dep^i(D).
	\]
	%
	It is not difficult to see that 
	\[
	T_\dep^\infty(D)\ =\ T_\dep^{|\base{D,\dep}|}(D),
	\]
	which in turn implies that $T_\dep(D)$ can be computed in polynomial time in the size of $D$; see, e.g.,~\cite{DEGV01}. Hence, for each $i \ge 0$, $T_\dep^i(D)$ can be computed in polynomial time in the size of $D$.
	Note that $T_\dep^\infty(D) = \dep(D)$~\cite{AbHV95}.
	%
	
	The above discussion, together with the following auxiliary lemmas, will prove our claim. For a fact $\alpha \in T_\dep^\infty(D)$, we write $\rank{\alpha}{D,\dep}$ for the integer $\min\{i \mid \alpha \in T_\dep^i(D)\}$.
	
	\begin{lemma}\label{lem:rank-leq}
		For every fact $\alpha \in T_\dep^\infty(D)$, $\rank{\alpha}{D,\dep} = \mgd{\alpha}{D}{\dep}$.
	\end{lemma}

	\begin{proof}
		Consider an arbitrary fact $\alpha \in T_\dep^\infty(D)$ and let $n = \mgd{\alpha}{D}{\dep}$. We proceed by induction on $n$.
		
		\medskip
		\noindent \textbf{Base Case.} For $n = 0$, the claim follows immediately by the definition of proof DAG and of $T_\dep$.
		
		\medskip
		\noindent \textbf{Inductive Step.} There is a proof DAG $G=(V,E,\lambda)$ of $\alpha$ w.r.t.~$D$ and $\dep$ with $\depth{G} = n$. Let $u_1,\ldots,u_k$ be the out-neighbours of the root node $v$, and let $\lambda(u_i) = R_i(\bar u_i)$, for each $i \in [k]$. Since $\depth{G} =n$, for each $i \in [k]$, $\mgd{R_i(\bar u_i)}{D}{\dep} \le \depth{G_i} < n$, with $G_i$ being the subDAG of $G$ rooted at $u_i$. By inductive hypothesis, $\rank{R_i(\bar u_i)}{D,\dep} \leq \mgd{R_i(\bar u_i)}{D}{\dep}$, for each $i \in [k]$, and thus, $\{R_1(\bar u_1),\ldots,R_k(\bar u_k)\} \subseteq T_\dep^{n-1}(D)$. Moreover, by the definition of proof DAG, there exist a rule $R_0(\bar x_0) \assign R_1(\bar x_1),\ldots,R_k(\bar x_k)$ in $\dep$ and a function $h : \bigcup_{i \in [n]} \bar x_i \ra \ins{C}$ such that $\alpha = R_0(h(\bar x_0))$ and $h(\bar x_i) = \bar u_i$ for each $i \in [k]$. Consequently, $\rank{\alpha}{D,\dep} \leq \max_{i \in [k]} \{\rank{R_i(\bar u_i)}{D,\dep}\} + 1 \leq n$.
		We further observe that $\alpha \not\in T_{\dep}^{n-1}(D)$ since, otherwise, by induction hypothesis we can conclude that $\mgd{\alpha}{D}{\dep} < n$, which is a contradiction. Hence, $\rank{\alpha}{D,\dep} = n$.
		%
	\end{proof}

	The claim follows by Lemma~\ref{lem:rank-leq}.
\end{proof}

The central notion of minimal-depth proof DAG follows:

\begin{definition}[\textbf{Minimal-Depth Proof DAG}]\label{def:minimal-depth-proof-dag}
	Consider a Datalog program $\dep$, a database $D$ over $\esch{\dep}$, and a fact $\alpha$ over $\sch{\dep}$. 
	A {\em minimal-depth proof DAG of $\alpha$ w.r.t.~$D$ and $\dep$} is a proof DAG $G$ of $\alpha$ w.r.t.~$D$ and $\dep$ such that $\depth{G}$ coincides with $\mgd{\alpha}{D}{\dep}$. \hfill\markfull
\end{definition}

The analogous result to Proposition~\ref{pro:characterization-nr-trees} follows:

\begin{proposition}\label{pro:characterization-md-trees}
	For a Datalog program $\dep$, there is a polynomial $f$ such that, for every database $D$ over $\esch{\dep}$, fact $\alpha$ over $\sch{\dep}$, and $D' \subseteq D$, the following are equivalent:
	\begin{enumerate}
		\item There is a minimal-depth proof tree $T$ of $\alpha$ w.r.t.~$D$ and $\dep$ such that $\support{T} = D'$.
		\item There is a minimal-depth proof DAG $G = (V,E,\lambda)$ of $\alpha$ w.r.t.~$D$ and $\dep$ with $\support{G} = D'$ and $|V| \leq f(|D|)$.
	\end{enumerate}
\end{proposition}

The direction (2) implies (1) is shown by ``unravelling'' the minimal-depth proof DAG $G$ into a minimal-depth proof tree $T$ such that $\support{G} = \support{T}$.
We use the same ``unravelling'' construction as in the proof of direction (2) implies (1) of Proposition~\ref{pro:characterization-all-trees}, which {\em preserves the minimality of the depth}.
We now proceed with the direction (1) implies (2). The underlying construction proceeds in two main steps captured by Lemmas~\ref{lem:scount-reduction-md} and~\ref{lem:from-trees-to-dags-md} given below.

\medskip 

$\bullet$ The \textbf{\textit{first step}} is to show that a minimal-depth proof tree $T$ of $\alpha$ w.r.t.~$D$ and $\dep$ with $\support{T} = D'$ can be converted into a minimal-depth proof tree $T'$ of $\alpha$ w.r.t.~$D$ and $\dep$ with $\support{T'} = D'$ that has ``small'' subtree count.

\begin{lemma}\label{lem:scount-reduction-md}
	For each Datalog program $\dep$, there is a polynomial $f$ such that, for every database $D$ over $\esch{\dep}$, fact $\alpha$ over $\sch{\dep}$, and $D' \subseteq D$, if there is a minimal-depth proof tree $T$ of $\alpha$ w.r.t.~$D$ and $\dep$ with $\support{T} = D'$, then there is also such a proof tree $T'$ with $\scount{T'} \leq f(|D|)$.
\end{lemma}

\begin{proof}
	We first argue that the proof tree $T$ has ``small'' depth. In particular, by Lemma~\ref{lem:depth-reduction}, there exists a polynomial $f$ and a proof tree $T'$ of $\alpha$ w.r.t.~$D$ and $\dep$ with $\depth{T'} \leq f(|D|)$. Since, by hypothesis, $T$ is of minimal-depth, we conclude that $\depth{T} \leq f(|D|)$.
	The other crucial observation is that the construction underlying Lemma~\ref{lem:scount-reduction}, which converts a proof tree of ``small'' depth into a proof tree of ``small`` subtree count with the same support {\em preserves the minimality of the depth}. It can be verified that the proof tree $T'$ obtained by applying on $T$ the construction underlying Lemma~\ref{lem:scount-reduction} is such that $\depth{T'} \leq \depth{T}$. Hence, since $T$ is a minimal-depth proof tree, then so is $T'$.
	Consequently, we can apply the construction of Lemma~\ref{lem:scount-reduction} to the minimal-depth proof tree $T$ and get a minimal-depth proof tree $T'$ with $\support{T} = \support{T'}$ such that $\scount{T'} \leq f(|D|)$, where $f$ is the polynomial provided by Lemma~\ref{lem:scount-reduction}.
\end{proof}

\smallskip

$\bullet$ The \textbf{\textit{second step}} shows that a minimal-depth proof tree $T$ of $\alpha$ w.r.t.~$D$ and $\dep$ with $\support{T} = D'$ of ``small'' subtree count can be converted into a compact minimal-depth proof DAG $G$ of $\alpha$ w.r.t.~$D$ and $\dep$ with $\support{G} = D'$.

\begin{lemma}\label{lem:from-trees-to-dags-md}
	For each Datalog program $\dep$ and a polynomial $f$, there is a polynomial $g$ such that, for every database $D$ over $\esch{\dep}$, fact $\alpha$, and $D' \subseteq D$, if there is a minimal-depth proof tree $T$ of $\alpha$ w.r.t.~$D$ and $\dep$ with $\support{T} = D'$ and $\scount{T} \leq f(|D|)$, then there is a minimal-depth proof DAG $G = (V,E,\lambda)$ of $\alpha$ w.r.t.~$D$ and $\dep$ with $\support{G} = D'$ and $|V| \leq g(|D|)$.
\end{lemma}

\begin{proof}
	We employ the construction underlying Lemma~\ref{lem:from-trees-to-dags}, which converts a proof tree of ``small'' subtree count into a non-recursive proof DAG of polynomial size with the same support, since it {\em preserves the minimality of the depth}. 
\end{proof}

It is now clear that the direction (1) implies (2) of Proposition~\ref{pro:characterization-md-trees} is an immediate consequence of Lemmas~\ref{lem:scount-reduction-md} and~\ref{lem:from-trees-to-dags-md}.

\medskip
\noindent
\textbf{Finalize the Proof.} We can now finalize the proof of the claim that $\mathsf{Why\text {-}Provenance_{MD}[\DAT]}$ is in \NP~in data complexity.
Fix a Datalog query $Q = (\dep,R)$. Given a database $D$ over $\esch{\dep}$, a tuple $\bar t \in \adom{D}^{\arity{R}}$, and a subset $D'$ of $D$, to decide whether $D' \in \mdwhy{\bar t}{D}{Q}$ we simply need to check for the existence of a minimal-depth proof tree $T$ of $R(\bar t)$ w.r.t.~$D$ and $\dep$ such that $\support{T} = D'$. By Proposition~\ref{pro:characterization-md-trees}, this is tantamount to the existence of a polynomially-sized minimal-depth proof DAG $G$ of $R(\bar t)$ w.r.t.~$D$ and $\dep$ with $\support{G} = D'$.
The existence of such a non-recursive proof DAG can be checked via a non-deterministic algorithm that runs in polynomial time in the size of the database as it was done for proving that $\mathsf{Why\text {-}Provenance[\DAT]}$ is in \NP~in data complexity (Theorem~\ref{the:recursive-complexity}). The only difference is that now we need to additionally check that the guessed DAG $G$ is of minimal-depth, i.e., that $\depth{G}$ coincides with $\mgd{R(\bar t)}{D}{\dep}$.
It remains to argue that the latter can be done in polynomial time. The fact that $\mgd{R(\bar t)}{D}{\dep}$ can be computed in polynomial time follows from Proposition~\ref{pro:depth-ptime}. Now, $\depth{G}$ can be computed by converting $G$ into an edge-weighted graph $G'$ by assigning weight $-1$ to each edge, and then running Dijkstra's polynomial-time algorithm for finding the smallest weight of a path between two nodes.\footnote{We recall that Dijkstra's algorithm is correct on graphs with \emph{arbitrary} integer edge labels, only when the input graph is a DAG. Indeed, for general graphs, computing the length of the longest path between two nodes is $\NP$-hard.}
Consequently, $\mathsf{Why\text {-}Provenance_{MD}}[Q]$ is in \NP, and thus, $\mathsf{Why\text {-}Provenance_{MD}[\DAT]}$ is in \NP~in data complexity.

\medskip

\noindent \underline{\textbf{Lower Bound}}
\smallskip

\noindent We proceed to establish that $\mathsf{Why\text {-}Provenance_{MD}[\LDAT]}$ is \NP-hard in data complexity. To this ends, we need to show that there exists a linear Datalog query $Q$ such that the problem $\mathsf{Why\text {-}Provenance_{MD}}[Q]$ is \NP-hard. The proof is via a reduction from $\mathsf{3SAT}$, which takes as input a Boolean formula $\varphi = C_1 \wedge \ldots \wedge C_m$ in 3CNF, where each clause has exactly 3 literals (a Boolean variable $v$ or its negation $\neg v$), and asks whether $\varphi$ is satisfiable.

\medskip
\noindent \textbf{The Linear Datalog Query.}
We start by defining the linear Datalog query $Q = (\dep,R)$. We actually adapt the query $Q$ used in the proof of the fact that $\mathsf{Why\text {-}Provenance[\LDAT]}$ is \NP-hard in data complexity (see Theorem~\ref{the:recursive-complexity}) in such a way that every proof tree has the same depth.
As usual, we use $\_$ if the name of a variable is not important, and semicolons in a tuple expression in order to separate terms with a different semantic meaning. The program $\dep$ follows:
\begin{eqnarray*}
\sigma_1 &:& R(x)\,\, \assign \,\,{\rm Var}(x;y,\_;z),{\rm Assign}(x,y,z), \\
\sigma_2 &:& R(x)\,\, \assign \,\, {\rm Var}(x;\_,y;z),{\rm Assign}(x,y,z),\\
\sigma_3 &:& {\rm Assign}(x,y,z)\,\, \assign \,\, {\rm NextC}(x;z,w;k,\ell),\\
&& \hspace{10mm}  C(x,y;\_,\_;\_,\_;z,w;k,\ell), {\rm Assign}(x,y,w),\\
\sigma_4 &:& {\rm Assign}(x,y,z) \,\, \assign \,\, {\rm NextC}(x;z,w;k,\ell),\\
&& \hspace{10mm} C(\_,\_;x,y;\_,\_;z,w;k,\ell),{\rm Assign}(x,y,w), 
\end{eqnarray*}
\begin{eqnarray*}
\sigma_5 &:& {\rm Assign}(x,y,z) \,\, \assign \,\, {\rm NextC}(x;z,w;k,\ell),\\
&& \hspace{10mm} C(\_,\_;\_,\_;x,y;z,w;k,\ell), {\rm Assign}(x,y,w), \\
\sigma' &:& {\rm Assign}(x,y,z) \,\, \assign \,\, {\rm NextC}(x;z,w;y,\_),\\
&& \hspace{45mm}{\rm Assign}(x,y,w), \\
\sigma'' &:& {\rm Assign}(x,y,z) \,\, \assign \,\, {\rm NextC}(x;z,w;\_,y),\\
&& \hspace{45mm}{\rm Assign}(x,y,w), \\
\sigma_6 &:& {\rm Assign}(x,z,w) \,\, \assign \,\, {\rm Next}(x,y;z,\_;w),R(y), \\
\sigma_7 &:& {\rm Assign}(x,z,w) \,\, \assign \,\, {\rm Next}(x,y;\_,z;w),P(y), \\
\sigma_8 &:& R(x) \,\, \assign \,\, {\rm Last}(x).
\end{eqnarray*}
The key difference compared to the Datalog program used in the proof of Theorem~\ref{the:recursive-complexity} is that now, roughly speaking, the rules $\sigma_1, \sigma_2$ also attach the id of the first clause to the variable assignment; the last position of the relation ${\rm Var}$ stores the id of the first clause. The rules $\sigma_3,\sigma_4,\sigma_5,\sigma',\sigma''$, where $\sigma_3$, $\sigma_4$ and $\sigma_5$ are adapted from the previous proof, whereas $\sigma'$, $\sigma''$ are new, force a subtree to always to perform $m$ ``steps'', when going through those rules; {\rm NextC} provides an ordering of the clauses of the formula.
It is easy to verify that $\dep$ is indeed a linear Datalog program.

\medskip
\noindent \textbf{From $\mathsf{3SAT}$ to  $\mathsf{Why\text {-}Provenance_{MD}}[Q]$.} We now establish that $\mathsf{Why\text {-}Provenance_{MD}}[Q]$ is \NP-hard by reducing from $\mathsf{3SAT}$.
Consider a 3CNF Boolean formula $\varphi = C_1 \wedge \cdots \wedge C_m$ with $n$ Boolean variables $v_1,\ldots,v_n$. For a literal $\ell$, we write $\lvar{\ell}$ for the variable occurring in $\ell$, and $\lsign{\ell}$ for the number $1$ (resp., $0$) if $\ell$ is a variable (resp., the negation of a variable).
We define $D_\varphi$ as the database over $\esch{\dep}$
\begin{eqnarray*}
&& \{{\rm Var}(v_i;0,1;1) \mid i \in [n]\}\\
&\cup& \{{\rm Next}(v_i,v_{i+1};0,1;m+1) \mid i \in [n-1]\}\\
&\cup& \{{\rm Next}(v_n,\bullet;0,1;m+1), {\rm Last}(\bullet)\}\\
&\cup& \{C(\lvar{\ell_1},\lsign{\ell_1};\lvar{\ell_2},\lsign{\ell_2};\lvar{\ell_3},\lsign{\ell_3};i,i+1;0,1) \mid \\
&& \hspace{20mm} i \in [m] \text{ with } C_i = (\ell_1 \vee \ell_2 \vee \ell_3)\}\\
&\cup& \{{\rm NextC}(v_i;j,j+1;0,1) \mid i \in [n]\ \text{ and } j \in [m]\}.
\end{eqnarray*}
We can show the next lemma, which essentially states that the above construction leads to a correct polynomial-time reduction from $\mathsf{3SAT}$ to $\mathsf{Why\text {-}Provenance}[Q]$:

\begin{lemma}\label{lem:reduction-from-3sat-md}
	The following hold:
	\begin{enumerate}
		\item $D_\varphi$ can be constructed in polynomial time in $\varphi$. 
		\item $\varphi$ is satisfiable iff $D_\varphi \in \why{(v_1)}{D_\varphi}{Q}$.
	\end{enumerate}
\end{lemma}

\begin{proof}
	It is straightforward to see that $D_\varphi$ can be constructed in polynomial time in the size of $\varphi$. We proceed to establish item (2). But let us first state and prove an auxiliary technical lemma, which essentially states that all proof trees of $R(v_1)$ w.r.t.~$D_\varphi$ and $\dep$ have the same depth. 
	
	\begin{lemma}\label{lem:same-depth}
		For every proof tree $T$ of $R(v_1)$ w.r.t.~$D_\varphi$ and $\dep$, $\depth{T} = n \cdot (m+2) + 1$.
	\end{lemma}

	\begin{proof}
		A proof tree of $R(v_1)$ w.r.t.~$D_\varphi$ and $\dep$ has a root node $v$ labeled with $R(v_1)$, and thus, its children are necessarily labeled with facts obtained by means of either rule $\sigma_1$ or $\sigma_2$. From these children, only one, say $u$, is labeled with an intensional fact, which is of the form ${\rm Assign}(v_1,\cdot,1)$ (the constant $1$ is due to the variable $z$ in the body of rules $\sigma_1,\sigma_2$). Then, the children of $u$ are necessarily labeled with facts obtained via one of $\sigma_3,\sigma_4,\sigma_5,\sigma',\sigma''$ (this is because of the constant $1$ in the fact labeling $u$). Moreover, due to the presence of the atom ${\rm NextC}(x;z,w;k,\ell)$ in the body of each such rule, only one child of $u$, say $u'$, is labeled with an intensional fact of the form ${\rm Assign}(v_1,\cdot,2)$. Applying the same reasoning, from this node with label ${\rm Assign}(v_1,\cdot,2)$, one of the rules $\sigma_3,\sigma_4,\sigma_5,\sigma',\sigma''$ will be used, and the id on the third position will be increased again, until it will coincide with $m+1$. Up to this point, the longest path in the tree from the root is of length $m$+1. It is not difficult to see now that the only rule that can be applied is either $\sigma_6$ or $\sigma_7$ due to the variable $w$ appearing in the atom ${\rm Next}(x,y;z,\_;w)$. After one of such rules is applied the only node labeled with an intensional fact has as label the fact $R(v_2)$, if $n>1$, and $R(\bullet)$ otherwise, and the total depth is $m+1+1 = m+2$. Thus, either $\sigma_1$ will then be used again, or $\sigma_8$. By applying the above reasoning for all the variables of $\varphi$, the total depth of a proof tree of $R(v_1)$ w.r.t.~$D_\varphi$ and $\dep$ is precisely $n \cdot (m+2) + 1$, as needed.
	\end{proof}

	With Lemma~\ref{lem:same-depth} in place, to prove our claim it suffices to show that $\varphi$ is satisfiable iff there is a proof tree $T$ (regardless of its depth) of $R(v_1)$ w.r.t.~$D_\varphi$ and $\dep$ such that $\support{T} = D_\varphi$. To prove this last claim, it is not difficult to see how one can adapt the proof given for Theorem~\ref{the:recursive-complexity}. The main difference is that we additionally need to argue that when a proof tree of $R(v_1)$ uses the rules $\sigma_3,\sigma_4,\sigma_5,\sigma',\sigma''$, and thus, we have a node in the proof tree labeled with a fact of the form ${\rm Assign}(v_i,\mathsf{val},j)$, if the truth value $\mathsf{val}$ chosen for $v_i$ does not make the clause $C_j$ true, then the proof tree will follow $\sigma'$ or $\sigma''$, depending on the value of $\mathsf{val}$. Hence, in any case, through rules $\sigma_3,\sigma_4,\sigma_5,\sigma',\sigma''$, any proof tree will be able to touch all facts with predicate ${\rm NextC}$ in the database. Since these are the only new facts w.r.t.~the ones used in the proof of Theorem~\ref{the:recursive-complexity} (all other facts have been only slightly modified), if $\varphi$ is satisfiable, then we have a proof tree $T$ of $R(v_1)$ w.r.t.~$D_\varphi$ and $\dep$ with the property that $\support{T} = D_\varphi$, and vice versa.
\end{proof}

By Lemma~\ref{lem:reduction-from-3sat-md}, $\mathsf{Why\text {-}Provenance_{MD}}[Q]$ is \NP-hard. Thus, $\mathsf{Why\text {-}Provenance_{MD}[\LDAT]}$ is \NP-hard in data complexity.

\subsection{Non-Recursive Queries}

We now focus on non-recursive Datalog queries, and show the following about the data complexity of why-provenance relative to minimal-depth proof trees:

\begin{theorem}\label{the:non-recursive-complexity-md}
	$\mathsf{Why\text {-}Provenance_{MD}[\NRDAT]}$ is in $\ACZ$ in data complexity.
\end{theorem}

This is shown via FO rewritability as done for Theorem~\ref{the:non-recursive-complexity}. Let us stress, however, that the target FO query cannot be obtained by simply refining the space of proof trees underlying the definition of $\cq{Q}$, for a Datalog query $Q$, as done in the case of non-recursive proof trees (see Theorem~\ref{the:non-recursive-complexity-nr}). The key reason is that the notion of minimal-depth is not over all proof trees, but over the proof trees of a certain fact. This dependency on the
fact (and indirectly, on the database) does not allow us to simply consider a refined family of proof trees as we did
for non-recursive proof-trees, and a slightly more involved construction is needed. Actually, we are going to adapt the FO query underlying Theorem~\ref{the:non-recursive-complexity}.

\medskip

\noindent \textbf{First-Order Rewriting.} Consider a non-recursive Datalog query $Q = (\dep,R)$. Recall that in the proof of Theorem~\ref{the:non-recursive-complexity}, for a CQ $\varphi(\bar y) \in \cqeq{Q}$, we defined an FO query $Q_{\varphi(\bar y)} = \psi_{\varphi(\bar y)}(x_1,\ldots,x_{\arity{R}})$, where $x_1,\ldots,x_{\arity{R}}$ are distinct variables that do not occur in any of the CQs of $\cqeq{Q}$, with the following property: for every database $D$ and tuple $\bar t \in \adom{D}^{\arity{R}}$, $\bar t \in Q_{\varphi(\bar y)}(D)$ iff $\bar t$ is an answer to $\varphi(\bar y)$ over $D$, and, in addition, {\em all} the atoms of $D$ are used in order to entail the sentence $\varphi[\bar y/\bar t]$, i.e., there are no other facts in $D$ besides the ones that have been used as witnesses for the atoms occurring in $\varphi[\bar y/\bar t]$.
We are going to extend $Q_{\varphi(\bar y)}$ into $Q^+_{\varphi(\bar y)} = \psi^+_{\varphi(\bar y)}(x_1,\ldots,x_{\arity{R}})$ that, in addition, performs the minimal depth check.
Recall that, with $\varphi$ being of the form $\exists \bar z \, (R_1(\bar w_1) \wedge \cdots \wedge R_n(\bar w_n))$, the formula $\psi_{\varphi(\bar y)}$, with free variables $\bar x = (x_1,\ldots,x_{\arity{R}})$, is of the form
\[
\exists \bar y \exists \bar z \left(\varphi_1\ \wedge\ \varphi_2\ \wedge\ \varphi_3\right),
\]
where $\varphi_1$ is defined as
\[
\bigwedge\limits_{i \in [n]}\, R_i(\bar w_i)\ \wedge\ (\bar x = \bar y)\ \wedge\ \bigwedge\limits_{\substack{u,v \in \var{\varphi}, \\ u \neq v}} \neg (u = v)
\]
$\varphi_2$ is defined as
\[
\bigwedge\limits_{P \in \{R_1,\ldots,R_n\}} \neg \left(\exists \bar u_{P} \left(P(\bar u_P)\ \wedge\ \bigwedge\limits_{\substack{i \in [n], \\ R_i = P}}\, \neg(\bar w_i = \bar u_P)\right)\right)
\]
and $\varphi_3$ is defined as
\[
\bigwedge\limits_{P \in \esch{\dep} \setminus \{R_1,\ldots,R_n\}} \neg \left(\exists \bar u_P \, P(\bar u_P)\right).
\]
Now, the formula $\psi^+_{\varphi(\bar y)}$ is
\[
\exists \bar y \exists \bar z \left(\varphi_1\ \wedge\ \varphi_2\ \wedge\ \varphi_3\ \wedge\ \varphi_4\right),
\]
where the additional conjunct $\varphi_4$ is defined as follows. For a CQ $\chi(\bar s) \in \cqeq{Q}$, we write $\depth{\chi(\bar s),Q}$ for the integer
\[
\min \{\depth{T} \mid T \text{ is a } Q\text{-tree with } \chi(\bar s) \eqtree \cq{T}\}
\]
that is, the smallest depth among all $Q$-trees whose induced CQ is isomorphic to $\chi(\bar s)$. Now, the formula $\varphi_4$ is
\begin{multline*}
	\bigwedge\limits_{\substack{\xi(\bar s) \in \cqeq{Q} \\ \text{with } \xi = \exists \bar u (P_1(\bar v_1),\ldots,P_m(\bar v_m)), \\ \depth{\xi(\bar s),Q} < \depth{\varphi(\bar y),Q}}} \neg \exists \bar s \exists \bar u \bigg(
	\bigwedge\limits_{i \in [m]}\, P_i(\bar v_i)\ \wedge\\
	(\bar x = \bar s)\ \wedge\ \bigwedge\limits_{\substack{u,v \in \var{\varphi}, \\ u \neq v}} \neg (u = v)\bigg)
\end{multline*}
which states that there is no other CQ $\xi(\bar s) \in \cqeq{Q}$ whose depth $d$ is smaller than the one of $\varphi(\bar y)$, and it witnesses the existence of a proof tree of depth $d$ of a fact $R(\bar t)$ w.r.t.~$D$ and $\dep$, where $\bar t$ is the tuple witnessed by $(x_1,\ldots,x_{\arity{R}})$.

With the FO query $Q^+_{\varphi(\bar y)}$ for each CQ $\varphi(\bar y) \in \cqeq{Q}$ in place, it should be clear that the desired FO query $Q^+_{\mi{FO}}$ is defined as $\Phi^+(x_1\ldots,x_{\arity{R}})$, where
\[
\Phi^+\ =\ \bigvee_{\varphi(\bar y) \in \cqeq{Q}} \psi^+_{\varphi(\bar y)}.
\]
We proceed to show the correctness of the construction.

\begin{lemma}\label{lem:fo-tree-equiv-md}
	Given a non-recursive Datalog query $Q=(\dep,R)$, a database $D$ over $\esch{\dep}$, $\bar t \in \adom{D}^{\arity{R}}$, and $D' \subseteq D$, $D' \in \mdwhy{\bar t}{D}{Q}$ iff $\bar t \in Q^+_{\mi{FO}}(D')$.
\end{lemma}

\begin{proof}
	We start with the $(\Rightarrow)$ direction. By hypothesis, there is a minimal-depth proof tree $T$ of $R(\bar t)$ w.r.t.~$D$ and $\dep$ such that $\support{T} = D'$.
	By Lemma~\ref{lem:fo-tree-equiv}, we get that the CQ $\varphi(\bar y) \in \cqeq{Q}$ induced by $T$ is such that $\bar t \in Q_{\varphi(\bar y)}(D')$. Thus, to prove that $\bar t \in Q^+_{\varphi(\bar y)}(D')$, it suffices to show that $\bar t \in Q'(D')$, where $Q' = \varphi_4(x_1,\ldots,x_{\arity{R}})$ with $\varphi_4$ being
	\begin{multline*}
	\bigwedge\limits_{\substack{\xi(\bar s) \in \cqeq{Q} \\ \text{with } \xi = \exists \bar u (P_1(\bar v_1),\ldots,P_m(\bar v_m)), \\ \depth{\xi(\bar s),Q} < \depth{\varphi(\bar y),Q}}} \neg \exists \bar s \exists \bar u \bigg(
	\bigwedge\limits_{i \in [m]}\, P_i(\bar v_i)\ \wedge\\
	(\bar x = \bar s)\ \wedge\ \bigwedge\limits_{\substack{u,v \in \var{\varphi}, \\ u \neq v}} \neg (u = v)\bigg).
	\end{multline*}
	Towards a contradiction, assume that $\bar t \not\in Q'(D')$. This in turn implies that there exists a CQ $\xi(\bar s) \in \cqeq{Q}$ of the form $\exists \bar u (P_1(\bar v_1),\ldots,P_m(\bar v_m))$ with $\depth{\xi(\bar s),Q} < \depth{\varphi(\bar y),Q}$ such that the sentence
	\[
	\exists \bar s \exists \bar u \bigg(
	\bigwedge\limits_{i \in [m]}\, P_i(\bar v_i)\ \wedge\\
	(\bar t = \bar s)\ \wedge\ \bigwedge\limits_{\substack{u,v \in \var{\varphi}, \\ u \neq v}} \neg (u = v)\bigg)
	\]
	is satisfied by $D'$. This implies that there is an assignment $h$ to the variables in $\xi$ with $h(\bar s) = \bar t$, $h(u) \neq h(v)$, for each $u,v \in \var{\xi}$ with $u \neq v$, and, for each $i \in [m]$, $P_i(h(\bar v_i)) \in D'$.
	Since $\xi(\bar s)$ is induced by a $Q$-tree $T'$ with depth $d=\depth{\xi(\bar s),Q}$, $T'$ is a proof tree of $R(\bar t)$ w.r.t.~$D$ and $\dep$ of depth $d$ such that 
	\[
	\support{T'} = \{P_1(h(\bar v_1)),\ldots,P_m(h(\bar v_m))\}\ \subseteq\ D'.
	\]
	However, since
	\[
	\depth{T'}\ =\ \depth{\xi(\bar s),Q}\ <\ \depth{\varphi(\bar y),Q}
	\]
	and
	\[
	\depth{\varphi(\bar y),Q}\ \leq\ \depth{T}, 
	\]
	we conclude that $\depth{T'} < \depth{T}$, which contradicts the fact that $T$ is a minimal-depth proof tree.
	
	We now proceed with direction $(\Leftarrow)$. By hypothesis, $\bar t \in Q^+_{\mi{FO}}(D')$. Therefore, there exists a CQ $\varphi(\bar y) \in \cqeq{Q}$ such that $\bar t \in Q^+_{\varphi(\bar y)}(D')$. It is clear that $\bar t \in Q_{\mi{FO}}(D')$, and form the proof of Lemma~\ref{lem:fo-tree-equiv} we get that $\varphi(\bar y)$ is induced by a $Q$-tree $T$ that is also a proof tree of $R(\bar t)$ w.r.t.~$D$ and $\dep$ with $\support{T} = D'$. We assume w.l.o.g. that $\depth{T} = \depth{\varphi(\bar y),Q}$, i.e., $T$ is the proof tree of the smallest depth among those that induce $\varphi(\bar y)$. It remains to show that $T$ is a minimal-depth proof tree. Towards a contradiction, assume that $T$ is not a minimal-depth proof tree. Thus, there exists another proof tree $T'$ for $R(\bar t)$ w.r.t.~$D$ and $\dep$ with $\depth{T'} < \depth{T}$. Let $\xi(\bar s) \in \cqeq{Q}$ be the CQ induced by $T'$. It is clear that
	\[
	\depth{\xi(\bar s),Q}\ \leq\ \depth{T'}\ <\ \depth{\varphi(\bar y),Q}.
	\]
	This allows us to conclude that $\bar t \not\in Q^+_{\varphi(\bar y)}(D')$, which is a contradiction. Consequently, $T$ is a minimal-depth proof tree of $R(\bar t)$ w.r.t.~$D$ and $\dep$. Since $\support{T} = D'$, we get that $D' \in \mdwhy{R(\bar t)}{D}{Q}$, as needed.
\end{proof}
\newcommand{\nodes}[1]{\mathsf{nodes}(#1)}
\newcommand{\edges}[1]{\mathsf{edges}(#1)}
\newcommand{\atomtotuple}[1]{\langle #1 \rangle}
\def\curnode{\mathsf{CurNode}}
\def\hedge{\mathsf{HEdge}}

\section{Unambiguous Proof Trees}\label{appsec:unambiguous-trees}
In this section, we provide proofs for all claims of Section~\ref{sec:unambiguous-trees}, and provide further details on our experimental evaluation.

\subsection{Proof of Theorem~\ref{the:complexity-unambiguous-proof-trees}}
We start by proving Theorem~\ref{the:complexity-unambiguous-proof-trees}, which we recall here for convenience:

\begin{manualtheorem}{\ref{the:complexity-unambiguous-proof-trees}}
	\theunambiguouscomplexity
\end{manualtheorem}
We prove item~(1) and item~(2) of Theorem~\ref{the:complexity-unambiguous-proof-trees} separately. We start by focusing on item~(1).

\medskip
\noindent
\underline{\textbf{Proof of Item~(1)}}
\smallskip

\noindent Our main task is to show that $\mathsf{Why\text {-}Provenance_{UN}[\DAT]}$ is in \NP. The \NP-hardness of $\mathsf{Why\text {-}Provenance_{UN}[\LDAT]}$ follows from the $\NP$-hardness of $\mathsf{Why\text {-}Provenance_{NR}[\LDAT]}$, which we have already shown in Section~\ref{appsec:refined-trees}. The latter follows from the observation that, in the case of linear Datalog programs, non-recursive proof trees and unambiguous proof trees coincide.
We now show the \NP~upper bound. 

This result relies on a characterization of the existence of an unambiguous proof tree of a fact $\alpha$ w.r.t.~a database $D$ and a Datalog program $\dep$ with $\support{T} = D' \subseteq D$ via the existence of a so-called {\em unambiguous proof DAG} $G$ of $\alpha$ w.r.t.~$D$ and $\dep$ with $\support{G} = D'$ of polynomial size. This in turn allows us to devise a guess-and-check algorithm that runs in polynomial time We proceed to formalize the above  high-level description.

For a rooted DAG $G=(V,E,\lambda)$ and a node $v \in V$, we use $G[v]$ to denote the subDAG of $G$ rooted at $v$. Moreover, two rooted DAGs $G=(V,E,\lambda)$ and $G'=(V',E',\lambda')$ are \emph{isomorphic}, denoted $G \eqtree G'$, if there is a bijection $h : V \ra V'$ such that , for each node $v \in V$, $\lambda(v) = \lambda(h(v))$, and for each two nodes $u,v \in V$,  $(u,v) \in E$ iff $(h(u),h(v)) \in E'$. With the above definitions  in place, we can now introduce the key notion of unambiguous proof DAG.

\begin{definition}[\textbf{Unambiguous Proof DAG}]\label{def:u-proof-dag}
	Consider a Datalog program $\dep$, a database $D$ over $\esch{\dep}$, and a fact $\alpha$ over $\sch{\dep}$. An \emph{unambiguous proof DAG of $\alpha$ w.r.t.\ $D$ and $\dep$} is a proof DAG $G=(V,E,\lambda)$ of $\alpha$ w.r.t.\ $D$ and $\dep$ such that, for all $v,u \in V$, $\lambda(u)=\lambda(v)$ implies $G[u] \eqtree G[v]$.\hfill\markfull
\end{definition}

We are now ready to present our characterization.

\begin{proposition}\label{pro:characterization-u-trees}
	For a Datalog program $\dep$, there is a polynomial $f$ such that, for every database $D$ over $\esch{\dep}$, fact $\alpha$ over $\sch{\dep}$, and $D' \subseteq D$, the following are equivalent:
	\begin{enumerate}
		\item There exists an unambiguous proof tree $T$ of $\alpha$ w.r.t.\ $D$ and $\dep$ such that $\support{T} = D'$.
		\item There is an unambiguous proof DAG $G=(V,E,\lambda)$ of $\alpha$ w.r.t.\ $D$ and $\dep$ with $\support{G}=D'$ and $|V| \le f(|D|)$.
	\end{enumerate}
\end{proposition}
\begin{proof}
	We first prove (1) implies (2). Let $T$ be an unambiguous proof tree of $\alpha$ w.r.t.~$D$ and $\dep$ with $\support{T} = D'$. 
	By definition, the subtree count of $T$ is ``small''; in fact, for every label $\alpha$ of $T$, $|\quot{T[\alpha]}|=1$. We then employ the construction underlying Lemma~\ref{lem:from-trees-to-dags}, which converts a proof tree of ``small'' subtree count into proof DAG of polynomial size with the same support, since it {\em preserves unambiguity}.
	
	For (2) implies (1), we employ the ``unravelling'' construction used to prove that (2) implies (1) in Proposition~\ref{pro:characterization-all-trees} since it also {\em preserves unambiguity}. 
\end{proof}

\noindent
\textbf{Finalize the Proof.}
With Proposition~\ref{pro:characterization-u-trees} in place, proving item~(1) of Theorem~\ref{the:complexity-unambiguous-proof-trees} is straightforward. Indeed, we can employ a guess-and-check algorithm similar in spirit to the one employed to prove the $\NP$ upper bound of Theorem~\ref{the:recursive-complexity}. The key difference is that here we also need to verify that the guessed DAG $G$ is unambiguous. This can be easily done by guessing, together with the graph $G$, for every pair $u,v$ of nodes of $G$ with the same label, a bijection $h_{(u,v)}$ from the nodes of $G[u]$ to the nodes of $G[v]$. The number of nodes of $G$ is polynomial w.r.t.\ $|D|$, by Proposition~\ref{pro:characterization-u-trees}, and thus the number of bijections to guess is polynomial w.r.t.\ $|D|$. With the above bijections in place, it is enough to verify that each bijection $h_{(u,v)}$ witnesses that $G[u] \eqtree G[v]$. The latter check can be easily performed in polynomial time.

\medskip
\noindent
\underline{\textbf{Proof of Item~(2)}}
\smallskip

\noindent 
This is shown via first-order rewritability as done for Theorem~\ref{the:non-recursive-complexity}. In fact, the construction of the target FO query is exactly the same as in the proof of Theorem~\ref{the:non-recursive-complexity} with the key difference that, for a Datalog query $Q$, the set of CQs $\cq{Q}$ is defined by considering only unambiguous proof trees, i.e., is the set $\{\cq{T} \mid T \text{ is a {\em unambiguous} $Q$-tree}\}$.


\subsection{Proof of Proposition~\ref{pro:why-provenance-sat}}
The goal is to prove Propostion~\ref{pro:why-provenance-sat}. But first we need to introduce some auxiliary notions and results, which will then allow us to formally define the Boolean formula $\phi_{(\bar t,D,Q)}$. We will then proceed with the proof of Propostion~\ref{pro:why-provenance-sat}

\medskip
\noindent
\textbf{A More Refined Characterization.} The Boolean formula in question relines on a more refined characterization than the one provided by Proposition~\ref{pro:characterization-u-trees}. For this, we need to define a new kind of graph that witnesses the existence of an unambiguous proof tree.

\begin{definition}[\textbf{Compressed DAG}]\label{def:compressed-dag}
Consider a Datalog program $\dep$, a database $D$ over $\esch{\dep}$, and a fact $\alpha$ over $\sch{\dep}$. A \emph{compressed DAG} of $\alpha$ w.r.t.\ $D$ and $\dep$ is a rooted DAG $G=(V,E)$, with $V \subseteq \base{D,\dep}$, such that:
\begin{enumerate}
	\item The root of $G$ is $\alpha$.
	\item If $\beta \in V$ is a leaf node, then $\beta \in D$.
	\item If $\beta \in V$ has $n \geq 1$ outgoing edges $(\beta,\gamma_1),\ldots,(\beta,\gamma_n)$, then there is a rule $R_0(\bar x_0)\ \assign\ R_1(\bar x_1),\ldots,R_m(\bar x_m) \in \dep$ and a function $h : \bigcup_{i \in [m]} \bar x_i \ra \ins{C}$ such that $\beta = R_0(h(\bar x_0))$, and $\{\gamma_i\}_{i \in [n]} = \{R_i(h(\bar x_i)) \mid i \in [m]\}$. \hfill\markfull
\end{enumerate}
\end{definition}

A compressed DAG can be seen as a proof DAG-like structure where no more than one node is labeled with the same fact.
The above definition allows us to refine the characterization given in Proposition~\ref{pro:characterization-u-trees} as follows; for a non-labeled DAG $G=(V,E)$, with a slight abuse of notation, we denote $\support{G} = \{v \in V \mid v \text{ is a leaf of } G\}$.


\begin{proposition}\label{pro:characterization-u-trees-improved}
	For a Datalog program $\dep$, database $D$ over $\esch{\dep}$, fact $\alpha$ over $\sch{\dep}$, and database $D' \subseteq D$, the following are equivalent:
	\begin{enumerate}
		\item There exists an unambiguous proof tree $T$ of $\alpha$ w.r.t.\ $D$ and $\dep$ such that $\support{T} = D'$.
		\item There exists a compressed DAG $G$ of $\alpha$ w.r.t.\ $D$ and $\dep$, such that $\support{G}=D'$.
	\end{enumerate}
\end{proposition}

\begin{proof}
	We first prove that (1) implies (2). Due to Proposition~\ref{pro:characterization-u-trees}, it suffices to show that if there exists an unambiguous proof DAG $G' = (V',E',\lambda')$ of $\alpha$ w.r.t.~$D$ and $\dep$ such that $\support{G'}=D'$, then there exists a compressed DAG $G=(V,E)$ of $\alpha$ w.r.t.\ $D$ and $\dep$ with $\support{G}=D'$.
	Assume that $G'=(V',E',\lambda')$ is an unambiguous proof DAG of $\alpha$ w.r.t.\ $D$ and $\dep$ such that $\support{G'}=D'$. Since $G'$ is unambiguous, for every two non-leaf nodes $u,v \in V'$ with $\beta = \lambda'(u)=\lambda'(v)$, we must have that $S= \{\lambda'(u_1),\ldots,\lambda'(u_n)\} = \{\lambda'(v_1),\ldots,\lambda'(v_m)\}$, where $u_1,\ldots,u_n$ and $v_1,\ldots,v_m$ are the children of $u$ and $v$ in $G'$, respectively. So, for each fact $\beta$ labeling a non-leaf node in $G'$, let us call the above (unique) set $S$ the \emph{justification of $\beta$ in $G'$}. Hence, constructing a compressed DAG $G=(V,E)$ for $\alpha$ w.r.t.\ $D$ and $\dep$, with $\support{G}=D'$ is straightforward. That is, the root of $G$ is $\alpha$, and if a node $\beta \in V$, letting $S=\{\gamma_1,\ldots,\gamma_n\}$ be its justification in $G'$, $G$ has nodes $\gamma_1,\ldots,\gamma_n$, and edges $(\beta,\gamma_1),\ldots,(\beta,\gamma_n)$.
	
	For proving (2) implies (1), we use an ``unravelling'' construction, similar to the one employed in the proof of Proposition~\ref{pro:characterization-all-trees} to convert a proof DAG to a proof tree. In particular, consider a compressed DAG $G=(V,E)$ of $\alpha$ w.r.t.~$D$ and $\dep$ with $\support{G}=D'$. By definition of $G$, for each non-leaf node $\beta$ of $G$, its children $\gamma_1,\ldots,\gamma_n$ in $G$ are such that there exists a rule $R_0(\bar x_0)\ \assign\ R_1(\bar x_1),\ldots,R_m(\bar x_m) \in \dep$ and a function $h : \bigcup_{i \in [m]} \bar x_i \ra \ins{C}$ such that $\beta = R_0(h(\bar x_0))$, and $\{\gamma_i\}_{i \in [n]} = \{R_i(h(\bar x_i)) \mid i \in [m]\}$; we call $(\sigma,h)$ the \emph{trigger of $\beta$ in $G$}, for some arbitrarily chosen pair $(\sigma,h)$ as described above.
	We unravel $G$ into an unambiguous proof tree $T=(V',E',\lambda')$ of $\alpha$ w.r.t.~$D$ and $\dep$ as follows. We add a node $v$ to $T$ with label $\lambda'(v) = \alpha$. Then, if $v$ is a node of $T$ with some label $\lambda'(v) = \beta$, letting $(\sigma,h)$ be the trigger of $\beta$ in $G$, where $\sigma = R_0(\bar x_0)\ \assign\ R_1(\bar x_1),\ldots,R_m(\bar x_m)$, we add $m$ fresh new nodes $u_1,\ldots,u_m$ to $T$, where $u_i$ has label $\lambda'(u_i)=R_i(h(\bar x_i))$, for $i \in [m]$, and we add edges $(v,u_1),\ldots,(v,u_m)$ to $T$.
	The fact that $T$ is a proof tree of $\alpha$ w.r.t.\ $D$ and $\dep$ follows by construction. To see that $T$ is unambiguous, observe that by the definition of compressed DAG, and by the construction of $T$, for every two non-leaf nodes $u,v$ of $T$ with the same label $\beta$, $u$ and $v$ have the same number of children $u_1,\ldots,u_n$, and $v_1,\ldots,v_n$, with $\lambda'(u_i)=\lambda'(v_i)$, for $i \in [n]$. Clearly, $\support{T}=D'$.
\end{proof}

With the above characterization in place, we are now ready to discuss how we construct our Boolean formula.

\medskip
\noindent
\textbf{Graph of Rule Instances and Downward Closure.} For our purposes, a \emph{(directed) hypergraph} $\mathcal{H}$ is a pair $(V,E)$, where $V$ is the set of \emph{nodes} of $\mathcal{H}$, and $E$ is the set of its \emph{hyperedges}, i.e., pairs of the form $(\alpha,T)$, where $\alpha \in V$, and $\emptyset \subsetneq T \subseteq V$. For two nodes $u,v$ of $\mathcal{H}$, we say that \emph{$u$ reaches $v$ in $\mathcal{H}$}, if either $u=v$, or there exists a sequence of hyperedges of the form $(u_1,T_1),\ldots,(u_n,T_n)$, with $u_1=u$, $v \in T_n$, and $u_i \in T_{i-1}$ for $1 < i \le n$. For $u \in V$, we write $\downof{\mathcal{H}}{u}$ for the hypergraph $(V',E')$ obtained from $\mathcal{H}$, where $V'$ contains $u$ and all nodes reachable from $u$, and the hyperedges are all the hyperedges of $\mathcal{H}$ mentioning a node of $V'$.
We can now introduce the notion of graph of rule instances.

\begin{definition}[\textbf{Graph of Rule Instances}]
Consider a Datalog program $\dep$, and a database $D$ over $\esch{\dep}$. 
The \emph{graph of rule instances (GRI) of $D$ and $\dep$} is the hypergraph $\gri{D}{\dep}=(V,E)$, with $V \subseteq \base{D,\dep}$, such that

\begin{enumerate}
	\item For each $\alpha \in D$, $\alpha \in V$.
	\item If there exists a rule $R_0(\bar x_0)\ \assign\ R_1(\bar x_1),\ldots,R_n(\bar x_n)$ in $\dep$ and a function $h : \bigcup_{i \in [n]} \bar x_i \ra \ins{C}$ such that $\alpha_i = R_i(h(\bar x_i)) \in V$, for $i \in [n]$, then $\alpha_0 = R_0(h(\bar x_0)) \in V$, and $(\alpha_0,\{\alpha_1,\ldots,\alpha_n\}) \in E$.
\end{enumerate}
\end{definition}

Roughly, $\gri{D}{\dep}$ is a structure that ``contains'' all possible compressed DAGs of $\alpha$ w.r.t.~$D$ and $\dep$.
Since we are interested in finding only compressed DAGs of a specific fact $\alpha$, we do not need to consider $\gri{D}{\dep}$ in its entirety, but we only need the sub-hypergraph of $\gri{D}{\dep}$ containing $\alpha$, and all nodes reachable from it.
Formally, for a Datalog program $\dep$, a database $D$ over $\esch{\dep}$, and a node (i.e., a fact) $\alpha$ of $\gri{D}{\dep}$, the \emph{downward closure} of $\alpha$ w.r.t.~$D$ and $\dep$ is the hypergraph $\downc{D}{\dep}{\alpha} = \downof{\gri{D}{\dep}}{\alpha}$.
In other words, the downward closure keeps from $\gri{D}{\dep}$ only the part that is relevant to derive the fact $\alpha$. It is easy to verify that $\downc{D}{\dep}{\alpha}$ ``contains'' all compressed DAGs of $\alpha$ w.r.t.~$D$ and $\dep$, and the next technical result follows:

\begin{lemma}\label{lem:dag-in-down}
	Consider a Datalog program $\dep$, a database $D$ over $\esch{\dep}$, a fact $\alpha$ over $\sch{\dep}$, and a compressed DAG $G=(V,E)$ of $\alpha$ w.r.t.\ $D$ and $\dep$. Then, for every node $\beta \in V$ with outgoing edges $(\beta,\gamma_1),\ldots,(\beta,\gamma_n)$ in $G$, we have that $(\beta,\{\gamma_1,\ldots,\gamma_n\})$ is a hyperedge of $\downc{D}{\dep}{\alpha}$.
\end{lemma}


%


\smallskip
\noindent
\textbf{The Boolean Formula.} We are now ready to introduce the desired Boolean formula. For a Datalog query $Q = (\dep,R)$, a database $D$ over $\esch{\dep}$, and a tuple $\bar t \in \adom{D}^{\arity{R}}$, we construct in polynomial time in $|D|$ the formula $\phi_{(\bar t,D,Q)}$ such that the why-provenance of $\bar t$ w.r.t.~$D$ and $Q$ relative to unambiguous proof trees can be computed from the truth assignments that make $\phi_{(\bar t,D,Q)}$ true. 

Let $\downc{D}{Q}{R(\bar t)}=(V,E)$. The set of Boolean variables of $\phi_{(\bar t,D,Q)}$ is composed of four disjoint sets $V_N$, $V_H$, $V_E$, and $V_C$ of variables. Each variable in $V_N$ corresponds to a node of $\downc{D}{Q}{R(\bar t)}$, i.e., $V_N = \{x_\alpha \mid \alpha \in V\}$, each variable in $V_H$ corresponds to a hyperedge of $\downc{D}{Q}{R(\bar t)}$, i.e., $V_H = \{y_e \mid e \in E\}$, and each variable in $V_E$ corresponds to a "standard edge" that can be extracted from a hyperedge of $\downc{D}{Q}{R(\bar t)}$, i.e., $V_E = \{z_{(\alpha,\beta)} \mid (\alpha,T) \in E \text{ with } \beta \in T \}$. The set $V_C$ will be discussed later. 
Roughly, the variables in $V_N$ and $V_E$ that will be true via a satisfying assignment of $\phi_{(\bar t,D,Q)}$ will induce the nodes and the edges of a compressed DAG $G$ for $R(\bar t)$ w.r.t.\ $D$ and $Q$, which, by Proposition~\ref{pro:characterization-u-trees-improved}, will imply that $\support{G} \in \unwhy{\bar t}{D}{Q}$.


The formula $\phi_{(\bar t,D,Q)}$ is of the form
\[
\phi_{\mi{graph}} \wedge \phi_{\mi{root}} \wedge \phi_{\mi{proof}} \wedge \phi_{\mi{acyclic}}.
\]
We proceed to discuss each of the above formulas. In the following, we use $A \Rightarrow B$ as a shorthand for $\neg A \vee B$.
The first formula $\phi_{\mi{graph}}$ is in charge of guaranteeing consistency between the truth assignments of the variables in $V_N$ and the variables in $V_E$, i.e., if an edge between two nodes is stated to be part of $G$, then the two nodes must belong to G as well:
$$ \phi_{\mi{graph}} = \bigwedge\limits_{z_{(\alpha,\beta)} \in V_E} (z_{(\alpha,\beta)} \Rightarrow x_\alpha) \wedge (z_{(\alpha,\beta)} \Rightarrow x_\beta).$$
The second formula $\phi_{\mi{root}}$ guarantees that the atom $R(\bar t)$ is indeed a node of $G$, it is the root of $G$, and no other atom that is a node of $G$ can be the root (i.e., it must always have at least one incoming edge):
\begin{multline*}
	\phi_{\mi{root}} = x_{R(\bar t)} \wedge
	\left(\bigwedge\limits_{z_{(\alpha,R(\bar t))} \in V_E} \neg z_{(\alpha,R(\bar t))} \right)\wedge \\
	\bigwedge\limits_{\substack{x_\alpha \in V_N \\ \text{ with } \alpha \neq R(\bar t)}} \left( x_\alpha \Rightarrow \bigvee\limits_{z_{(\beta,\alpha)} \in V_E} z_{(\beta,\alpha)}\right).
\end{multline*}
We now move to the next formula $\phi_{\mi{proof}}$. Roughly, this formula is in charge of ensuring that, whenever an intensional atom $\alpha$ is a node of $G$, then it must have the correct children in $G$. That is, its children are the ones coming from some hyperedge of $\downc{D}{Q}{R(\bar t)}$, and no other nodes are its children (this is needed to guarantee that $G$ is a \emph{compressed} DAG). This is achieved with two sub-formulas. The first part is in charge of choosing some hyperedge $(\alpha,T)$ of $\downc{D}{Q}{R(\bar t)}$, for each intensional atom $\alpha$, while the second guarantees that for each selected hyperedge $(\alpha,T)$ (one per intensional atom $\alpha$), \emph{all and only} the nodes in $T$ are children of $\alpha$ in $G$:
\begin{multline*}
	\phi_{\mi{proof}} = 
	\bigwedge\limits_{\substack{x_\alpha \in V_N \text{ with } \\ \alpha \text{ intensional }}} 
	\left( x_\alpha \Rightarrow \bigvee\limits_{y_{(\alpha,T)} \in V_H} y_{(\alpha,T)} \right) \wedge \\
	\bigwedge\limits_{\substack{y_{e} \in V_H \\ \text{with } e=(\alpha,T)}} \left( \bigwedge\limits_{z_{(\alpha,\beta)} \in V_E} y_{e} \Rightarrow 
	\ell_{e,\beta}
	\right),
\end{multline*}
where, for a hyperedge $e=(\alpha,T)$, $\ell_{e,\beta}$ denotes $z_{(\alpha,\beta)}$ if $\beta \in T$, and $\neg z_{(\alpha,\beta)}$ otherwise.

\medskip
\noindent
\textbf{Remark.} Although we are interested in choosing \emph{exactly one} hyperedge $(\alpha,T)$ for each intensional atom $\alpha$, the above formula uses a simple disjunction rather an exclusive or. This is fine as any truth assignment that makes two variables $y_{(\alpha,T_1)}$ and $y_{(\alpha,T_2)}$ true cannot be a satisfying assignment, since the second subformula in $\phi_{\mi{proof}}$, e.g., requires that the variables $z_{(\alpha,\beta)}$ with $\beta \in T_1$ are true, while all others must be false. Hence, since $T_1 \neq T_2$, when considering the hyperedge $(\alpha,T_2)$, this subformula will not be satisfied.

%

\medskip
\noindent The remaining formula to discuss is $\phi_{\mi{acyclic}}$. This last formula is in charge of checking that $G$, i.e., the graph whose edges correspond to the true variables in $V_E$, is acyclic. Checking acyclicity of a graph encoded via Boolean variables in a Boolean formula is a well-studied problem in the SAT literature, and thus different encodings exist. For the sake of our construction, it is enough to use the simplest (yet, not very efficient in practice) encoding, which just encodes the transitive closure of the graph. However, for our experimental evaluation, we will implement this formula using a more efficient encoding based on so-called vertex elimination, which reduces by orders of magnitude the size of the formula $\phi_{\mi{acyclic}}$~\cite{RankoohR22}.

To encode the transitive closure, we now need to employ the set of Boolean variables $V_C$, having a variable of the form $t_{(\alpha,\beta)}$, for every two nodes $\alpha,\beta$ of $\downc{D}{Q}{R(\bar t)}$. Intuitively, $t_{(\alpha,\beta)}$ denotes whether a path exists from $\alpha$ to $\beta$ in $G$.
With these variables in place, writing $\phi_{\mi{acyclic}}$ is straightforward: it just encodes the transitive closure of the underlying graph, and then checks whether no cycle exists:
\begin{multline*}
	\phi_{\mi{acyclic}} = \left(\bigwedge\limits_{z_{(\alpha,\beta)} \in V_E} z_{(\alpha,\beta)} \Rightarrow t_{(\alpha,\beta)}\right) \wedge \\
	\left(\bigwedge\limits_{\substack{z_{(\alpha,\beta)} \in V_E, t_{(\beta,\gamma)} \in V_C}} z_{(\alpha,\beta)} \wedge t_{(\beta,\gamma)} \Rightarrow t_{(\alpha,\gamma)}\right) \wedge \\
	\left( \bigwedge\limits_{t_{(\alpha,\alpha)} \in V_C} \neg t_{(\alpha,\alpha)} \right).
\end{multline*}

This completes the construction of our Boolean formula. One can easily verify that the formula is in CNF, and that can be constructed in polynomial time. 
We proceed to show its correctness, i.e., Proposition~\ref{pro:why-provenance-sat}.
For a truth assignment $\tau$ from the variables of $\phi_{(\bar t,D,Q)}$ to $\{0,1\}$, we write $\db{\tau}$ for the database $\{\alpha \in D \mid x_\alpha \in V_N \text{ and } \tau(x_\alpha)  =1\}$, i.e., the database collecting all facts having a corresponding variable in $\phi_{(\bar t,D,Q)}$ that is true w.r.t.~$\tau$. Finally, we let $\sem{\phi_{(\bar t,D,Q)}}$ as
\[
\left\{\db{\tau} \mid \tau \text{ is a satisfying assignment of } \phi_{(\bar t,D,Q)}\right\}.
\]

We are now ready to prove Proposition~\ref{pro:why-provenance-sat}, which we report here for convenience:

\begin{manualproposition}{\ref{pro:why-provenance-sat}}
	\prowhyprovenancesat
\end{manualproposition}

\begin{proof}
Due to Proposition~\ref{pro:characterization-u-trees-improved}, it suffices to show that:

\begin{lemma}\label{lem:sat-solutions}
	For every database $S \subseteq D$, the following are equivalent:
	\begin{enumerate}
		\item There exists a compressed DAG $G$ of $R(\bar t)$ w.r.t.~$D$ and $\dep$ with $\support{G}=S$.
		\item There exists a satisfying truth assignment $\tau$ of $\phi_{(\bar t,D,Q)}$ such that $\db{\tau}=S$.
	\end{enumerate}
\end{lemma}

\begin{proof}
	We start with the implication $(1) \Rightarrow (2)$. Assume that $G=(V,E)$ is a compressed DAG for $R(\bar t)$ w.r.t.~$D$ and $\dep$ with $\support{G}=S$. We construct the following truth assignment $\tau$ of $\phi_{(\bar t,D,Q)}$:
	\begin{itemize} 
		\item For each $x_\alpha \in V_N$, $\tau(x_\alpha)=1$ if $\alpha$ is a node of $G$, otherwise $\tau(x_\alpha)=0$. 
		\item For each $z_{(\alpha,\beta)} \in V_E$, $\tau(z_{(\alpha,\beta)})=1$ if there is an edge $(\alpha,\beta)$ in $G$, otherwise $\tau(z_{(\alpha,\beta)})=0$. 
		\item For each $y_{(\alpha,T)} \in V_H$, $\tau(y_{(\alpha,T)}) = 1$ if $\alpha$ is a node in $G$ with outgoing edges $(\alpha,\beta_1),\ldots,(\alpha,\beta_n)$ such that $T=\{\beta_1,\ldots,\beta_n\}$, otherwise $\tau(y_{(\alpha,T)})=0$.
		\item For each $t_{(\alpha,\beta)} \in V_C$, $\tau(t_{(\alpha,\beta)}) = 1$ if there is a path from $\alpha$ to $\beta$ in $G$, otherwise $\tau(t_{(\alpha,\beta)})=0$.
	\end{itemize}
	We now claim that $\tau$ makes $\phi_{(\bar t,D,Q)}$ true and $\db{\tau}=S$.
	
	\medskip
	\noindent
	\textbf{Observation 1.} By Lemma~\ref{lem:dag-in-down}, every node $\alpha$ of $G$ has a variable $x_\alpha$ in $\phi_{(\bar t,D,Q)}$. Similarly, if a node $\alpha$ has outgoing edges $(\alpha,\beta_1),\ldots,(\alpha,\beta_n)$ in $G$, then $z_{(\alpha,\beta_i)}$ for $i \in [n]$, and $y_{(\alpha,T)}$, with $T=\{\beta_1,\ldots,\beta_n\}$, are all variables of $\phi_{(\bar t,D,Q)}$.
	
	\medskip
	
	From Observation 1, each fact $\alpha \in \support{G}=S$ has a corresponding variable $x_\alpha$ in $\phi_{(\bar t,D,Q)}$. Moreover, by construction of $\tau$, $\tau(x_\alpha)=1$ for each $\alpha \in S$. Furthermore, for all facts $\beta$ of $D$ not in $S$ it means there is no node in $G$ labeled with $\beta$, and thus $\tau(x_\beta)=0$. Hence $\db{\tau}=S$. 
	
	We now show that $\tau$ makes $\phi_{(\bar t,D,Q)}$ true. We proceed by considering its subformulas separately.
	\begin{description}
	\item[\underline{$\phi_\mi{graph}$}.] It is clear that $\tau$ makes this formula true since when $\tau(z_{(\alpha,\beta)}) = 1$, it means that $\alpha$ and $\beta$ are nodes of $G$, and thus, by construction of $\tau$, $\tau(x_\alpha) = \tau(x_\beta)=1$.
	
	\item[\underline{$\phi_\mi{root}$}.] Since $R(\bar t)$ is the root of $G$, by Observation 1 and by construction of $\tau$, $\tau(x_{R(\bar t)}) = 1$. To see why the formula
	\[
	\left(\bigwedge\limits_{z_{(\alpha,R(\bar t))} \in V_E} \neg z_{(\alpha,R(\bar t))} \right)
	\]
	is true, since $R(\bar t)$ is the root, it does not have any incoming edges, and thus, by construction of $\tau$, $\tau(z_{(\alpha,R(\bar t))})=0$, for each $z_{(\alpha,R(\bar t))} \in V_E$. 
	Finally, regarding the formula
	$$\bigwedge\limits_{\substack{x_\alpha \in V_N \\ \text{ with } \alpha \neq R(\bar t)}} \left( x_\alpha \Rightarrow \bigvee\limits_{z_{(\beta,\alpha)} \in V_E} z_{(\beta,\alpha)}\right),$$
	if $\tau(x_\alpha)=1$, for some $\alpha \neq R(\bar t)$, by construction of $\tau$, $\alpha$ is a node of $G$. By Observation 1, for every edge $(\beta,\alpha)$ in $G$, we have that $z_{(\beta,\alpha)}$ is a variable of $\phi_{(\bar t,D,Q)}$, and $\tau(z_{(\beta,\alpha)}) = 1$, hence the disjunction is true.
	
	\item[\underline{$\phi_\mi{proof}$}.] We start by considering the formula
	\[
	\bigwedge\limits_{\substack{x_\alpha \in V_N \text{ with } \\ \alpha \text{ intensional }}} 
	\left( x_\alpha \Rightarrow \bigvee\limits_{y_{(\alpha,T)} \in V_H} y_{(\alpha,T)} \right).
	\]
	If $\tau(x_\alpha)=1$, for some intensional fact $\alpha$, it means that $\alpha$ is a node of $G$. Let $(\alpha,\beta_1),\ldots,(\alpha,\beta_n)$ be the outgoind edges of $\alpha$ in $G$. By Observation 1, $y_{(\alpha,T)}$, with $T=\{\beta_1,\ldots,\beta_n\}$, is a variable of $\phi_{(\bar t,D,Q)}$, and by construction of $\tau$, $\tau(y_{(\alpha,T)}) = 1$, hence the disjunction is true.
	We now consider the formula
	\[
	\bigwedge\limits_{\substack{y_{e} \in V_H \\ \text{with } e=(\alpha,T)}} \left( \bigwedge\limits_{z_{(\alpha,\beta)} \in V_E} y_{e} \Rightarrow 
	\ell_{e,\beta}
	\right).
	\]
	If $\tau(y_e)=1$, for $e=(\alpha,T)$, then $\alpha$ is a node of $G$ with outgoing edges $(\alpha,\beta_1),\ldots,(\alpha,\beta_n)$, where $T=\{\beta_1,\ldots,\beta_n\}$. By Observation 1, $z_{(\alpha,\beta_i)} \in V_E$, for $i \in [n]$, and by construction of $\tau$, $\tau(z_{(\alpha,\beta_i)})=1$, for $i \in [n]$. Hence, all implications of the form $y_e \Rightarrow \ell_{e,\beta}$, where $\beta \in \{\beta_1,\ldots,\beta_n\}$ are true.
	Regarding the other implications, i.e., when $\beta \not \in \{\beta_1,\ldots,\beta_n\}$, since $\alpha$ has no other outgoing edges in $G$, by construction of $\tau$, $\tau(z_{(\alpha,\beta)})=0$, for all other facts $\beta \not \in \{\beta_1,\ldots,\beta_n\}$. Hence, the whole formula is satisfied.
	
	\item[\underline{$\phi_{\mi{acyclic}}$}.] The fact that the formula is true follows from the acyclicity of $G$, Observation 1, and the construction of $\tau$.
	\end{description}

	We now proceed with $(2) \Rightarrow (1)$. By hypothesis, there is a truth assignment $\tau$ that makes $\phi_{(\bar t,D,Q)}$ true and $\db{\tau} = S$. We define the auxiliary sets
	\begin{eqnarray*}
		\nodes{\tau} &=& \{\alpha \mid x_\alpha \in V_N \text{ and } \tau(x_{\alpha}) = 1\}\\
		\edges{\tau} &=& \{(\alpha,\beta) \mid z_{(\alpha,\beta)} \in V_E \text{ and } \tau(z_{(\alpha,\beta)})=1\}. 
	\end{eqnarray*}
	
	
	\noindent Since $\tau$ makes $\phi_\mi{graph}$ true, every fact occurring in $\edges{\tau}$ also occurrs in $\nodes{\tau}$; hence, $G=(\nodes{\tau},\edges{\tau})$ is a well-defined directed graph. Moreover, since $\db{\tau}=S$, all the nodes without outgoing edges in $G$ are precisely the ones in $S$. Finally, since $\tau$ makes $\phi_{\mi{acyclic}}$ true, $G$ is acyclic, and since $\tau$ satisfies $\phi_\mi{root}$, $R(\bar t)$ is the only node of $G$ without incoming edges. Thus, $G$ is a DAG, its root is $R(\bar t)$, and its leaves is exactly the set $S$. It remains to argue that $G$ is a compressed DAG of $R(\bar t)$ w.r.t.~$D$ and $\dep$.
	
	Consider a node $\alpha$ of $G$ which is an intensional fact. It is clear that $\tau(x_\alpha)=1$, and thus, the dijunction in $\phi_\mi{proof}$ must be true. Hence, $\tau(y_{(\alpha,T)})=1$, for some hyperedge $(\alpha,T)$ of $\downc{D}{Q}{R(\bar t)}$. In particular, since $\tau$ makes
	\[
	\bigwedge\limits_{\substack{y_{e} \in V_H \\ \text{with } e=(\alpha,T)}} \left( \bigwedge\limits_{z_{(\alpha,\beta)} \in V_E} y_{e} \Rightarrow 
	\ell_{e,\beta}
	\right)\]
	true in $\phi_{\mi{proof}}$, and $\tau(y_{(\alpha,T)})=1$, we must have that $\tau$ assigns $1$ to all variables of the form $z_{(\alpha,\beta)}$ with $\beta \in T$, and $0$ to all other variables of the form $z_{(\alpha,\beta)}$ with $\beta \not \in T$. This means that there is no other variable of the form $y_{(\alpha,T')}$ with $T'=T$ that is assigned $1$ by $\tau$. Hence, 
	for each node $\alpha$ of $G$ which is intensional, $\alpha$ has outgoing edges $(\alpha,\beta_1),\ldots,(\alpha,\beta_n)$, and these are such that there exists a hyperedge of $\downc{D}{Q}{R(\bar t)}$ of the form $(\alpha,T)$, with $T=\{\beta_1,\ldots,\beta_n\}$. The latter, by definition of downward closure, implies that for each node $\alpha$ of $G$ which is intensional, $\alpha$ has outgoing edges $(\alpha,\beta_1),\ldots,(\alpha,\beta_n)$, and these are such that there exists a rule $\sigma \in \dep$ of the form $R_0(\bar x_0)\ \assign\ R_1(\bar x_1),\ldots,R_m(\bar x_m)$, with $m \ge n$, and a function $h : \bigcup_{i \in [n]} \bar x_i \rightarrow \ins{C}$, such that $R_0(h(\bar x_0))=\alpha$ and $\{R_1(h(\bar x_1)),\ldots,R_m(h(\bar x_m))\} = \{\beta_1,\ldots,\beta_n\}$. Hence, $G$ is a compressed DAG of $R(\bar t)$ w.r.t.~$D$ and $\dep$, as needed.
\end{proof}

Proposition~\ref{pro:why-provenance-sat} immediately follows from Lemma~\ref{lem:sat-solutions}.
\end{proof}

\subsection{Implementation Details}
In this section, we expand on the discussion of our implementation presented in the main body of the paper. In what follows, fix a Datalog query $Q = (\dep,R)$, a database $D$ over $\esch{\dep}$, and a tuple $\bar t \in \adom{D}^{\arity{R}}$.

\medskip
\noindent \textbf{Constructing the Downward Closure.}
Recall that the construction of $\phi_{(\bar t,D,Q)}$ relies on the downward closure of $R(\bar t)$ w.r.t.~$D$ and $\dep$. It turns out that the hyperedges of the downward closure can be computed by executing a slightly modified Datalog query $Q_{\downarrow}$ over a slightly modified database $D_{\downarrow}$. In other words, the answers to $Q_{\downarrow}$ over $D_{\downarrow}$ coincide with the hyperedges of the downward closure.
For this, we are going to adopt a slight modification of an existing approach presented in~\cite{ElKM22}. In that paper, the authors were studying the problem of computing the why-provenance of Datalog queries w.r.t.\ \emph{standard} proof trees. However, except for the construction of $\downc{D}{Q}{R(\bar t)}$, their approach to compute the supports for general trees is fundamentally different from ours, since they employ existential rule-based engines rather than SAT solvers; it is not clear how their approach could be adapted for our purposes, as we require checking whether the underlying trees are unambiguous. Moreover, the approach of~\cite{ElKM22} computes the \emph{whole set} of supports all at once, while our approach based on SAT solvers allows to \emph{enumerate} supports, and thus allows the incremental construction of the why-provenance.
Nonetheless, the construction of $\downc{D}{Q}{R(\bar t)}$ is common to both approaches, and thus we borrow the techniques of~\cite{ElKM22} for this task, which we briefly recall in the following.

The main idea is to employ an existing Datalog engine to compute the answers of a query $Q_{\downarrow}$ obtained from $Q$ over a slight modification $D_{\downarrow}$ of $D$; the answers in $Q_{\downarrow}(D_{\downarrow})$ will coincide with all the hyperedges of $\downc{D}{Q}{R(\bar t)}$. 
To this end, the rules of $Q_{\downarrow}$ contain all the rules in $\dep$, which will be in charge of deriving all nodes of $\gri{D}{\dep}$, plus an additional set of rules that will be in charge of using such nodes to construct all the hyperedges of $\downc{D}{Q}{R(\bar t)}$. Formally, let $\omega$ be the maximum arity of predicates in $\dep$, and $b$ the maximum number of atoms in the body of a rule of $\dep$. We define two new predicates: 
\begin{itemize}
	\item $\curnode$ of arity $\omega+1$ that stores the current node of $\downc{D}{Q}{R(\bar t)}$ being processed during the evaluation of $Q_{\downarrow}$.
	\item $\hedge$ of arity $(\omega+1)+b\times(\omega+1)$, which stores the hyperedges being constructed during the evaluation of $Q_{\downarrow}$.
\end{itemize}
Furthermore, for an atom $\alpha = P(\bar u)$, we denote $\atomtotuple{\alpha}$ as the tuple of length $\omega+1$ of the form $c_P,\bar u,\star,\ldots,\star$, where $c_P$ and $\star$ are constants not in $D$. Intuitively, $\atomtotuple{\alpha}$ encodes the atom $\alpha$ as a tuple of fixed length. Finally, for an atom $\alpha$ and a sequence of atoms $\beta_1,\ldots,\beta_n$, with $n \le b$, we use $\tup{\alpha,\beta_1,\ldots,\beta_n}$ to denote the tuple of length $(\omega+1)+b \times (\omega+1)$ of the form $\tup{\alpha},\tup{\beta_1},\ldots,\tup{\beta_n},\star,\ldots,\star$.

We define the Datalog query $Q_{\downarrow} = (\dep',\hedge)$, where $\dep' = \dep \cup \dep''$, where, for each rule $\sigma \in \dep$ of the form $R_0(\bar x_0)\ \assign\ R_1(\bar x_1),\ldots,R_n(\bar x_n)$, $\dep''$ contains the rules
\begin{multline*}
\sigma_1\ =\ \hedge(\atomtotuple{R_0(\bar x_0),R_1(\bar x_1),\ldots,R_n(\bar x_n)})\ \assign\ \\ \curnode(\atomtotuple{R_0(\bar x_0)}), R_1(\bar x_1),\ldots,R_n(\bar x_n)
\end{multline*}
and, for each $i \in [n]$,
\begin{multline*}
\sigma^{(i)}_2\ =\ \curnode(\atomtotuple{R_i(\bar x_i)})\ \assign\ \\ \curnode(\atomtotuple{R_0(\bar x_0)}),R_1(\bar x_1),\ldots,R_n(\bar x_n).
\end{multline*}
Essentially, the rule $\sigma_1$ will construct a hyperedge $(\alpha,T)$ of $\downc{D}{Q}{R(\bar t)}$ whenever it is known that $\alpha$ is a node of $\downc{D}{Q}{R(\bar t)}$, and all the atoms in $T$ are used in $\gri{D}{\dep}$ to generate $\alpha$.
The rules of the form $\sigma^{(i)}_2$ are marking the new nodes as being part of $\downc{D}{Q}{\bar t}$.
%
Note that the rules of $\dep''$ contain constants, whereas, according to our definition, rules are constant-free. Nevertheless, all existing Datalog engines support rules with constants, and for the sake of keeping the discussion simple, we slightly abuse our definition of rules in this section. It is easy to adapt the above set of rules to a set of rules without constants, by adding some auxiliary facts to the database.


The database $D_{\downarrow}$ is $D \cup \{\curnode(\tup{R(\bar t)})\}$, which simply states that $R(\bar t)$ must be a node of $\downc{D}{Q}{R(\bar t)}$.
One can easily see that each tuple in $Q_{\downarrow}(D_{\downarrow})$ encodes a hyperedge of $\downc{D}{Q}{R(\bar t)}$, and thus, we can construct $\downc{D}{Q}{R(\bar t)}$ by simply asking $Q_{\downarrow}$ over $D_{\downarrow}$.

Let us note that the main differences between our definition of $Q_{\downarrow}$ and $D_{\downarrow}$ w.r.t.\ the ones of~\cite{ElKM22} is that we encode nodes and hyperedges of $\downc{D}{Q}{R(\bar t)}$ as tuples, which allows us to employ the same Datalog engine that is used to answer the original query $Q$, rather than using external engines supporting more expressive languages such as existential rules.

\medskip
\noindent \textbf{Constructing the Formula.} 
Regarding the construction of the formula $\phi_{(\bar t,D,Q)}$, as already discussed before, for efficiency reasons, we consider a different encoding of the subformula $\phi_\mi{acyclic}$. Rather than using the transitive closure, we employ the technique of vertex elimination~\cite{RankoohR22}. The advantage of this approach is that it requires a number of Boolean variables for the encoding of $\phi_\mi{acyclic}$ which is of the order of $O(n \times \delta)$, where $n$ is the number of nodes of the underlying graph, and $\delta$ is the so-called \emph{elimination width} of the graph, which, roughly, is related to how connected the underlying graph is. Hence, we can avoid the costly construction of quadratically many variables whenever the elimination width is low.

\medskip
\noindent \textbf{Incrementally Constructing the Why-Provenance.}
Recall that we are interested in the incremental computation of the why-provenance, which is more useful in practice than computing the whole set at once. To this end, we need a way to enumerate all the members of the why-provenance without repetitions. This is achieved by adapting a standard technique from the SAT literature for enumerating the satisfying assignments of a Boolean formula, called {\em blocking clause}.
We initially collect in a set $S$ all the facts of $D$ occurring in the downward closure of $R(\bar t)$ w.r.t.~$D$ and $\dep$. Then, after asking the SAT solver for an arbitrary satisfying assignment $\tau$ of $\phi_{(\bar t,D,Q)}$, we output the database $\db{\tau}$, and then construct the ``blocking'' clause
$
\vee_{\alpha \in S} \ell_\alpha,
$
where $\ell_\alpha = \neg x_\alpha$ if $\alpha \in \db{\tau}$, and $\ell_\alpha = x_\alpha$ otherwise. We then add this clause to the formula, which expresses that no other satisfying assignment $\tau'$ should give rise to the same member of the why-provenance.
This will exclude the previously computed explanations from the computation. We keep adding such blocking clauses each time we get a new member of the why-provenance until the formula is unsatisfiable.

{\footnotesize 
	\begingroup
	\setlength{\tabcolsep}{5pt} 
	\renewcommand{\arraystretch}{1.6} 
	\begin{figure*}[t]
		\centering
		\begin{tabular}{cc}
			\multicolumn{2}{c}{\includegraphics[width=135mm]{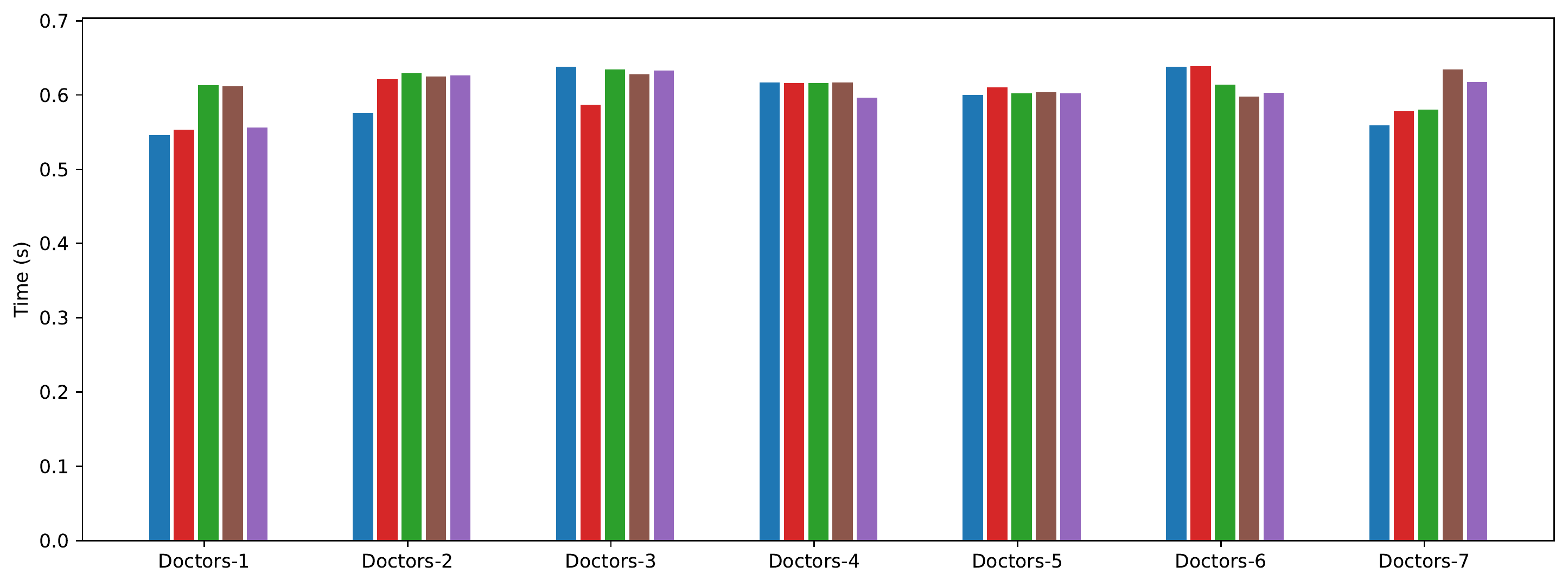}} \\
			\multicolumn{2}{c}{(a) $\mathsf{Doctors}$} \\[6pt]
			\includegraphics[width=65mm, height=53.2mm]{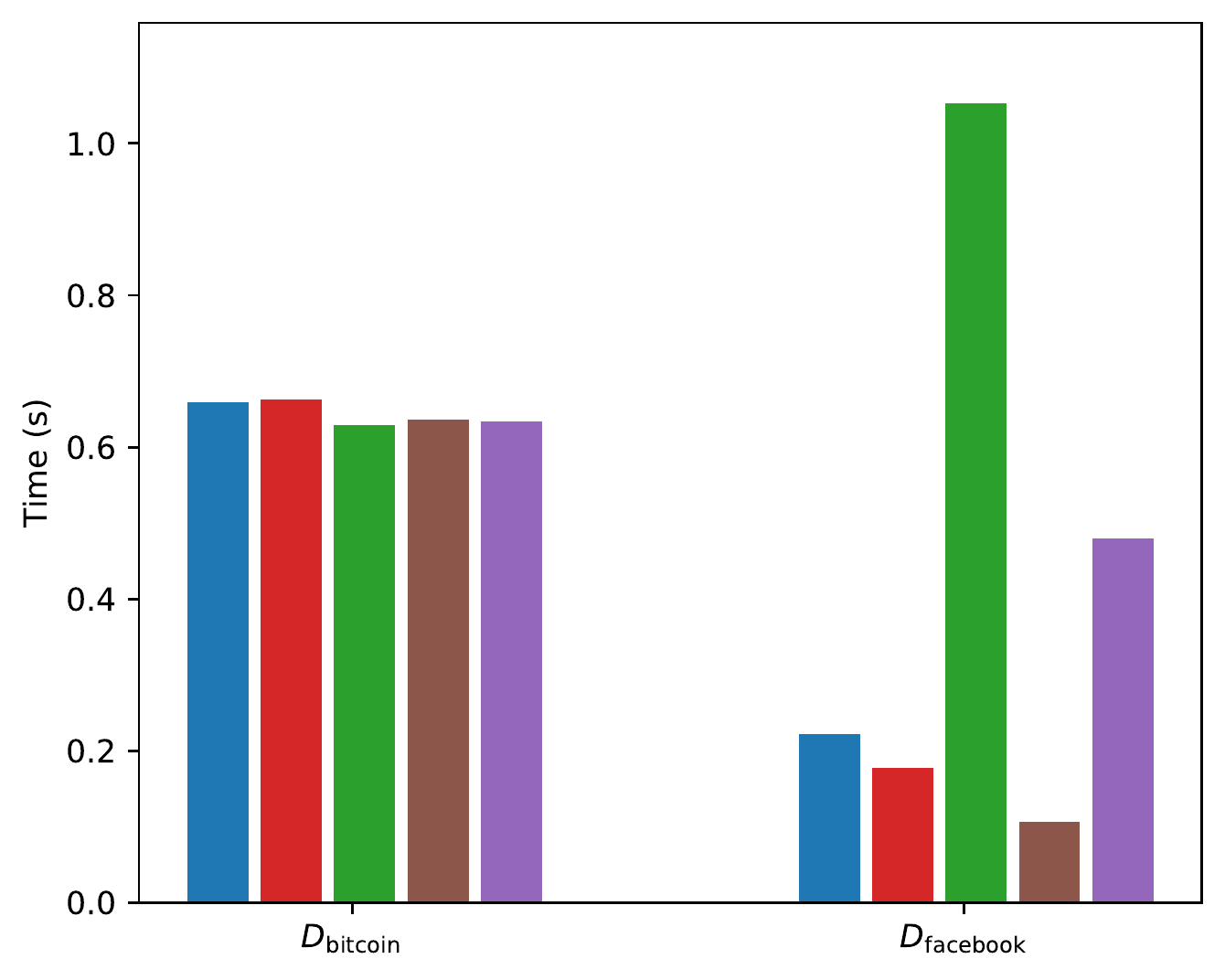} &   \includegraphics[width=65mm]{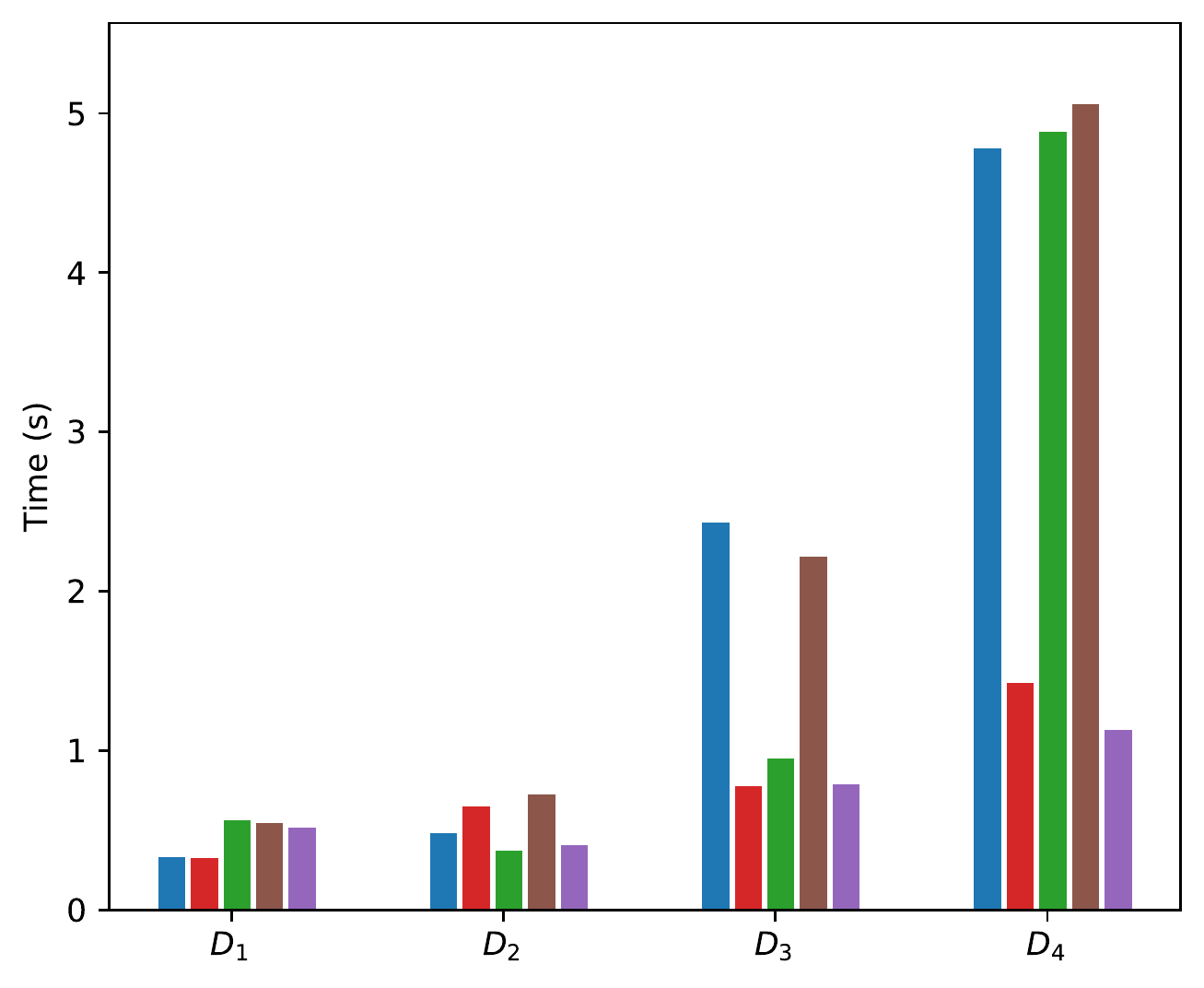} \\
			(b) $\mathsf{TransClosure}$ & (c) $\mathsf{Galen}$ \\[6pt]
			\includegraphics[width=65mm]{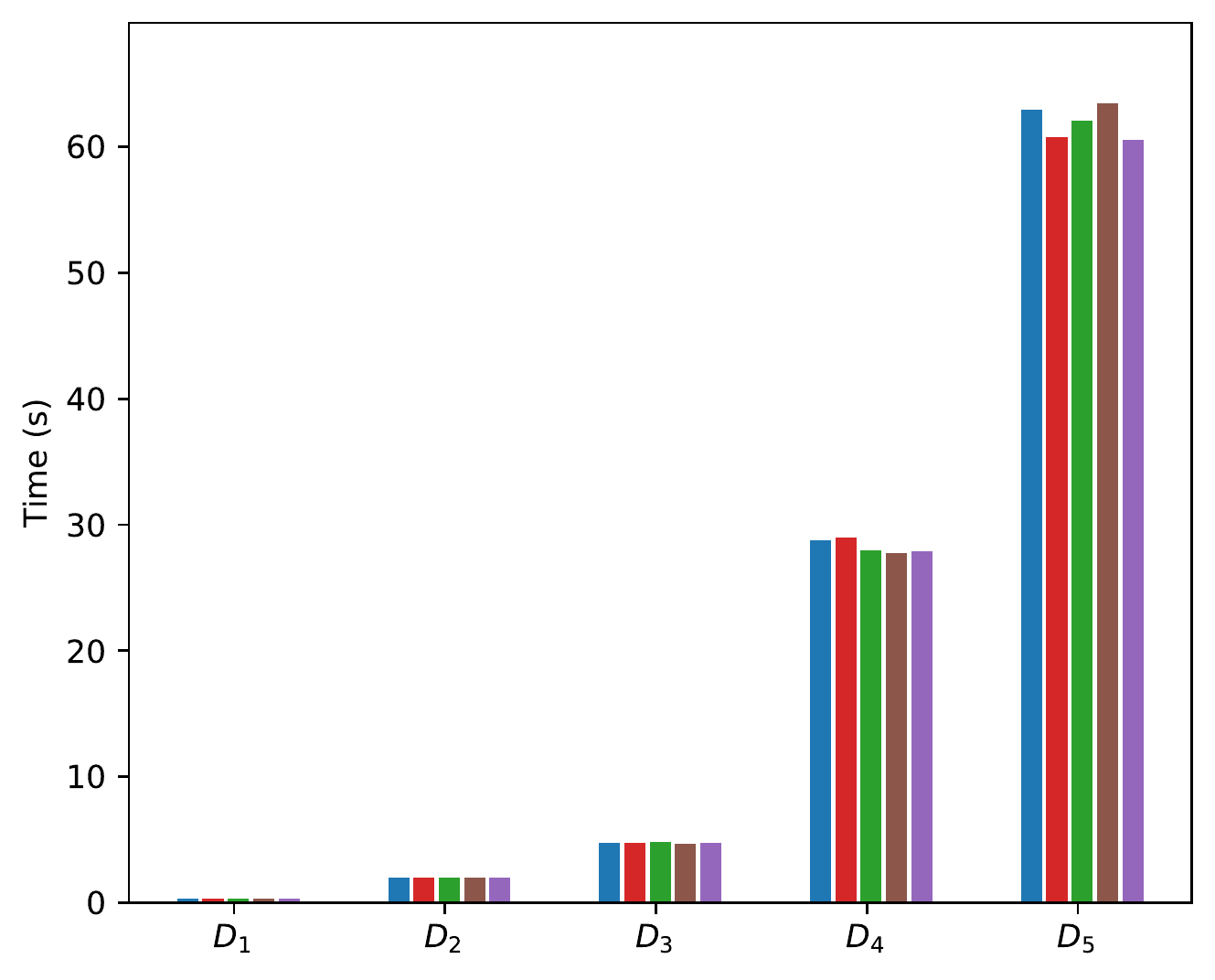} &   \includegraphics[width=65mm]{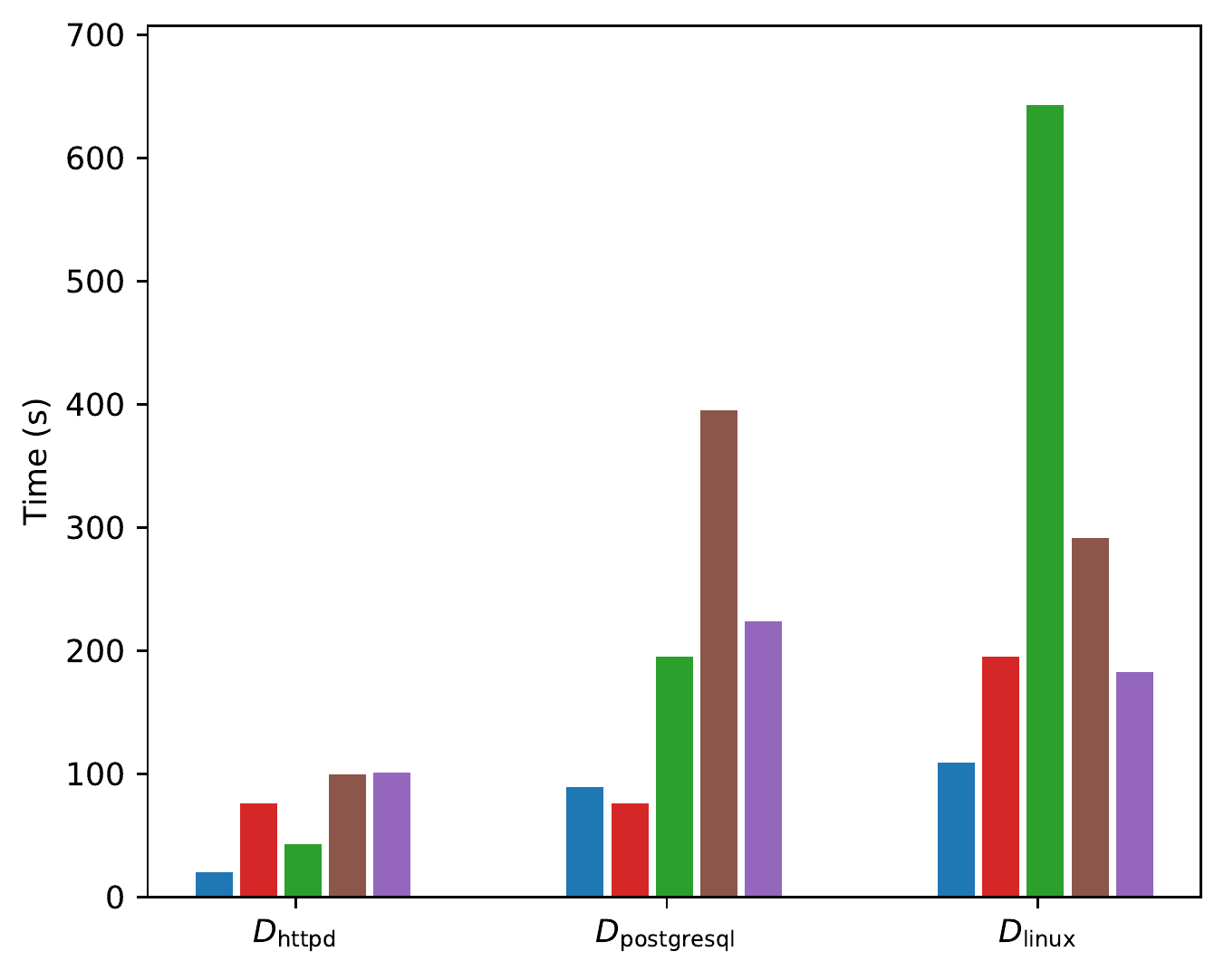} \\
			(d) $\mathsf{Andersen}$ & (e) $\mathsf{CSDA}$
		\end{tabular}
		\caption{Building the downward closure and the Boolean formula (all scenarios).}
		\label{fig:all-task1}
	\end{figure*}
	\endgroup
}

\subsection{Experimental Evaluation}
In this section, we provide further details on the performance of our SAT-based approach, by presenting the results of the experimental evaluation over all the scenarios we considered in the paper; we report again the $\mathsf{Andersen}$ scenario for the sake of completeness. Recall that in our experimental analysis we consider two main tasks separately: (1) construct the downward closure and the Boolean formula, and (2) incrementally compute the why-provenance using the SAT solver.

Concerning task 1, we report in Figure~\ref{fig:all-task1} one plot for each scenario we consider, where in each plot we report the total running time for each database of that scenario. Furthermore, for each plot, and each database considered therein, we have five bars, that correspond to the five randomly chosen tuples. Each such bar shows the time for building the downward closure plus the time for constructing the Boolean formula. To ease the presentation, we grouped the $\mathsf{Doctors}$-based scenarios in one plot (recall that all such scenarios share a single database).

We can see that in most of the scenarios, the running time is in the order of some seconds. This is especially true for the $\mathsf{TransClosure}$ and $\mathsf{Doctors}$-based scenarios, having the simplest queries, while for the $\mathsf{Galen}$ scenario, where the query is more complex, as it involves non-linear recursion, the time is slightly higher for the largest database $D_{4}$. The $\mathsf{Andersen}$ and $\mathsf{CSDA}$ scenarios are the most challenging, since they both contain very large databases. Moreover, although the databases in $\mathsf{Andersen}$ are smaller than those of $\mathsf{CSDA}$, the complexity of its query, which involves non-linear recursion, makes the running time of $\mathsf{Andersen}$ over its largest database (6.8M facts) comparable to $\mathsf{CSDA}$ with the larger database $D_\mathsf{httpd}$ (10M facts). Of course, for the much larger databases $D_\mathsf{postgresql}$ and $D_\mathsf{linux}$, the running time is much higher, going up to 6-7 minutes for some tuples.
Considering the size of the databases at hand, we believe the running times for these last scenarios are quite reasonable.
As already discussed in the main body of the paper, we observed that most of the time is spent in building the downward closure.

Concerning task 2, that is, the incremental computation of the why-provenance, we present in Figure~\ref{fig:all-task2} one plot for each scenario we consider, where in each plot of scenario $s$ we report, for each database of $s$, the times required to build an explanation, that is, the time between the current member of the why-provenance and the next one (this time is also known as the delay).
Each plot collects the delays of constructing the members of the why-provenance (up to a limit of 10K members or 5 minutes timeout) for each of the five randomly chosen tuples. We use box plots, where the bottom and the top borders of the box represent the first and third quartile, i.e., the delay under which 25\% and 75\% of all delays occur, respectively, and the orange line represents the median delay. Moreover, the bottom and the top whisker represent the minimum and maximum delay, respectively. All times are expressed in milliseconds and we use logarithmic scale. 
As we did for task 1, we grouped the $\mathsf{Doctors}$-based scenarios in one plot.

{\footnotesize 
	\begingroup
	\setlength{\tabcolsep}{5pt} 
	\renewcommand{\arraystretch}{1.6} 
	\begin{figure*}[t]
		\centering
		\begin{tabular}{cc}
			\multicolumn{2}{c}{\includegraphics[width=135mm]{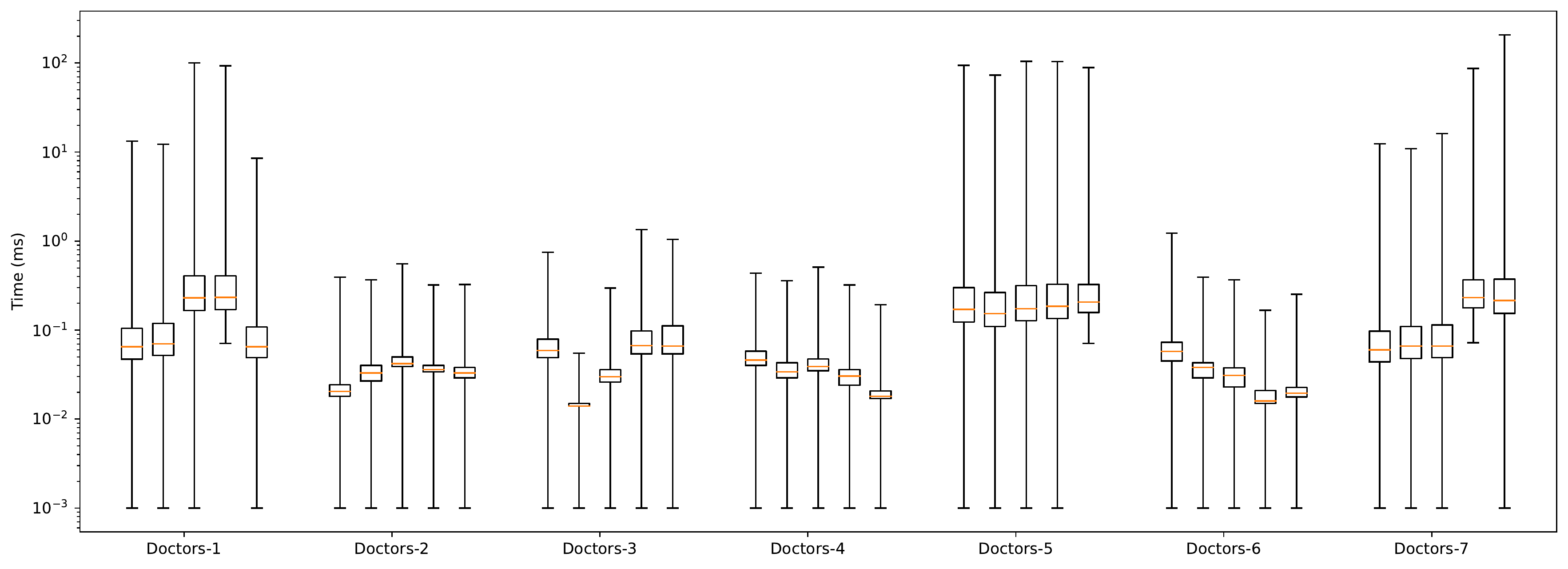}} \\
			\multicolumn{2}{c}{(a) $\mathsf{Doctors}$} \\[6pt]
			\includegraphics[width=65mm, height=53.2mm]{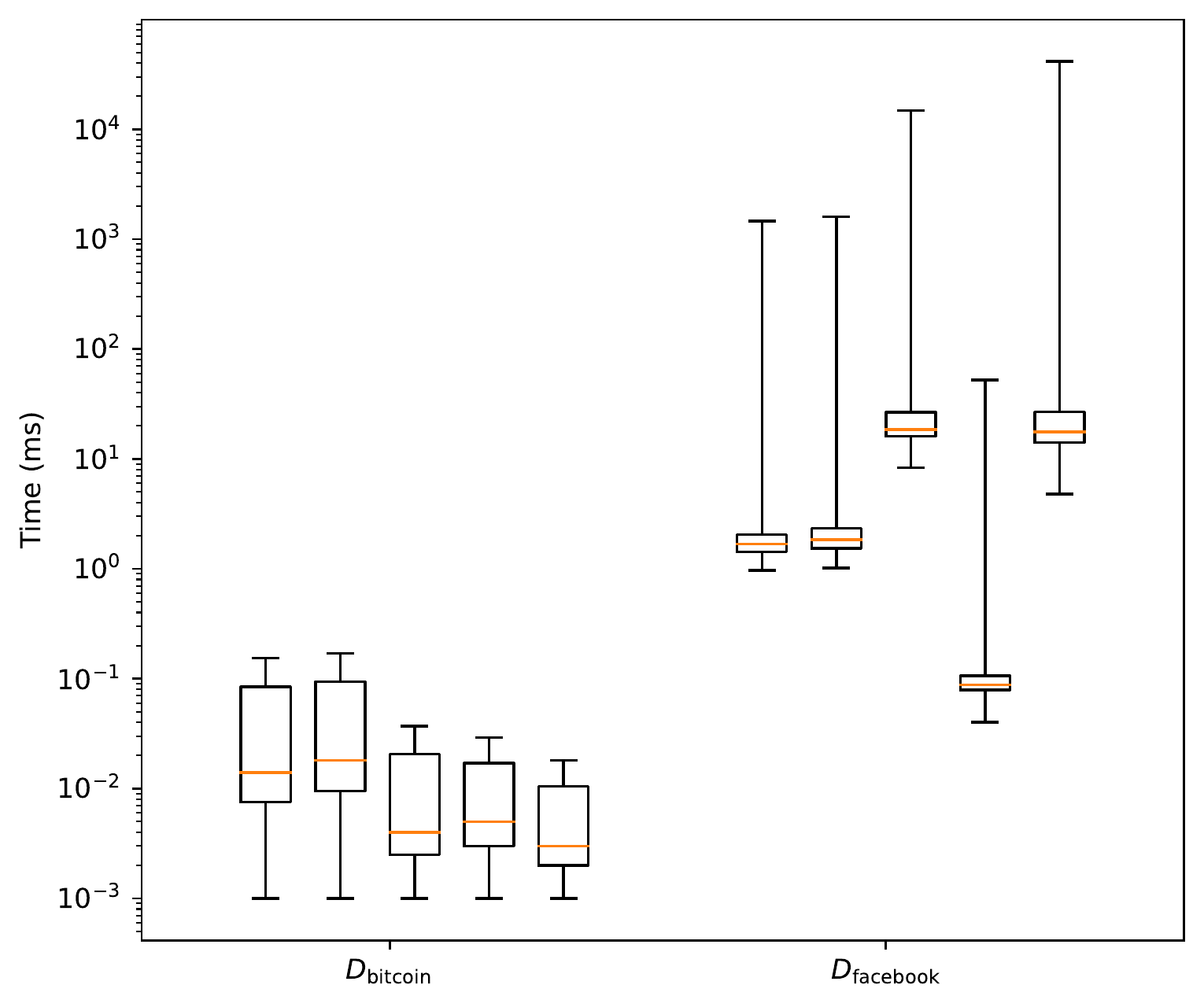} &   \includegraphics[width=65mm]{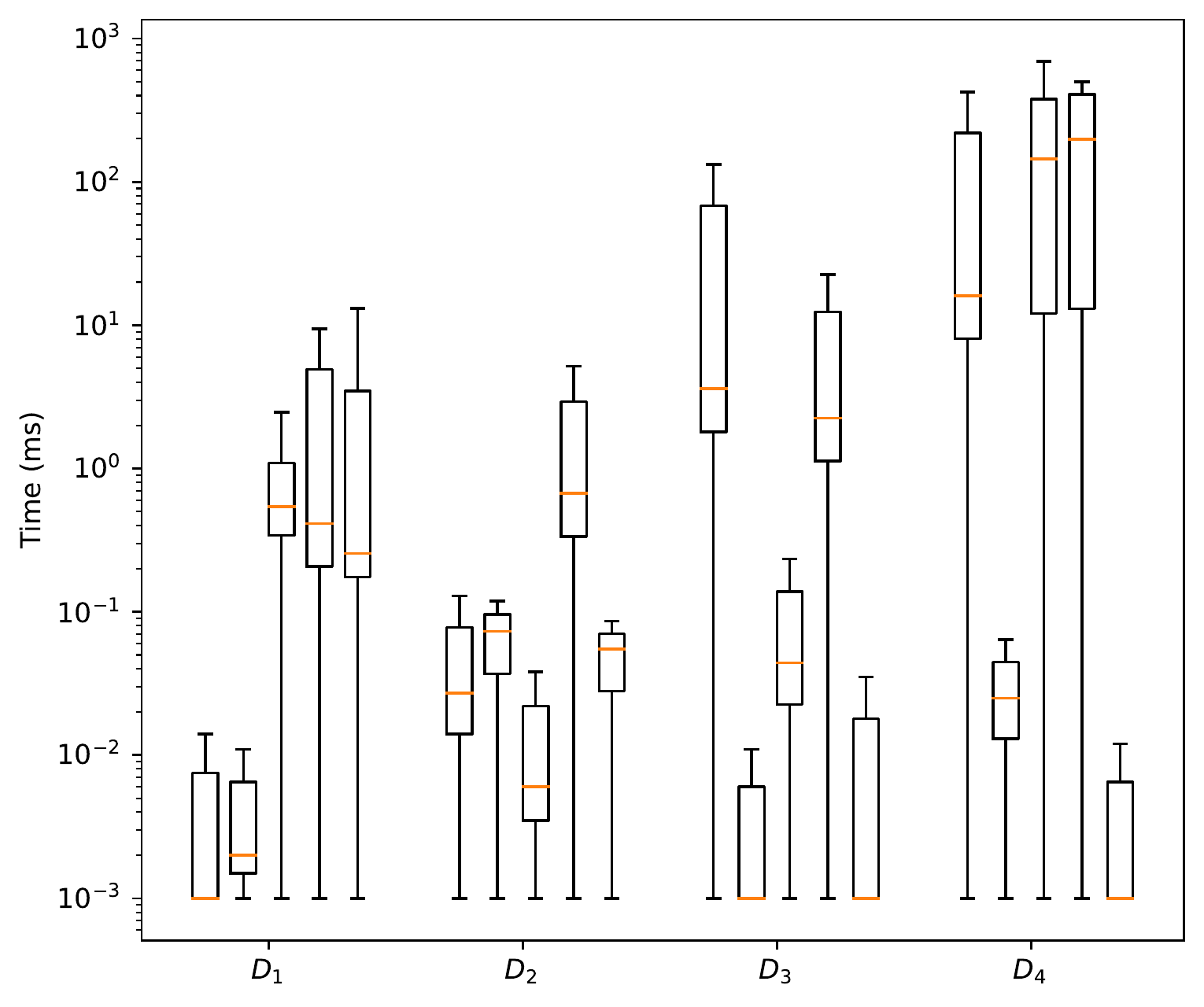} \\
			(b) $\mathsf{TransClosure}$ & (c) $\mathsf{Galen}$ \\[6pt]
			\includegraphics[width=65mm]{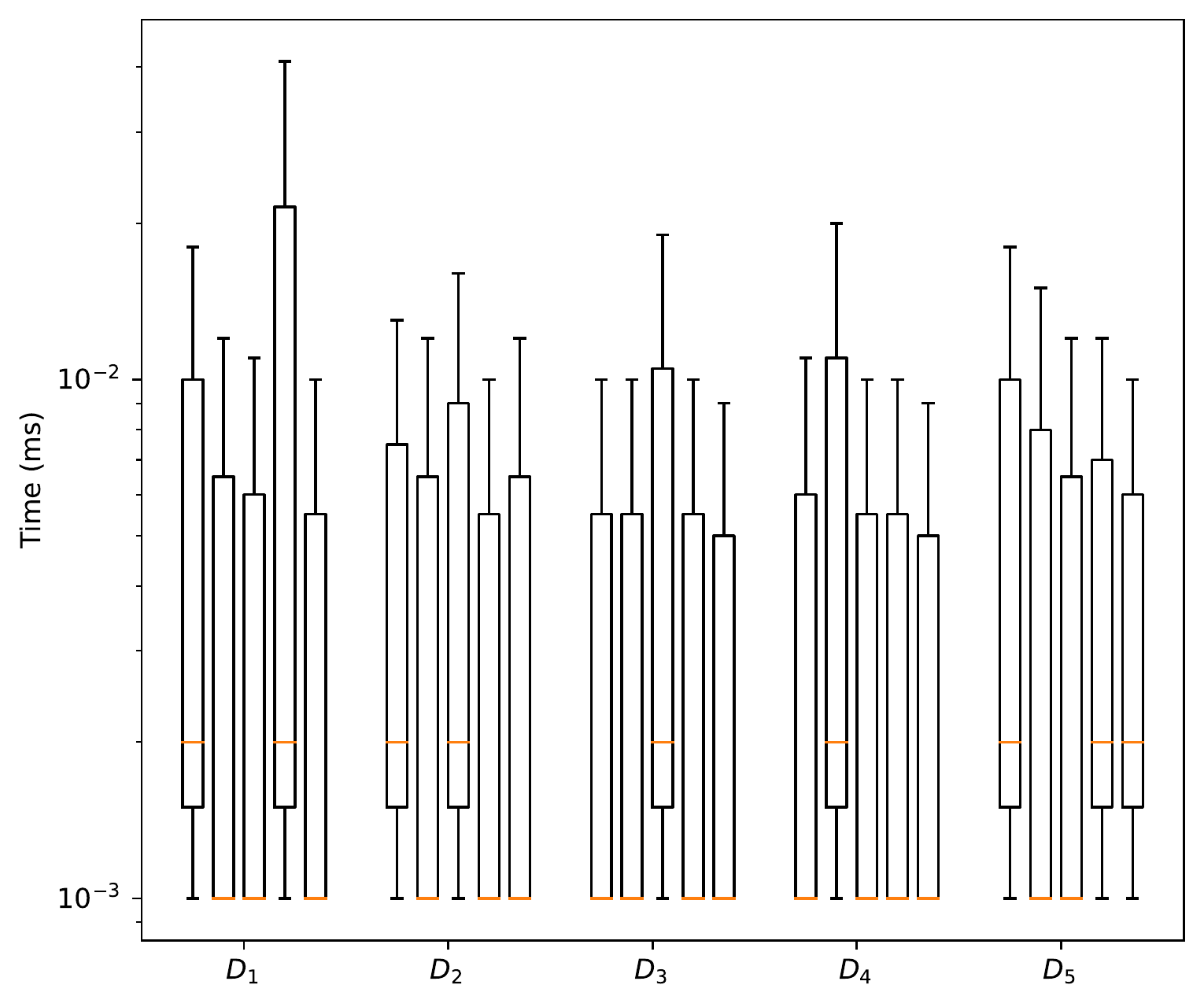} &   \includegraphics[width=65mm]{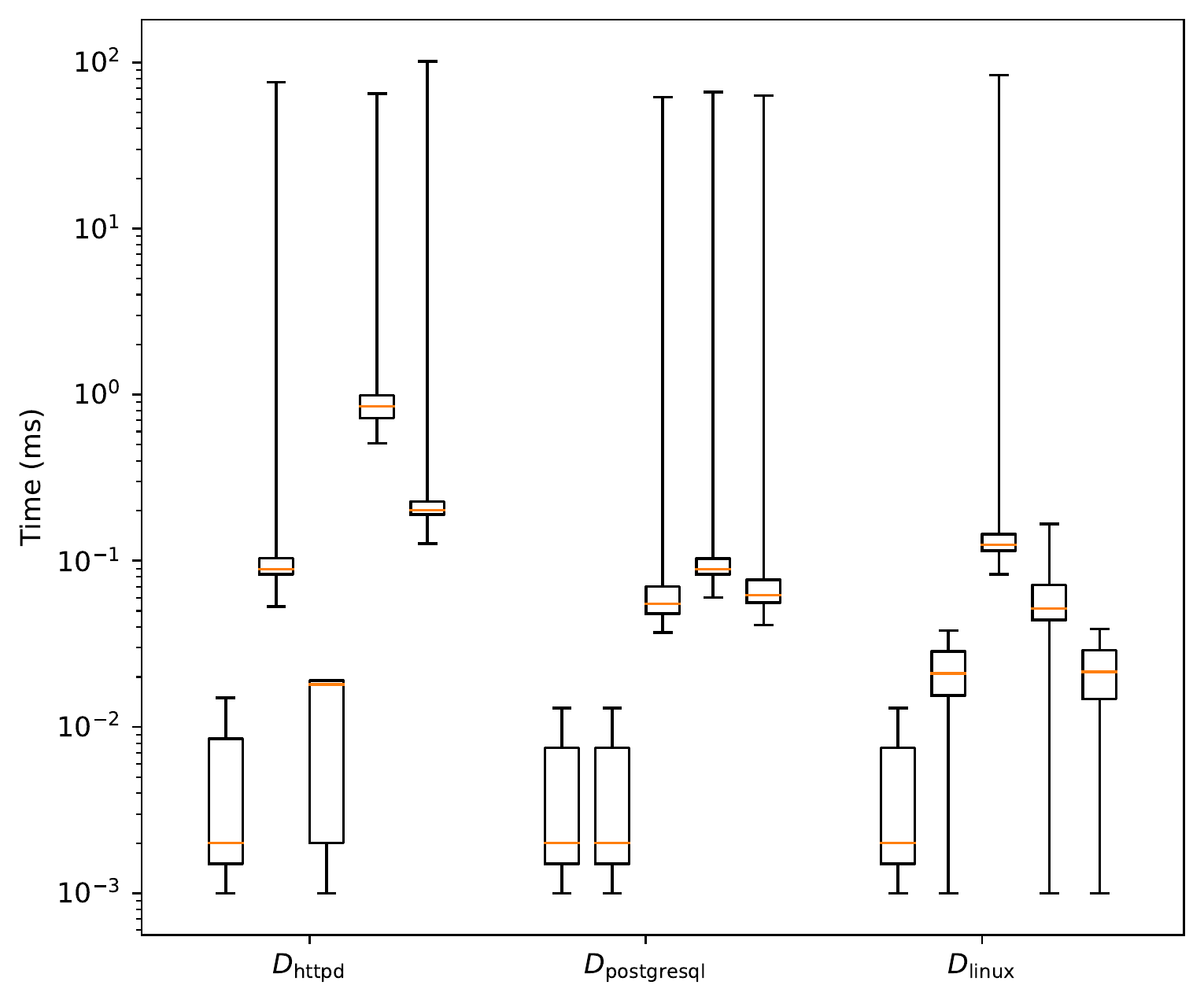} \\
			(d) $\mathsf{Andersen}$ & (e) $\mathsf{CSDA}$
		\end{tabular}
		\caption{Incremental computation of the why-provenance (all scenarios).}
		\label{fig:all-task2}
	\end{figure*}
	\endgroup
}

As observed in the main body of the paper, most of the delays are even lower than 1 millisecond, with the median in the order of microseconds. Therefore, once we have the Boolean formula in place, incrementally computing the members of the why-provenance is extremely fast. The worst case occurs in the $\mathsf{TransClosure}$ scenario, when considering the Facebook database, where we also have the only two cases where the construction of the supports exceeds the 5 minutes mark before being able to construct 10K supports, i.e., for the third and fifth tuple. Overall, in the $\mathsf{TransClosure}$ scenario w.r.t.~the Facebook database, the average delay is higher, and some supports require up to 10 seconds to be constructed. We believe that this has to do with the fact that the database $D_\text{facebook}$ encodes a graph which is highly connected, and thus the CNF formula, and in particular the formula $\phi_{\mi{acyclic}}$ encoding the acyclicity check, becomes quite large, and thus is much more demanding for the SAT solver. This was somehow expected since the vertex-elimination technique for checking acyclicity performs better the less connected the underlying graph is. We have confirmed this by running some other experiments with databases taken from~\cite{FanMK22}, which contain highly connected, synthetic graphs. In this case, although constructing the downward closure is very efficient (in the order of seconds), the construction of the formula $\phi_\mi{acyclic}$ goes out of memory. Hence, we expect that in applications where highly connected input graphs are common, a different approach for checking acyclicity in a CNF formula would be required. Nonetheless, we can safely conclude that in most cases, computing the members of the why-provenance can be done very efficiently.

\subsection{Comparative Evaluation}
As mentioned in Section~\ref{sec:conclusions}, we have performed a preliminary comparison with~\cite{ElKM22}. We conclude this section by discussing  the details of this comparative evaluation.
Let us first clarify that our implementation deals with a different problem. For a Datalog query $Q = (\dep,R)$, a database $D$ over $\esch{\dep}$, and a tuple $\bar t \in \adom{D}^{\arity{R}}$, the approach from~\cite{ElKM22} has been designed and evaluated for building the whole set $\why{\bar t}{D}{Q}$, whereas our approach has been designed and evaluated for incrementally computing $\unwhy{\bar t}{D}{Q}$. However, there is a setting where a reasonable comparison can be performed, which will provide some insights for the two approaches.
This is when the Datalog query $Q$ is both linear and non-recursive in which case the sets $\why{\bar t}{D}{Q}$ and $\unwhy{\bar t}{D}{Q}$ coincide since a proof tree of $R(\bar t)$ w.r.t.~$D$ and $\dep$ is trivially unambiguous.
Therefore, towards a fair comparison, we are going to consider the scenarios $\mathsf{Doctors}\text{-}i$, for $i \in [7]$, which consist of a Datalog query that is linear and non-recursive, and consider the end-to-end runtime of our approach (not the delays) without, of course, setting a limit on the number of members of why-provenance to build, or on the total runtime.

The comparison is shown in Figure~\ref{fig:comparison-time}. For each scenario, we present the runtime for all five randomly chosen tuples for our approach (in blue) and the approach of~\cite{ElKM22} (in red); if a bar is missing for a certain tuple, then the execution ran out of memory.
We observe that for the simple scenarios the two approaches are comparable in the order of a second.
Now, concerning the demanding scenarios, i.e., $\mathsf{Doctors}\text{-}i$ for $i \in \{1,5,7\}$, we observe that our approach is, in general, faster. Observe also that for some of the most demanding cases, the approach of~\cite{ElKM22} runs out of memory.
We believe that the latter is due to the use of the rule engine VLog, which is intended for materialization-based reasoning with existential rules, whereas our approach relies on a Datalog engine (in particular, DLV), and thus, exploiting all the optimizations that are typically employed for evaluating a Datalog query. For example, the technique of {\em magic-set rewriting}, implemented by DLV, can greatly reduce the memory usage by building much fewer facts during the evaluation of the rules; see, e.g.,~\cite{LAAC+19}.

{\footnotesize 
	\begingroup
	\setlength{\tabcolsep}{5pt} 
	\renewcommand{\arraystretch}{1.6} 
\begin{figure*}[t]
	\centering
	\includegraphics[width=135mm]{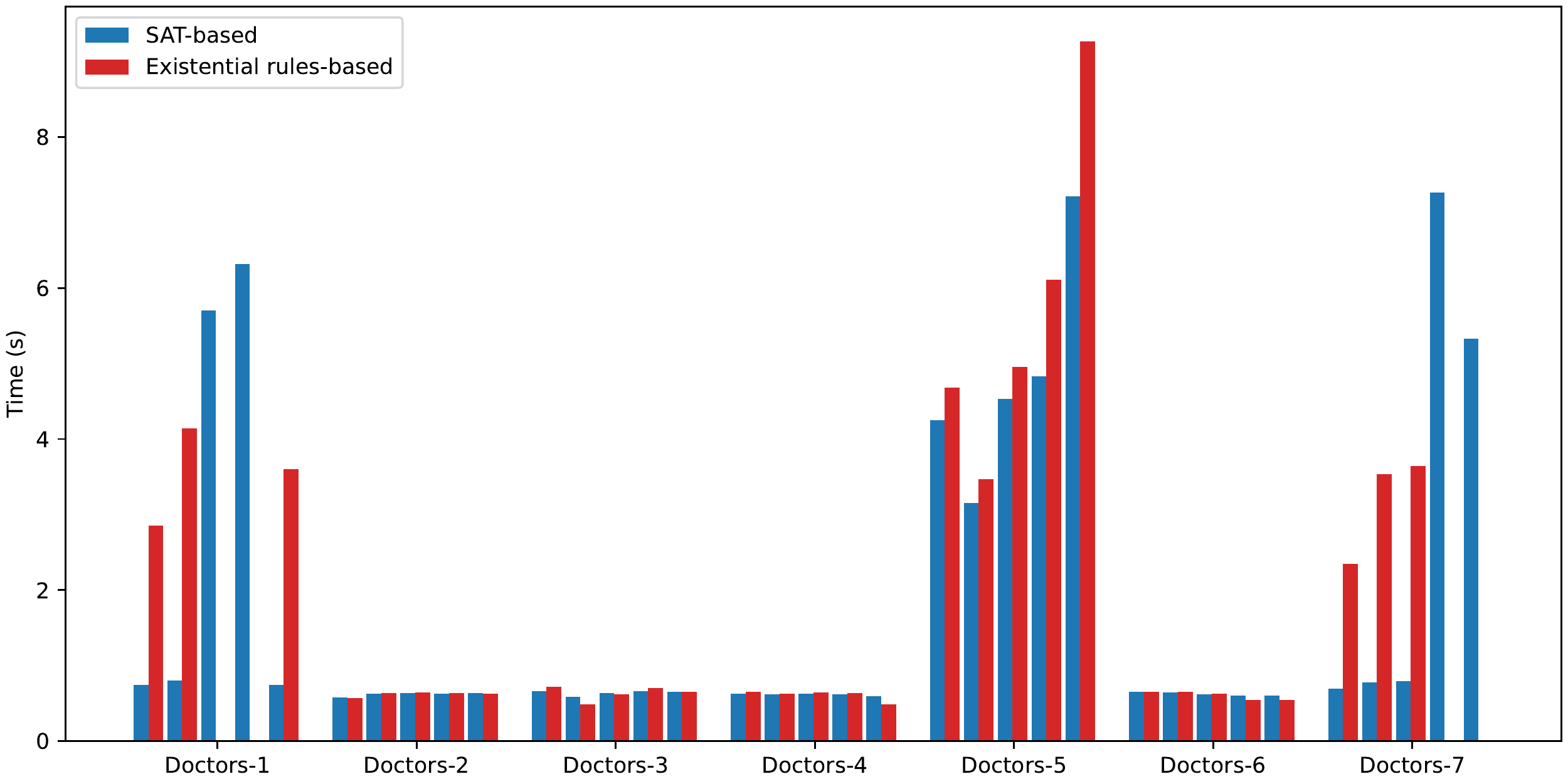}
	\caption{Comparison of the end-to-end computation of the why-provenance.}
	\label{fig:comparison-time}
\end{figure*}
\endgroup
}

\end{document}